\documentclass[a4paper,11pt,leqno]{article}

\makeatletter
\usepackage{hyperref}
\usepackage{amsfonts,dsfont,delarray,amssymb,amsmath,amsthm,a4,a4wide,color}
\usepackage[ansinew]{inputenc}

\newtheorem{theo}{Theorem}
\newtheorem{pro}{Proposition}[section]
\newtheorem{lem}[pro]{Lemma}

\newtheorem{remark}[pro]{Remark}
\newtheorem{defi}[pro]{Definition}

\def\Xint#1{\mathchoice
   {\XXint\displaystyle\textstyle{#1}}%
   {\XXint\textstyle\scriptstyle{#1}}%
   {\XXint\scriptstyle\scriptscriptstyle{#1}}%
   {\XXint\scriptscriptstyle\scriptscriptstyle{#1}}%
   \!\int}
\def\XXint#1#2#3{{\setbox0=\hbox{$#1{#2#3}{\int}$}
     \vcenter{\hbox{$#2#3$}}\kern-.5\wd0}}
\def\dashint{\Xint-}

\DeclareMathOperator{\supp}{\mathrm{Supp}}

\def\1{\mathds{1}}

\def\ainfty{{\mathcal{A}_1}}

\def\bm{{\underline{m}}}

\def\dist{\text{dist}\ }
\def\div{\mathrm{div} \ }

\def\D{\displaystyle}
\def\({\left(}
\def\){\right)}

\def\E{E}

\def\ep{\varepsilon}

\def\hal{\frac{1}{2}}

\def\indic{\mathds{1}}

\def\mr{\mathbb{R}}

\def\nab{\nabla}

\def\p{\partial}
\def\ro{\rho}

\def\supp{\text{Supp}}

\def\fae{f_{\alpha, \eta}}

\def\vp{\varphi}
\def\cc{c_1}
\def\cC{c_2}

\def\w{{H_n}}
\def\W{\mathcal{W}}

\def\Z{Z_{n,\beta}}

\def\xbf{{\mathbf x}}

\def\I{\mathcal{E}}

\def\g{g}

\def\Q{\mathbb{P}_{n,\beta}}

\newcommand{\op}[1]{{\rm{#1}}}

\def\P{\mathcal{P}}

\def\En{\mathcal{E}}

\def\yg{|y|^\gamma}
\def\muv{\mu_V}
\def\c{c_{d,s}}
\def\M{M_{R,\eta}}
\def\Mp{M_{R,\eta'}}

\newcommand{\hne}{h_{n,\eta}}
\newcommand{\Y}{\mathcal Y}

\newcommand{\mz}{\mathbb{Z}}
\newcommand{\R}{\mathbb{R}}

\def\dist{\mathrm{dist}}

\newcommand{\loc}{\mathrm{loc}}

\newcommand{\drd}{\delta_{\mr^d}}

\numberwithin{equation}{section}

\title{Next  Order Asymptotics and Renormalized Energy for Riesz Interactions }
\author{Mircea Petrache\footnote{UPMC Univ. Paris 6,  UMR 7598 Laboratoire Jacques-Louis Lions,  Paris, F-75005 France}  and Sylvia Serfaty\footnote{UPMC Univ. Paris 6,  UMR 7598 Laboratoire Jacques-Louis Lions,  Paris, F-75005 France \newline \& Courant Institute, New York University, 251 Mercer st, NY NY 10012, USA.}
}
\date{September 26, 2014}

\begin{document}

\maketitle

\begin{abstract}
We study systems of $n$ points in the Euclidean space of dimension $d \ge 1$  interacting  via a Riesz kernel  $|x|^{-s}$ and confined by an external potential, in the regime where $d-2\le s<d$. We also treat the case of logarithmic interactions in dimensions $1$ and $2$. Our study includes and retrieves all cases previously  studied in   \cite{ss2d,ss1d,rs}.  Our approach is based on the Caffarelli-Silvestre extension formula which allows to view the Riesz kernel as the kernel of a (inhomogeneous) local operator in the extended space $\R^{d+1}$.
 
As $n \to \infty$, we exhibit  a next to leading order term in $n^{1+s/d}$  in the  asymptotic expansion of the  total energy of the system, where the constant term in factor of $n^{1+s/d}$ depends on the microscopic arrangement of the points and is expressed in terms of a ``renormalized energy." This new object  is expected to penalize the disorder of an infinite set of points in whole space, and to be minimized by Bravais lattice (or crystalline) configurations.    We  give applications to the statistical mechanics in the case where temperature is added to the system, and identify an expected ``crystallization regime."
We also obtain  a result of separation of the points  for minimizers of the energy.
\end{abstract}

%\tableofcontents

\section{Introduction}

We study the equilibrium properties of  a system of $n$  points  in the full space of dimension $d\ge 1$, interacting via Riesz kernel interactions
and confined by an ``external field" or potential $V$. More precisely, we are considering energies (or Hamiltonians) of the form  
\begin{equation}\label{Hn}
\w(x_1, \dots, x_n)=   \sum_{i\neq j} \g (x_i-x_j)  +n \sum_{i=1}^n V(x_i)
\end{equation}
where  $x_1 , \dots, x_n$ are  $n$ points in $\R^d$ and the interaction kernel is given by either
\begin{equation}\label{kernel}
 \g(x)=\frac{1}{|x|^{s}}\qquad \max(0, d-2)\le s<d,\end{equation} 
or  \begin{equation}
\label{wlog}
\g(x)=-\log |x| \quad \text{in dimension } d=1,\end{equation} 
or 
\begin{equation}
\label{wlog2d}
\g(x) = - \log |x| \quad \text{in dimension } d=2.\end{equation}
We are interested in the asymptotics $n\to \infty$ of the minimum of $\w$. Note here that  the factor $n$ in front of second term in \eqref{Hn}  puts us in a mean-field scaling where the potential term  and the pair interaction terms  are of the same order of magnitude.
  This choice is equivalent to demanding that the pair-interaction strength be of order $n^{-1}$. One can always reduce to this situation in the particular case where the confining potential $V$ has some homogeneity. 

The case $s=d-2$ for $d \ge 3$  and the case  \eqref{wlog2d} correspond to the Coulomb interaction cases in dimension $d\ge 2$. The cases $ d-2<s<d$ correspond to more general Riesz interactions. 
Systems of points with Riesz interaction have particularly attracted attention in approximation theory. We refer to the forthcoming monograph of Borodachev-Hardin-Saff \cite{bhslivre}, the review papers \cite{sk,bhs}  and references therein. Such systems are mostly studied on the $d$-dimensional sphere or torus, but studying them in Euclidean space with an external confining potential $V$  is also of  interest, as  it can also correspond to physically meaningful particle  systems, or can provide tools to study such interactions on manifolds.
 The Coulomb case  was already studied in \cite{rs,ss2d}, while the case \eqref{wlog} was studied in \cite{ss1d}. We refer to these papers and the book of Forrester \cite{forrester} for references on these classical cases of Coulomb and log gases and their importance in mathematical physics and  random matrix theory.  Here we  generalize these approaches, in particular that of \cite{rs},  to the case of non-Coulomb interaction, and we retrieve by the same token most of the results obtained in \cite{ss2d,ss1d,rs}.

In approximation theory (cf. references above) the whole range $s\in \mr$ is of interest. The case $s \ge d$ is called the {\it hypersingular case} and the case $s<d$ the  {\it potential case} (cf. \cite{landkof}).
 When $s\to \infty$  the problem connects with best packing problems, and when $s\to 0$ it connects by means of $(|\cdot |^{-s}-1) /s  \to - \log |\cdot|$ to the  logarithmic interaction $-\sum_{i\neq j} \log |x_i-x_j|$, whose minimization is equivalent to the  maximization of the product of distances  $\prod_{i\neq j} |x_i-x_j|$, i.e. to Fekete points. These  are related to orthogonal polynomials and are of major interest in interpolation theory. 
The cases \eqref{wlog}--\eqref{wlog2d} which we study here correspond exactly to one and two-dimensional weighted Fekete points (cf.  \cite{safftotik})  or physically to the  ``log gas" systems mentioned above (cf. \cite{forrester}).

It is well-known since \cite{choquet}  that to leading order 
\begin{equation}\label{1order}
\min \w = n ^2 \En(\muv) +o(n^2) 
\end{equation}
in the limit $n\to \infty$, where 
\begin{equation}\label{MFener}
\En(\mu) = \iint_{\mr^d\times \mr^d} \g(x-y) \, d\mu(x)\, d\mu(y)+ \int_{\mr^d}V(x)\, d\mu(x)
\end{equation}
is the mean-field energy functional defined for Radon measures $\mu$, and the {\it equilibrium measure} $\muv$ is the minimizer of $\En$ in the space of  probability measures on $\R ^d$, denoted $\mathcal P(\R^d)$. 
This is true only for $s<d$, which is the condition for \eqref{MFener} to make sense and to have a minimizer. This is why this case is called the potential case.  

In this paper we characterize the next to leading order term, by showing among other things the following sample result (for more details, see Theorem \ref{th2}):
\begin{theo}\label{th1}
Assume \eqref{kernel} and that  $V$ is  such that the equilibrium measure  $\muv$ exists and satisfies some suitable regularity assumptions (in particular has a density, cf. below). Then we have the expansion
\begin{equation}\label{expH}
\min \w= n^2 \En(\muv)+ n^{1+\frac{s}{d}}  C_{d,s,V} +o(n^{1+\frac{s}{d}})\end{equation}
where  $$C_{d,s,V}= \xi_{d,s} \int_{\R^d} \muv^{1+\frac{s}{d}}(x)\, dx$$
and  the number $\xi_{d,s}$ depends only on $d,s$  and is characterized as the minimum of a function $\W$ described below.
\end{theo}
The fact that the next order term lies at order $n^{1+s/d}$ is not difficult to guess by scaling,  but  the existence of an asymptotic limit for it was the main  open question. In the context of the Riesz energy on the sphere, 
 this is  listed as a conjecture (with a precise conjectured value for the constant if $d=2,4,8,24$)  in \cite[Conjecture 3]{bhs}. In addition to finding the existence of the asymptotic term, we exhibit a new object, called  ``renormalized energy" and denoted $\W$, which governs this order in the energy for arbitrary (possibly non minimizing) configurations. 
  
Our method relies on  expressing the interaction as a quadratic integral  of the potential  generated by the point configuration via
$$ \g* \sum_i \delta_{x_i}$$
and expanding this integral interaction to next order in $n$ to obtain a next order asymptotic limit of the energy, by an exact splitting formula, as done in \cite{ss2d,ss1d,rs}.  The main difference is that  here the Riesz kernel $\g$
is not the convolution kernel of a local operator, as  in  the  Coulomb case  $s=d-2$ or \eqref{wlog2d}, where $\g$ is the kernel of the (inverse)  Laplacian operator. Instead it is the kernel of a   nonlocal one, more precisely a fractional Laplacian. It turns out however that if  $d-2<s<d$, this fractional Laplacian nonlocal operator can be transformed into a local but inhomogeneous operator of the form $\div (\yg \nab \cdot)$ by adding one space variable  $y\in \R$ to the space $\R^d$. In the particular case of $s=d-1$ then $\gamma=0$ and this corresponds to using a harmonic extension, a relatively common procedure which  seems to have originated in the probability literature in \cite{molchanov} . In the more general setting, the extension procedure is due Caffarelli and Silvestre \cite{cafsil}, and    has been much used to study nonlocal PDE's involving fractional Laplacians.   Using it in our context is  in the line of \cite{ss1d} where the harmonic  extension from dimension $1$ to $2$ was 
used to transform the one-dimensional logarithmic interaction into the two-dimensional Coulomb interaction.
We note that the boundaries of the constraint  $s\in [d-2, d)$ we inherit from this approach, are quite natural: $d-2$ is the Coulomb case, and $s<d$ is the regime for which a potential theory associated to the Riesz kernel exists  \cite{landkof}.  The constraint $s\in [d-2,d)$ also appears in most of the results on the sphere in \cite{bhs2}. 

Transforming the nonlocal relation into a local one (in an extended space) allows us to follow the strategy of \cite{ss2d,rs}, however work is required to show that the strategy still works in an extended space and with an inhomogeneous operator. Also, \cite{rs} relied on  Onsager's lemma, which itself relies on Newton's theorem (that any point charge generates  outside of a ball the same Coulomb potential  as  the same charge which has been radially smeared out in the ball), only valid for Coulomb potentials. Here, we replace that use with a simple truncation procedure (a simplification which works in the Coulomb case too).

The proof of our main result relies on first  proving   lower bounds for the next order in $\w$ of a generic configuration in terms of a limiting energy $\widetilde{\W}$ that we introduce, and second in constructing test configurations which allow to achieve the value $\min \widetilde{\W}$. Just as in \cite{ss2d,rs} the upper bound  construction relies on a ``screening procedure" of a generic configuration.  The proof of this screening takes up a large part of the paper, since it needs to be completely redone  in the extended space $\R^{d+1}$ with the corresponding inhomogeneous operator. 

As a byproduct of these matching lower and upper bounds, we not only obtain the asymptotic expansion \eqref{expH} but we also obtain information on the minimizers themselves: after a suitable rescaling, they have to minimize $\widetilde{\W}$. $\widetilde{\W}$ is an appropriate average of a quantity   $\W$, which itself is the energy of an infinite configuration of discrete points in the whole space $\R^d$, with a suitable uniform ``neutralizing background" corresponding to the average point density. It is the Riesz analogue of the renormalized energy that was introduced in the 2D Coulomb case first by Sandier-Serfaty \cite{gl13,ss2d}, and then in higher dimension by Rougerie-Serfaty \cite{rs}.  

Our analysis thus  leads to the question of minimizing $\W$ itself.
While we know how to prove the existence of minimizers of $\W$ and  a few of their qualitative properties, the identification of its minimum remains  widely open. The only few exceptions are   the case of the one-dimensional logarithmic  interaction (this could very likely be extended to Riesz interactions), for which the minimum is proven in \cite{ss1d} to be achieved at the perfect lattice configuration $\mathbb{Z}$ (for some uniqueness result, see \cite{leb}), and  the two-dimensional Coulomb case  in \cite{gl13} where it is shown that within the class of configurations of points  that are forming a perfect lattice, the minimizer is the ``Abrikosov" triangular lattice (with $\pi/3$ angles).  This relies heavily on  the corresponding result from number theory \cite{cassels} which asserts that  in dimension 2 the minimum of the Epstein Zeta function of a lattice of fixed volume is uniquely achieved by the triangular lattice. \footnote{Some analogous special lattices exist in dimensions $4,8$ and $24$: the lattice $D_4$ in dimension 
$4$, $E_8$ in dimension $8$ and the Leech lattice in dimension $24$. They are proven in \cite{sarnakstormbergsson} to be local minimizers of the Epstein Zeta function among lattices of fixed volume. }
These results led to conjecturing in \cite{gl13} that the triangular lattice indeed achieves the minimum of $\W$ among all possible configurations. By mapping the plane to the $2$-sphere,  B\'etermin \cite{bet}  showed that that conjecture is equivalent to the conjecture of \cite{bhs}, itself formulated on the sphere.
Here we prove the analogous result to that of \cite{gl13} in the more general Riesz case, i.e. that the triangular lattice is the minimizer of $\W$ among lattices of volume $1$. We may then naturally extend the previous conjecture to one that says  that in dimension 2 the minimum of $\W$ for all  $0\le s<2$ is achieved by the triangular lattice. For general dimension, we may also conjecture that the  minimum of $\W$ for all  $\max(0, d-2)\le s<d$ is always achieved  by some lattice, which in dimensions $4,8,24$  is the lattice $D_4$, $E_8$, and  Leech respectively; this can be expected to also be equivalent to Conjecture 3 in \cite{bhs} in the corresponding cases. Note that the fact that these special lattices should be minimizing for a broad class of interaction kernels appears e.g. in the Cohn-Kumar conjecture \cite{ck}.
Finally, let us mention that a few such cristallization results are known mostly in one dimension, e.g. one-dimensional Coulomb gases (i.e. with the one-dimensional Coulomb interaction kernel $|x|$ which is not treated here) \cite{bl,len1,len2}, zeroes of orthogonal polynomials \cite{als}, and in dimension  2 and 3 for some very particular interaction kernels  \cite{theil,radin,harristheil}.
\medskip

As mentioned, the approach used here allows to retrieve all previous results \cite{ss2d,ss1d,rs} in one unified approach, but also brings a few  simplifications :
\begin{itemize}
\item we use a simple definition of the renormalized energy by truncation. This allows to bypass the use of Onsager's  lemma and Newton's theorem in \cite{rs}. As in \cite{rs}, this also avoids the use of mass displacement of  \cite{ss2d,ss1d} since the energy thus defined is immediately bounded below.
\item by proving a monotonicity in the truncation parameter, we easily deduce that $\W$ is bounded from below without having to first prove that 
 minimizing configurations have well separated points (which we do prove later). 
 \item our proof of screening, which has many features in common with that  of \cite{ss1d} and with that of \cite{rs}  is also more general than both of them in that  we screen (almost) arbitrary configurations and not only those that have well separated points, a result of independent interest. 
\item the assumptions on the equilibrium measure are weakened, and we distinguish more precisely those that suffice for the lower bound and those that are needed for the upper bound.
\end{itemize}

  In addition to our main results about minimizers,  we obtain two additional results, which are stated at the end of the introduction. One is an independent result of  good separation of points for energy minimizers.  The other is an application of our 
method to the statistical mechanics model  of particles interacting via $\w$ with temperature, in particular a next order expansion of the partition function.  

Let us now get into the specifics.

\subsection{The equilibrium measure and our assumptions}
  We first place assumptions on $V$ that ensure the existence of  the equilibrium $\muv$ 
  from standard potential theory: 
\begin{eqnarray}\label{assv1}
&V \mbox{ is l.s.c. and bounded below}
\\ \label{assv2}
&\{x: V(x)<\infty\} \mbox{ has positive $\g$-capacity}
\\ \label{assv3}
&\lim_{|x|\to\infty} V(x)=+\infty, \quad \text{resp.} \  \lim_{|x|\to \infty} \frac{V(x)}{2}- \log |x|= + \infty \  \text{in cases }\eqref{wlog}-\eqref{wlog2d}
\end{eqnarray}

The following theorem, due to Frostman (cf. also \cite{safftotik}) then gives  the existence and characterization of the equilibrium measure:

\begin{theo}[\cite{frostman}]\label{thfrostman}
 Assume  that $V$ satisfies \eqref{assv1}--\eqref{assv2}, 
then there exists a unique minimizer $\mu_V\in\mathcal P(\mathbb R^d)$ of $\I$ and $\I(\mu_V)$ is finite. Moreover the following properties hold:
\begin{itemize}
 \item $\Sigma := \supp(\muv)$ is bounded and has positive $\g$-capacity,
 \item for $c:=\I(\muv) - \int\tfrac{V}{2}d
 \muv$ and $h^{\muv}(x):=\int \g(x-y) d\muv(y)$ there holds
 \[
  \left\{\begin{array}{ll} h^{\muv} +\tfrac{V}{2} \ge c& \text{ q.e.}\ , \\ [2mm]
          h^{\muv}+\tfrac{V}{2}=c&\text{ q.e. on } \Sigma.
         \end{array}\right.
 \]
\end{itemize}\end{theo}
We will write 
\begin{equation}\label{defzeta}
\zeta:=h^{\muv} +\tfrac{V}{2} - c\ge 0.\end{equation}
We will assume that $\muv$ is really a $d$-dimensional measure (i.e. $\Sigma$ is a nice $d$-dimensional set), with a density, and just as 
 in \cite{ss2d,rs}, in order to make the explicit constructions easier, we need  to assume that this  density (that we still denote $\muv$) is  bounded  and sufficiently regular on its support. More precisely, we make the following assumptions (which are technical and could certainly be somewhat relaxed): 

\begin{eqnarray}
\label{assumpsigma}
& \partial \Sigma\  \text{is} \ C^1
\\
\label{assmu1}
& \text{$\muv$  has a density which is $ C^{0, \beta}$ in $\Sigma$, } \\ \label{assmu2}
&  \text{ $\exists \cc, \cC, \overline m>0$ s.t.  $\cc \dist (x, \p \Sigma)^\alpha \le \muv(x) \le \min(\cC \dist(x, \p \Sigma)^\alpha,  \overline m)<\infty $  in $\Sigma$, }
\end{eqnarray}
with the conditions 
\begin{equation}\label{condab}
0<\beta \le 1,\qquad 0\le \alpha \le \frac{2\beta d}{2d-s}.\end{equation}
Of course if $\alpha<1$ one should take $\beta=\alpha$, and if $\alpha\ge 1$, one should take $\beta=1$ and $\alpha \le \frac{2d}{d-s}$.
These assumptions are meant to include the case of the semi-circle law  $\frac{1}{2\pi} \sqrt{4-x^2}\indic_{|x|<2}$ arising for the quadratic potential in the setting \eqref{wlog}. We also know that in the Coulomb cases, a quadratic potential gives rise to an equilibrium  measure  which is a multiple of a characteristic function of a ball, also covered by our assumptions with $\alpha=0$. Finally, in the Riesz case, it was noticed in \cite[Corollary 1.4]{CGZ}  that any compactly supported radial profile can be obtained as the equilibrium measure associated to some potential. Our assumptions are thus never empty.

We also  note that the problem of minimizing $\I$ can be recast as a fractional obstacle problem \cite{sil,cs,crs}, thus the regularity of $h^{\muv}$ and the free-boundary $\partial \Sigma$ as a function of $V$ are known.

% In the case \eqref{wlog}, we refer to \cite{ss1d} for the assumptions on $V$ and the existence of the equilibrium measure.

\subsection{The extension representation for the fractional Laplacian}
In what follows, $k$ will denote the dimension extension. We will take $k=0 $ in all the Coulomb cases, i.e. $s=d-2$ and $d \ge 3$ or \eqref{wlog2d}. In all other cases, we will need  to take $k=1$.
 Points in the space $\mr^d$ will be denoted by $x$, and points in the extended space $\mr^{d+k}$ by $X$, with $X=(x,y)$, $x\in \mr^d$, $y\in \mr^k$. We will often identify $\R^d \times \{0\}$ and $\R^d$.
By balls $B_R$ or $B(X,R)$  we will mean balls in the space $\mr^{d+k}$, unless otherwise specified.
% We will also use
%the notation 
%\begin{equation}\label{notkr}
%K_R=[-R/2,R/2]^d  \mbox{ as well as} \ K_R=[-R/2,R/2]^d \times \{0\}\end{equation}
%\begin{equation}\label{tkr}
%\tilde{K}_R= [-R/2,R/2]^d \times (-R/2,R/2)\end{equation}

The extension representation  of \cite{cs} relies on the remark that if  $d-2+k+\gamma>1$, a function of the form $C/|x|^{d-2+k+\gamma}$ appears as the $y=0$ restriction of the fundamental solution
\[
 G_\gamma(x,y)=\frac{C_{\gamma+d+k}}{\left(|x|^2+|y|^2\right)^{\frac{d-2+k+\gamma}{2}}}
\]
of the operator $\div(\yg\nabla u)$ on $\mathbb R^d\times\mathbb R^k$. Here $C_q=\pi^{q/2}\Gamma(q/2-1)/4$. What we mean by fundamental solution is that it solves
\[
\left\{
 \begin{array}{l}
\div(\yg\nabla u)=0\text{ on }\mathbb R^d\times(\mathbb R^k\setminus\{0\})\ ,\\
-\lim_{|y|\downarrow 0}|y|^{\gamma+k-1}\p_y u(\cdot,y)=\delta_0\text{ on }\mathbb R^d\times\{0\}\, ,
 \end{array}
\right.
\]
where $\delta_0$ denotes the Dirac mass at the origin. 
%where $\p_r$ indicating the derivative in the radial direction of $\mathbb R^k$.
It follows that if $\gamma$ is chosen such that 
\begin{equation}\label{gs}
d-2+k+ \gamma  =s\, , \end{equation}
then, given a measure $\mu $ on $\R^d$, the potential $h^\mu(x)$ generated by $\mu$ defined in $\R^d$ by 
\begin{equation}\label{hmu}
h^\mu(x)= \g* \mu(x)= \int_{\R^d} \frac{1}{|x-x'|^s} \, d\mu(x')\end{equation}
can be extended to a function $h^\mu(X)$ on $\R^{d+k}$ defined by 
\begin{equation}\label{hmu2}
h^\mu(X) = \int_{\R^d} \frac{1}{|X- (x',0)|^s } \, d\mu(x') \end{equation}
and this function satisfies 
\begin{equation}
\label{divh}
- \div (\yg \nab h^\mu)=  c_{d,s} {\mu}\delta_{\mr^d}\end{equation}
where by $\drd $ we mean the uniform  measure on $\mr^d\times \{0\}$ i.e.  $\mu \drd$ 
acts on test functions $\vp$ by 
$$\int_{\R^{d+k}} \vp (X)d (\mu \drd)(X) = \int_{\R^d} \vp(x, 0) \, d\mu(x),$$
 and 
 \begin{equation}
 c_{d,s}= \left\{\begin{array}{ll}2s\,\frac{2\pi^{\frac{d}{2}}\Gamma\left(\frac{s+2-d}{2}\right)}{\Gamma\left(\frac{s+2}{2}\right)}&\text{ for }s>\max(0,d-2)\ ,\\[3mm]
                      (d-2)\frac{2\pi^{\frac{d}{2}}}{\Gamma(d/2)}&\text{ for }s=d-2>0\ ,\\[3mm]
                      2\pi&\text{ in cases \eqref{wlog}, \eqref{wlog2d}}\ .
                     \end{array}
\right.
 \end{equation}
 In particular $\g(X)= |X|^{-s}$ seen as a function of $\R^{d+k}$ satisfies
 \begin{equation}\label{eqg}
 - \div (\yg \nab \g)= \c \delta_0.
 \end{equation}
 In order to recover the Coulomb cases, it suffices to take $k=\gamma=0$. If $s>d-2$ we take $k=1$ and 
$\gamma$ satisfying \eqref{gs}. In the case \eqref{wlog}, we note that $\g(x)=- \log |x|$ appears as the $y=0$ restriction 
of $-\log |X|$, which is (up to a factor $2\pi$) the fundamental  to the Laplacian operator in dimension $d+k=2$. In this case, we may thus choose $k=1$ and $\gamma=0$, $c_{d,s}=c_{1,0}=2\pi$, and the potential $h^\mu=\g *\mu$ still satisfies \eqref{divh}, while $\g$ still satisfies \eqref{eqg}. This is the procedure that was used in \cite{ss1d}, and we see that this case naturally embeds into the Riesz setting we are studying. 

To summarize, we will take 
\begin{eqnarray}
\label{casriesz} 
&\text{in the case} \ \max(0, d-2) < s<d, & \text{ then } \  k=1, \ \gamma=s-d+2-k\, , 
\\
\label{cas1d}
&\text{in the case } \eqref{wlog}, & \text{ then } \ k=1, \ \gamma=0\, , \\
\label{cas2d}
&\text{in the case } \eqref{wlog2d} \text{ or } d\ge 3, s=d-2, & \text{ then} \ k=0, \ \gamma=0\, .
\end{eqnarray}
We note that the formula \eqref{gs}
  always remains formally true when taking the convention that $s=0$ in the case $\g(x)=-\log |x|$, and we also note that the  assumption $d-2\le s<d$  implies that in all cases $\gamma \in  (-1,1).$

Through this extension procedure, we will be led to studying equations of the form \eqref{divh}. These are degenerate elliptic equations, however they are associated to weights which are in the Muckenhoupt class $A_2$ for which there is a good elliptic theory \cite{fks}.
\subsection{Definition of $\W$}
Before defining our renormalized energy $\W$, we define  the  truncated Riesz (or logarithmic) kernel as follows:
for  $1>\eta>0$ and  $X\in \R^{d+k}$, let 
\begin{equation}
\label{feta} f_\eta(X)= \left(\g(X)- \g(\eta)\right)_+.
\end{equation}
We note that the function  $f_\eta$ vanishes outside of $B(0, \eta)$ and satisfies that
 \begin{equation}\label{defde}
\delta_0^{(\eta)}:= \frac{1}{\c}\div (\yg \nab f_\eta)+ \delta_0\end{equation}
 is a positive  measure supported on $\p B(0, \eta)$, and which is such that  for any test-function $\varphi$, 
$$\int \vp \delta_0^{(\eta)}=\frac{1}{\c} \int_{\p B(0, \eta)}  \vp (X)\yg \g'(\eta) .$$
One can thus check that $\delta_0^{(\eta)}$ is a positive measure of mass $1$, and we may write 
\begin{equation}\label{divf}
-\div (\yg \nab f_\eta) = \c ( \delta_0 - \delta_0^{(\eta)})\quad \text{in} \ \R^{d+k}.\end{equation}
We will also  denote by $\delta_p^{(\eta)}$ the measure $\delta_0^{(\eta)} (X-p)$, for $p \in \R^d\times \{0\}$.
Again, we note that this includes the cases \eqref{cas1d}--\eqref{cas2d}. In the Coulomb cases, i.e. when $k=0$, then $\delta_0^{(\eta)}$ is simply the normalized surface measure on $\p B(0, \eta)$, it is thus a particular case of the radially symmetric  smearing out performed in \cite{rs}, and thus the renormalized energy we will define next is the same as in \cite{rs}.

The renormalized energy of an infinite configuration of points is defined via the gradient of the  potential generated by the point configuration,   embedded into the extended space $\R^{d+k}$. That gradient is a vector field that we denote $E$ (like electric field, by analogy with the Coulomb case). In view of the discussion of the previous subsection, it is no surprise that $E$ will solve a  relation of the form 
\begin{equation}
\label{eqe} -\div (\yg \E) = c_{d,s} \Big(\sum_{p \in \Lambda} N_p \delta_p - m(x)\drd\Big) \quad \text{in} \ \R^{d+k}.\end{equation}
where $\Lambda $ is some discrete set in $\R^d \times \{0\}$ (identified with $\R^d$), $N_p$ are positive integers,  and $m(x)$ is to be specified.  
For any such $E$ (defined over $\R^{d+k}$ or over subsets of it), we define 
\begin{equation}\label{defeeta}
\E_\eta :=  \E- \sum_{p \in \Lambda } N_p \nab f_\eta (x-p).\end{equation}
If $\E$ happens to be the gradient of a function $h$, then we will also denote 
\begin{equation}
\label{defheta}
h_\eta:= h- \sum_{p \in \Lambda } N_p f_\eta (x-p).\end{equation} 
%We then note that 
%\begin{equation}\label{divE}
%-\div(\yg \nab h_\eta)=  \c (\sum_{p\in \Lambda} N_p \delta_p^{(\eta)} -m(x)\drd).\end{equation}
We will write $\Phi_\eta$ for the map that sends $E$ to $E_\eta$, and note that it is a bijection from the set of vector fields satisfying a relation of the form  \eqref{eqe} to those satisfying a relation of the form 
\begin{equation}
\label{eqeeta} -\div (\yg \E_\eta) = c_{d,s} \Big(\sum_{p \in \Lambda}N_p \delta_p ^{(\eta)}- m(x)\drd\Big) \quad \text{in} \ \R^{d+k}.\end{equation}

\begin{remark}\label{rmktrunc}
If $h=\g* (\sum_{i=1}^n\delta_{x_i}- m(x)\drd)$ then the 
 transformation from $h$ to $h_\eta$  amounts to truncating the kernel $\g$, but only for the Dirac part of the r.h.s. Indeed, letting $\g_\eta(x)=\min (\g(x), \g(\eta))$ be the truncated kernel, we have 
$$h_\eta= \g_\eta* (\sum_{i=1}^n\delta_{x_i}) - \g* (m\drd).$$
\end{remark}

\begin{defi}[Admissible vector fields]\label{defbam}Given a number $m\ge 0$, we define the class $\mathcal{A}_m$ to be the class of gradient vector fields $E=\nab h$ that satisfy
\begin{equation}\label{eqclam}
-\div (\yg \nab h) = \c \Big( \sum_{p \in \Lambda} N_p \delta_p-m\drd \Big)
\quad \text in\ \R^{d+k}\end{equation}
where $\Lambda $ is a discrete set of points in $\R^d\times \{0\}$ and $N_p$ are integers in $\mathbb{N}^*$.
\end{defi}This class corresponds to vector fields that will be limits of blown-ups of those generated by the original configuration $(x_1, \dots, x_n)$, after blow-up at the scale $n^{1/d}$ near the point $x$ where $m=\muv(x)$ and can be understood as the local density of points.  We note that since such vector fields   blow up exactly in $1/|X|^{s+1}$ near each $p \in \Lambda$  (with the convention $s=0$ for the cases \eqref{wlog}--\eqref{wlog2d}),  such vector fields naturally belong to the space $L^p_{\loc}(\R^{d+k}, \R^{d+k})$ for $p<\frac{d+k}{s+1}$. In fact, to ensure stability in that class, we will need to take $p<\min(2,\frac{2}{\gamma+1}, \frac{d+k}{s+1})$.

We are now in a position to define the renormalized energy. 
In the definition, we let $K_R$ denote the hypercubes $[-R/2,R/2]^d.$

\begin{defi}[Renormalized energy]
For $\nab h\in \mathcal{A}_m$ and $0<\eta<1$,  we define 
\begin{equation}
\label{Weta}
\W_\eta(\nab h)=\limsup_{R\to \infty} \left(\frac{1}{R^d}\int_{K_R\times \R^k}\yg |\nab h_\eta|^2 - m c_{d,s}\g(\eta)\right)
\end{equation}
and 
\begin{equation}\label{defW}
\W(\nab h) = \lim_{\eta\to 0} \W_\eta(\nab h).\end{equation}
\end{defi}
This is a generalization of the renormalized energy defined in \cite{rs}. As in \cite{rs} it differs from the one defined in \cite{gl13,ss1d} for the one and two-dimensional logarithmic interaction, essentially in the fact that the order of the limits $R\to \infty $ and $\eta\to 0$ is reversed.  We refer to \cite{rs} for a further discussion of the comparison between the two.

By scaling, we may always reduce to studying the class $\ainfty$, indeed, if $E\in \mathcal A_m$, then $\hat{E} = m^{-\frac{s+1}{d}}   E(\cdot m^{-1/d})\in \ainfty$ \footnote{with the convention $s=0$ in the case \eqref{wlog}} and 
\begin{equation}\label{scalingW}
\W_\eta(E)=m^{1+s/d} \W_{\eta m^{1/d}}  (\hat{E})  \qquad \W(E)= m^{1+s/d} \W(\hat{E}).  
\end{equation} in the case \eqref{kernel}, and 
respectively
\begin{equation}\label{scalinglog}
\W_\eta(E)= m\left(\W_{m\eta} (\hat{E})-\frac{2\pi}{d} \log m\right)
\qquad
 \W(E)= m\left(\W(\hat{E})-\frac{2\pi}{d} \log m\right)
\end{equation}
in the cases \eqref{wlog}--\eqref{wlog2d}.

We denote
\begin{equation}\label{xi}
\xi_{d,s}=\frac{1}{c_{d,s}}\min_{\ainfty} \W.\end{equation}
The name renormalized energy (originating in Bethuel-Brezis-H\'elein \cite{bbh} in the context of two-dimensional Ginzburg-Landau vortices) reflects the fact that $\int \yg |\nab h|^2 $ which is infinite, is computed in renormalized way by first changing $h$ into $h_\eta$ and then removing the appropriate divergent part $\c\g(\eta)$ per point.

We will prove the following facts about $\W$.
\begin{pro}[Minimization of $\W$]\label{prow}
\begin{enumerate}
\item 
The limit in \eqref{defW} exists. 
\item 
$\{\W_\eta\}_{\eta<1}$ are uniformly bounded below on $\ainfty$ by a finite constant depending only on $s $ and $d$.
\item
 $\W$ and  $\W_\eta$ have a minimizer over the class $ \ainfty$. 
 \item There exists a minimizing sequence for $\W $ (resp. $\W_\eta$) formed of periodic configurations (in $\R^d$) with period $N \to \infty$.\end{enumerate}
\end{pro} 
We can also  note that $\W$ does not feel compact perturbations of the points in $\Lambda$.
As already mentioned the questions of identifying $\min_{\mathcal A_1} \W$ is open, and we expect  some (Bravais) lattice configuration to achieve the minimum. In view of Proposition \ref{prow}, to identify the value of the minimum it would suffice to compute $\min \W$ over periodic configurations with larger and larger period, for which we have an explicit formula:

\begin{pro}[Periodic case]\label{periodic}
 Let $a_1,\ldots,a_N$ be $N$ points, possibly with repetition, in a torus $\mathbb T$ of volume $N$ in $\mathbb R^d$, for $d\ge 1$. On $\mathbb R^d$ consider the configuration corresponding to the $\mathbb T$-periodic repetition of the configuration $a_1,\ldots,a_N$. Then the following hold:
 \begin{enumerate}\item If the points $a_1,\ldots, a_N$ are not distinct then for all $E$ such that $$-\div(\yg E)=c_{d,s}\left(\sum_{i=1}^N \delta_{a_i} - \drd \right)\quad \text{in}\ \mathbb T \times \R^k $$ (we call such $E$ \emph{compatible} with the points) there holds $\W(E)=+\infty$.   \item If all points are distinct then let $H$ be the  function satisfying
  \[
   -\div\left(\yg\nab H\right)=c_{d,s}\left(\sum_{i=1}^N \delta_{a_i} -\drd \right)\  \text{in} \ \mathbb{T} \times \R^k,\quad\int_{\mathbb T\times \{0\}}H=0\ .
  \]
 Then for any vector field $E=\nab h$ compatible with $a_1,\ldots,a_N$ there holds
  \[
   \W(E)\ge\W(\nab H)\ ,
  \]
 with equality precisely when $E=\nab H$. Moreover there holds
  \begin{equation}\label{minper}
   \W(\nab H) =\frac{c_{d,s}^2}{N}\sum_{i\neq j}G(a_i-a_j) + c_{d,s}^2\lim_{x\to 0}\left(G(x) - \frac{\g(x)}{c_{s,d}}\right)\ ,
  \end{equation}
 where $G$ satisfies 
  \begin{equation}\label{defG}
   (-\Delta)^\alpha G=\delta_0 - \frac{1}{|\mathbb T|} \quad \text{in } \ \mathbb T\, \quad \int_{\mathbb T} G=0\ .
  \end{equation}
  where $\alpha= \hal(2-\gamma-k)$.
   \item In case $d=1$ we have 
   \begin{equation}
    \label{G1d}
    G(x) = 2\frac{N^{2\alpha-1}}{(2\pi)^{2\alpha}\Gamma(2\alpha)}\int_0^\infty \frac{t^{2\alpha-1} \(e^t \cos \(\frac{2\pi}{N}x\)  - 1\) }{1-2e^t \cos \(\frac{2\pi}{N}x\) + e^{2t}} \, dt\ .
   \end{equation}

 \end{enumerate}

\end{pro}
In the periodic case, $\W$ can thus be seen as a function of the points rather than of the vector fields $E$, via the formula \eqref{minper}.
We note that it would suffice to check that $G$ given by \eqref{G1d} is convex to apply the same proof as in \cite{ss1d} to obtain that the minimum of $\W$ in dimension 1 is achieved over the class $\mathcal A_1$  by the lattice $\mathbb{Z}$. We did not pursue this.

Moreover, specializing \eqref{minper} to the case of Bravais lattice configurations (i.e. $N=1$), we are able to prove the analogue of the result of \cite{gl13} in dimension 2:
\begin{theo}\label{abrikolat}
Assume $d=2$.  For any $0<s<2$, the minimum of $\W$ (seen as a function of the point configuration via \eqref{minper})  over lattices $\Lambda=\vec a \mathbb Z +\vec b\mathbb Z$ of volume $1$ (i.e. such that $\op{det}(\vec a, \vec b)=1$) is uniquely achieved (up to rotation) by  the triangular lattice, i.e the one for which $\vec{a}=(\frac{\sqrt{3}}{2})^{-1/2}(1,0)$ and $\vec{b}=
(\frac{\sqrt{3}}{2})^{-1/2}(\hal, \frac{\sqrt{3}}{2})$. 
 \end{theo}
As in \cite{gl13} the proof consists in connecting the minimization to that of finding the minimimizer over lattices of volume 1 of the Epstein Zeta function of the lattice $\Lambda$,  $\zeta_\Lambda(s)= \sum_{p \in \Lambda\backslash \{0\}}\frac{1}{|p|^s}$, which was proved in \cite{cassels,dianada,rankin,ennolacassels,montgomery} to be the triangular lattice.

\subsection{Connection to the original problem and splitting formula}
We reproduce here the framework of \cite{ss2d,ss1d,rs}.
The renormalized energy will appear as a next order limit of $\w$ after a blow-up is performed, at the inverse of the typical nearest neighbor distance between the points, i.e. $n^{1/d}$. It is expressed in terms of the potential generated by the configuration $x_1, \dots, x_n$ and defined by 
\begin{equation}\label{defhn}
h_n(X)=  \g * \left(\sum_{i=1}^n \delta_{(x_i,0)} - n {\muv}\drd\right).\end{equation}
For the blown-up quantities  we will use the following notation (with the convention $s=0$ in the cases \eqref{wlog} or \eqref{wlog2d}):
\begin{align}
& x'= n^{1/d} x \quad X'=n^{1/d} X \quad
x_i'=n^{1/d} x_i
\\ &  \muv'(x')= \muv(x)\\
& h_n'(X')= n^{-\frac{s}{d}} h_n(X).\label{rescalh}
\end{align}
We note that in view of \eqref{defhn}, \eqref{hmu}, \eqref{divh},    $h_n$ and  $h_n'$ satisfy
\begin{equation}\label{hn}
-\div (\yg \nab h_n) = \c \Big( \sum_{i=1}^n \delta_{x_i} - n \muv \drd\Big) \quad \text{in } \ \R^{d+k}\, ,
\end{equation}
\begin{equation}\label{hnp}
-\div (\yg \nab h_n') = \c \Big( \sum_{i=1}^n \delta_{x_i'} - \muv' \drd\Big) \quad \text{in } \ \R^{d+k}\, ,  \end{equation}
while in view of   \eqref{divf}, $h_{n, \eta}$ and $h_{n, \eta}'$, defined from  $h_n$ and $h_n'$ via \eqref{defheta}, satisfy
\begin{equation}\label{hne}
-\div (\yg \nab h_{n, \eta}) = \c \Big( \sum_{i=1}^n \delta_{x_i}^{(\eta)} - n \muv \drd\Big) \quad \text{in } \ \R^{d+k}\, ,  \end{equation}
\begin{equation}\label{hnpe}
-\div (\yg \nab h_{n, \eta}') = \c \Big( \sum_{i=1}^n \delta_{x_i'}^{(\eta)} - \muv' \drd\Big) \quad \text{in } \ \R^{d+k}\, ,  \end{equation}
with the usual embedding of $\R^d $ into $\R^{d+k}$.
\\
The next proposition connects $\w$ with these quantities via an algebraic  identity. 
\begin{pro}[Splitting formula]\label{splitting}
For any $n$, any $x_1, \dots, x_n$ distinct points in $  \R^d\times\{0\}$, letting $h_n$ be as in \eqref{defhn} and  $\hne$ deduced from it via   \eqref{defheta}, we have in the case \eqref{kernel}
\begin{multline}\label{propsplit}
\w(x_1, \dots, x_n)  
= n^2 \I (\muv) +2 n\sum_{i=1}^n \zeta(x_i) \\
+  n^{1+\frac{s}{d}}    \lim_{\eta\to 0} \frac{1}{\c} \left( \frac{1}{n  }\int_{\R^{d+k}}\yg |\nab \hne'|^2 -\c  \g(\eta)\right),
\end{multline} respectively in the cases \eqref{wlog}--\eqref{wlog2d}
\begin{multline}
\w(x_1, \dots, x_n)
= n^2 \I (\muv) +2 n\sum_{i=1}^n \zeta(x_i) - \frac{n}{d}\log n\\
+ n    \lim_{\eta\to 0}\frac{1}{\c} \left( \frac{1}{n }\int_{\R^{d+k}}\yg |\nab \hne'|^2 - \c \g(\eta)\right).\end{multline}
\end{pro}

Recalling that $\zeta $ defined in \eqref{defzeta} is nonnegative and $0$ in $\Sigma$ (it acts like an effective potential whose only role is to confine the points to $\Sigma$), we see that this formula easily allows to get a next order lower bound for $\w$, and that there remains to take the $n\to \infty$ limit in the parenthesis in the right-hand side of \eqref{propsplit}, which will lead to $\W$. 

In order to make this rigorous, we need to introduce a $\W$ at the level of the ``electric field process", and a way of averaging $\W$ with respect to blow-up centers in $\Sigma.$

More precisely, given any configuration $\xbf_n = (x_1, \dots, x_n)$, we denote 
$\nu_n= \sum_{i=1}^n \delta_{x_i}$. 
%  and  $\nu_n' = \sum_{i=1}^n\delta_{x_i'}$  where the primes denote blown-up quantities ($x'=n^{1/d}x$). 
 The  configuration generates (at the blown-up scale) an electric field 
$E_{\nu_n} $ given by $\nab h_n'(X') $ above.
As already mentioned, such electric fields naturally live in the spaces of vector fields 
$ L^p_\loc(\mr^{d+k},\mr^{d+k}) $
 for $p\in[1,  \min(2,\frac{2}{\gamma+1}, \frac{d+k}{s+1}))$. Choosing once and for all such a $p$, we define $\mathcal X:= \Sigma \times L^
 p_\loc(\mr^{d+k},\mr^{d+k}) $ the space of ``marked" electric fields, where the mark $x\in \Sigma$ corresponds to the point where we center the blow-up.   We denote by 
 $\P(\mathcal X)$  the space of probability measures on $\mathcal X$ endowed with the topology of weak convergence, which can be seen as the space of  random marked vector fields. 
 
We  may now naturally associate  to  each configuration $\xbf_n= (x_1, \dots, x_n)$ a ``marked electric field process" $P_{\nu_n}$ via the map 
\begin{eqnarray}
\label{in}
i_n : &  (\mr^d)^n \longrightarrow \P(\mathcal X) 
\\  \label{pnun} & \xbf_n \mapsto P_{\nu_n} :=  \dashint_{\Sigma} \delta_{(x,\E_{\nu_n}(n^{1/d} x+\cdot))}\,dx, 
\end{eqnarray}
i.e. $P_{\nu_n} $ is  the push-forward of the normalized Lebesgue measure on $\Sigma$ by $x \mapsto (x, \E_{\nu_n} ({n}^{1/d} x+\cdot)).$
Another way of saying is that each $P_{\nu_n} (x, \cdot)$ is equal to a Dirac at the electric field generated by $\xbf_n$, after centering at the point $x$. 

The nice feature is that, under an energy bound on the sequence  $(\xbf_n)_n$,  the sequence $\{P_{\nu_n}\}_n$ will be tight as $n\to \infty$, and thus converge to an element $P$ of  $\P (\mathcal X)$. In a probabilistic point of view, $P$  is  a marked  electric field process. In an analysis point of view, $P$ is similar to a Young measure on micropatterns formed by the configuration, as e.g. in \cite{albmu}. 

By construction, such a $P= \lim_{n\to \infty} P_{\nu_n}$ satisfies three properties, which are summarized in the following definition:
\begin{defi}[admissible probabilities]\label{admissible}
We say $P\in \P (\mathcal X)$ is \em{admissible} if 
\begin{itemize}
\item The first marginal of $P$ is the normalized Lebesgue measure on $\Sigma$.
\item It holds for $P$-a.e. $(x, E)$ that $E \in \mathcal{A}_{\muv(x)}$.
\item $P$ is $T_{\lambda(x)}$-invariant.
\end{itemize} 
\end{defi}
Here $T_{\lambda(x)}$-invariant is a strengthening of translation-invariance, related to the marking:
\begin{defi}[ $T_{\lambda(x)}$-invariance]\label{invariance} We say a probability measure $P$ on $\Sigma \times L^p_\loc(\mr^{d+k},\mr^{d+k})$ is $T_{\lambda(x)}$-invariant if $P$ is invariant  by $(x,E)\mapsto\left(x,E(\lambda(x)+\cdot)\right)$, for any $\lambda(x)$ of class $C^1$  from $\Sigma$ to $\mr^d$.
\end{defi}

Note that from such an admissible electric field process, using that $E$ solves  \eqref{eqclam} one can immediately get a (marked)  point process by taking the push-forward of  $P(x, E) $ by $E \mapsto \frac{1}{c_{s,d}} \div(\yg E )+ \muv(x) \drd$.

For each  $P\in \P(\mathcal X)$, we may then define
\begin{equation}
\label{Wtilde}
\widetilde{\W}(P) = \frac{|\Sigma|}{c_{d,s}} \int \W(E)\, dP(x, E)
\end{equation} if $P$ is admissible, and $+\infty$ otherwise.
%In view of Theorem \ref{thmini} and the definition of admissible, the minimum of $\widetilde{\W}$ can be readily computed to be 
%\begin{equation}
%\label{defa}
%\min \widetilde{\W}=  \min_{\bar\ainfty} \W  \int_\Sigma \muv(x)^{1+s/d}, dx.\end{equation}

With these definitions at hand, we may state our main result on 
 minimizers of $\w$ which improves Theorem \ref{th1}.
It  identifies the next order  $\Gamma$-limit (in the sense of $\Gamma$-convergence) of $\w$ and allows a description of  the minimizers at the microscopic level. Below we abuse notation by writing $\nu_n= \sum_{i=1}^n \delta_{x_i}$ when it should be $\nu_n= \sum_{i=1}^n \delta_{x_{i,n}}$.

\begin{theo}[Microscopic behavior of Riesz energy minimizers] \label{th2}  Let the potential $V$ satisfy assumptions \eqref{assv1}--\eqref{assv3} and \eqref{assumpsigma}--\eqref{assmu2}.
% Fix  $1<p<2$ and let $X= \E \times L^p_\loc(\mr^2,\mr^2)$.
\smallskip
Let $(x_1, \dots, x_n) $ minimize $\w$ and  let $P_{\nu_n}$ be associated via 
\eqref{pnun}.
 Then any subsequence of $\{P_{\nu_n}\}_n$ has a convergent subsequence converging
as $n\to \infty$  to an admissible  probability measure $P\in\P(\mathcal X)$  and, in the case \eqref{kernel},
\begin{equation}\label{thlow}\lim_{n\to \infty} n^{  -1-s/d }\left( \w(x_1, \dots, x_n) - n^2 \En(\muv)\right) = \widetilde{\W} (P) = \min_{P \ \mathrm{admissible }} \widetilde{\W}=  \xi_{d,s}\int \muv^{1+s/d} ,\end{equation} 
respectively in the case \eqref{wlog}--\eqref{wlog2d}
\begin{multline}
\label{thlow1d}\lim_{n\to \infty} n^{  -1}\left( \w(x_1, \dots, x_n) - n^2 \En(\muv)+\frac{n}{2} \log n \right)
 \\= \widetilde{\W} (P) = \min_{P \ \mathrm{admissible }} \widetilde{\W}=  \xi_{d,0}-\frac{1}{d}\int \muv \log \muv ;\end{multline} 
and $P$ is a minimizer of $\widetilde{\W}$, and $E$ minimizes $\W$ over $ \mathcal{A}_{\muv(x)}$ for $P$-a.e. $(x,E)$.
\end{theo}
Thus, our result reduces the original question to the minimization of $\widetilde{\W}$ hence that of $\W$, which we already discussed.  If one believes in the conjecture that minimizers of $\W$ are Bravais lattices, then from the formal statement ``$E$ minimizes $\W$ over $ \mathcal{A}_{\muv(x)}$"  it   can  be expected that after zooming at the right scale around $x$, minimizing configurations look  like a cristal with the appropriate density  $\muv(x)$.

We conclude our introduction with the two additional results. The first one states that minimizers have points that are well-separated at the expected scale of their typical distance, i.e. $n^{-1/d}$. We provide a short proof of this result based solely on the extension representation and maximum principle arguments. This is to be compared with the analogous statements in \cite{bhs2} for points on the sphere, which rely on fine potential theory arguments.  

\begin{theo}[Point separation]\label{separation}
Assume  $V$ satisfies \eqref{assv1}--\eqref{assv3} and $\muv$ has a density which satisfies  $\|\muv\|_{L^\infty} \le \overline{m}$. Let $(x_1, \dots, x_n)$ minimize $H_n$. Then for each $i\in [1,n]$, $x_i \in \Sigma$, and for each $i\neq j$, it holds
$$|x_i-x_j| \ge \frac {r}{(n \overline{m}) ^{1/d}}, $$ where $r$ is some positive constant depending only on $s$ and $d$. \end{theo}

The last result concerns the application to statistical mechanics. 
As in \cite{ss2d,ss1d,rs}, our method can be used to obtain as well next order information on  such a system of particles with temperature. More precisely, let us consider the Gibbs measure
\begin{equation}\label{gibbs}
d\Q(x_1, \dots, x_n)= \frac{1}{\Z} e^{-\beta \w(x_1, \dots, x_n)}\, dx_1 \dots, dx_n
\end{equation}
where  $\beta>0$ is an inverse temperature and $\Z$ is the partition function of the system, i.e. a number that  normalizes $\Q$ to a probability measure on $(\R^d)^n$. Then, under the additional assumption that 
there exists $\beta_1>0$ such that
\begin{equation}\label{integrabilite}
\begin{cases}
\int e^{-\beta_1 V(x)/2} \, dx <\infty  & \mbox{in the case } \eqref{kernel} \\
\int e^{-\beta_1( \frac{V(x)}{2}- \log |x|)} \, dx <\infty & \mbox{in the cases \eqref{wlog}--\eqref{wlog2d}}, 
\end{cases}
\end{equation}
we can obtain  the following:
\begin{theo}[System with temperature]\label{stat} Assume all the previous assumptions on $V$, i.e. \eqref{assv1}--\eqref{assv2}, \eqref{assumpsigma}--\eqref{assmu2}, and \eqref{integrabilite}.
Let $\bar{\beta}:=\limsup_{n\to+\infty} \beta n^{s/d}$, \footnote{always with the convention $s=0$ in the case \eqref{wlog}} and assume $0<\bar{\beta}\le + \infty$. There exists $C_{\bar{\beta}}>0$ depending only on $V$,  $d$ and $s$,  with $\lim_{\beta \to \infty} C_{\bar{\beta}}=0$ and  $C_{\infty}=0$  such that  the following hold.
\begin{enumerate}
\item
\begin{equation}\label{eq:partition low T}
\limsup_{n\to \infty}  n^{-1-s/d} \left| -\frac{\log \Z}{\beta}  - n ^2 \En (\muv)  - n ^{1+s/d} \min \widetilde{\W}     \right| \leq    C_{\bar{\beta}},
\end{equation} 
respectively in the cases \eqref{wlog}--\eqref{wlog2d}  
\begin{equation}\label{eq:partition low T 2d}
\limsup_{n\to \infty} n^{-1} \left| -\frac{\log\Z}{\beta} - n ^2 \En (\muv) +  \frac{n}{d}\log n - n\min \widetilde{\W}  \right| \leq  C_{\bar{\beta}} .
\end{equation}
\item 
For fixed $\beta>0$,  letting $\widetilde{ \Q}$ denote the push-forward of $\Q$ by $i_n$ (defined in \eqref{in}), $\{\widetilde{\Q}\}_{n}$ is tight and converges as $n\to \infty$, up to a subsequence, to a probability measure on $\mathcal P(\mathcal X)$ which is concentrated on admissible probabilities satisfying  $\widetilde{\W}(P) \le \min \widetilde{\W} + C_{\bar{\beta}}$. 
\end{enumerate}
\end{theo}We could also express our results in terms of next order correction to mean field theory exactly as  in \cite{rs}. This extension is left to the reader.

To our knowledge, there is no such result in the literature  beyond the much studied case of systems with logarithmic interactions in one and two dimensions, also called $\beta$-ensembles  (cf. \cite{ss1d} and \cite{ss2d} for the references).
 
This theorem indicates that there is a transition temperature regime $\beta\gg n^{-s/d}$ beyond which the system tends to concentrate on minimizers of $\widetilde{\W}$, hence should be expected to cristallize. For lower $\beta$ on the contrary, it is expected that temperature creates disorder.

\medskip

The rest of the paper is organized as follows.
We start by proving the splitting formula and preliminaries on $\W$.
In particular we show that $\W_\eta$ is essentially monotone in $\eta$. This fact directly provides a  uniform lower bound  on $\W_\eta$. In Section \ref{sec3} we study the periodic case and prove Proposition \ref{periodic} and Theorem \ref{abrikolat}. In Section~\ref{sec3b}, we prove Theorem \ref{separation} and an analogous separation result for minimizers of $\W$ with periodic boundary condition. 
In Section \ref{sec4} we show the general next order lower bound corresponding to the main result, using the approach based on the ergodic theorem initiated in \cite{gl13}.
Section~\ref{sec5} is devoted to the proof of the general screening result.
In Section~\ref{sec6}, using the screening result we make the construction that allows to obtain the upper bound for the minimal energy.
\\

{\bf Additional remarks and acknowledgements: } After this work was completed, we learnt of a related forthcoming work by Hardin, Saff, Simanek and Su \cite{hsss}, who obtain the existence of the same order asymptotic expansion of the minimal energy for the Riesz or logarithmic interaction on a flat torus for any $s<d$.

This paper benefited from comments and discussions with Thomas Lebl\'e, whom we would like to warmly thank here.
 
 The first author was supported by a postdoctoral fellowship of the Fondation Sciences Math\'ematiques de Paris. The second author was supported by a EURYI award.
 
\section{Splitting formula and preliminaries on $\W$}\label{sec2}
In the whole paper, whenever possible we treat 
 all the cases \eqref{casriesz}--\eqref{cas1d}--\eqref{cas2d} in one unified way. We thus need to carry the extension dimension as $k=0$ or $1$.
We note that in Sections \ref{sec2} to \ref{sec4}, we only use  weak assumptions on $\muv$: that $\muv$ has an $L^\infty$ density, and that  it is continuous at almost every point of $\Sigma$.
In the statements of all these sections where the parameter $s$ appears, the convention is that $s$ should be taken to be $0$ in cases \eqref{wlog}--\eqref{wlog2d}.

\subsection{Proof of the splitting formula}
\begin{proof}
We let $\nu_n= \sum_{i=1}^n \delta_{x_i}.$  Denoting by $\triangle $ the diagonal in $\R^d \times \R^d$, we may write 
\begin{eqnarray}
\nonumber H_n(x_1, \dots, x_n)  & = &  \sum_{i \neq j} \g(x_i- x_j) + n \sum_{i=1}^n V(x_i)  \\ 
\nonumber
& = & \iint_{\triangle^c} \g(x-y) d\nu_n(x) d\nu_n(y) + n \int V d\nu_n \\
\nonumber 
& = &  n^2 \iint_{\triangle^c} \g(x-y) d\muv(x) d\muv(y) + n^2 \int V d\muv 
 % \mbox{     } \right\rbrace \mbox{ terms in } \mu_0 
\\ 
\nonumber
 %\left. 
 & + &  2n \iint_{\triangle^c} \g(x-y) d\muv(x) d(\nu_n - n \muv)(y)
 % \mbox{     } \right\rbrace \mbox{ term in } n\mu_0 \times (\nu_n - n \mu_0) 
+ n \int V d(\nu_n - n \muv) \\ 
%\left. 
\label{finh}
& + & \iint_{\triangle^c} \g(x-y) d(\nu_n - n\muv)(x)d(\nu_n - n\muv)(y).
%\mbox{     } \right\rbrace \mbox{ terms in } d(\nu_n - n\mu_0)(x)d(\nu_n - n\mu_0)(y).
\end{eqnarray}
We now recall that $\zeta$ was defined in \eqref{defzeta} by
%\begin{eqnarray}
%h^{\mu_0} + \f{V}{2}= c & \mu_0 \textrm{-q.e. in } \Sigma, \\ 
\begin{equation}\label{defzetap}
\zeta = h^{\muv} + \frac{V}{2} - c = \int \g(x-y)\, d\muv(y)  + \frac{V}{2} - c\end{equation}
and that $\zeta=0$  in $\Sigma$ (with the assumptions we made, one can check that  $\zeta$ is continuous, so the q.e. relation can be upgraded to everywhere).

With the help of this we may rewrite the middle line in the right-hand side of \eqref{finh} as
\begin{multline*}
2n \iint_{\triangle^c} \g(x-y) d\muv(x) d(\nu_n - n \muv)(y) + n \int V d(\nu_n - n \muv) \\
 = 2n  \int (h^{\muv} + \frac{V}{2}) d(\nu_n - n\muv) = 2n  \int (\zeta + c) d(\nu_n - n\muv)  \\
 = 2n \int \zeta d\nu_n - 2n^2 \int \zeta d\muv + 2 n c \int  d(\nu_n - n\muv) = 2n \int \zeta d\nu_n.
\end{multline*}
The last equality is due to the facts that  $\zeta \equiv 0$ on the support of $\muv$ and that  $\nu_n$ and $n \muv$ have the same mass $n$. We also have to notice that since $\muv$ has an $L^{\infty}$ density with respect to the Lebesgue measure, it does not charge the diagonal $\triangle$ (whose Lebesgue measure is zero) and we can include it back in the domain of integration.
By that same argument, one may recognize in the first line of the right-hand side of \eqref{finh}, the quantity $n^2 \I(\muv)$.

We may thus rewrite \eqref{finh} as
\begin{multline} \label{splitting2}
H_n(x_1, \dots, x_n) = n^2 \I(\muv) + 2n \sum_{i=1}^n \zeta(x_i) \\ + \iint_{\triangle^c} \g(x-y) d(\nu_n - n\muv)(x)d(\nu_n - n\muv)(y).
\end{multline}
Next, we show that 
$$\iint_{\triangle^c} \g(x-y) d(\nu_n - n\muv)(x)d(\nu_n - n\muv)(y)=  \lim_{\eta\to 0}\left(\frac{1}{\c}\int_{\R^{d+k}}\yg |\nab \hne|^2 -  n \g(\eta)\right) .$$
To this aim, we compute the right-hand side of this relation. Let us extend the space dimension by $k$ and  choose $R$ so that all the points are in $B(0,R-1)$ in $\mr^{d+k}$, and $\eta$ small enough that $2\eta<\min_{i\neq j}|x_i-x_j|$.    Since $\hne= h_n$ (defined in \eqref{defhn}) at distance $\ge \eta$ from the points,  by Green's formula and \eqref{hne}, we have
\begin{multline}\label{greensplit1}
\int_{B_R}\yg  |\nabla \hne |^2 = \int_{\p B_R} h_n \frac{\p  h_n}{\p \nu} -
\int_{B_R } \hne    \div  (\yg \nab \hne)
  \\= \int_{\p B_R} h_n \frac{\p  h_n}{\p \nu} +
\c \int_{B_R } \hne    \left(\sum_i \delta_{x_i}^{(\eta)} -{\muv}\drd\right)
\end{multline}
It is easy to check that, since the total mass on the right-hand side of \eqref{hn} is $0$, $h_n$ decreases like $\g'(X)$ i.e. like $|X|^{-s-1}$ at infinity, and $\nab h_n$ like $\g''(X)$ i.e. $|X|^{-s-2}$ (with the convention $s=0$ in the logarithmic cases) and in all cases the boundary integral hence  tends to $0$ as  $R\to \infty$. We thus find
\begin{multline}\label{greensplit2}
\int_{\R^{d+k} } \yg  |\nabla \hne |^2 =
\c \int_{\R^{d+k} }  \hne    \left(\sum_{i=1}^n \delta_{x_i}^{(\eta)} -\muv\drd\right)\\  =
\c \int_{\R^{d+k} }  \left(h_n - \sum_{i=1}^n f_\eta(x-x_i)\right)   \left(\sum_{i=1}^n \delta_{x_i}^{(\eta)} -\muv\drd \right) \, .\end{multline}
 Since $f_\eta(x-x_i) =0 $ on $\p B(x_i, \eta)=\supp (\delta_{x_i}^{(\eta)})$ and outside of $B(x_i, \eta)$, and since the balls $B(x_i, \eta)$ are disjoint, we may write 
$$
\int_{\R^{d+k} } \yg  |\nabla \hne |^2 =\c  \int_{\R^{d+k}}   h_n\left(  \sum_{i=1}^n \delta_{x_i}^{(\eta)}  -\muv\drd\right) -
\c \int_{\R^{d+k}} \sum_{i=1}^n f_\eta(x-x_i) \muv\drd .$$
Let us now use (temporarily) the notation $h^i_n(x) = h_n(x) - \g(x -x_i)$ (for the potential generated by the distribution  bereft of the point $x_i$).
The function $h_n^i $ is regular near $x_i$, hence $\int h_n^i \delta_{x_i}^{(\eta)} \to h_n^i(x_i)$ as $\eta \to 0$.
It follows that 
\begin{multline}\label{estdistinct}
 \int_{\R^{d+k}}  h_n\left(  \sum_{i=1}^n \delta_{x_i}^{(\eta)}  - \muv\drd\right)
 \\= n\c  \g(\eta)  + \c \sum_{i=1}^n h_n^i(x_i)- \c \int_{\R^{d+k}}  h_n\muv \drd + O(n \|\muv\|_{L^\infty} ) \int_{B(0, \eta)} |f_\eta|\drd\ .
\end{multline}

In the case \eqref{kernel} we have   $ \int_{B(0, \eta)} |f_\eta|
\drd\le C\eta^{d-s}$ while in the cases \eqref{wlog}--\eqref{wlog2d} we have   $ \int_{B(0, \eta)} |f_\eta|\le \int_{B(0, \eta)} \le \eta^{d} |\log \eta|$, 
and thus, letting $\eta\to 0$, we find
\begin{equation}\label{Wsplit11}\lim_{\eta \to 0} \frac{1}{\c}  \int_{\R^{d+k} } \yg  |\nabla \hne |^2- n\g(\eta)  = \sum_{i=1}^n h_n^i(x_i)-\int_{\R^{d+k}} h_n{\muv}\drd.\end{equation}
 Now, from the definitions it is easily seen that 
\begin{equation}
h^i_n(x_i) = \int_{\R^{d}\backslash \{x_i\}} \g(x_i-y) d(\nu_n - n\muv)(y),
\end{equation}
from which it follows that 
\begin{multline*}
\iint_{\triangle^c} \g(x-y) d(\nu_n - n\mu_V)(x)d(\nu_n - n\muv)(y)
\\ = \sum_{i=1}^n  \int_{\R^d \backslash \{x_i\}}  \g(x_i-y) d(\nu_n - n\muv)(y) - n \int_{\R^{d+k}} h_n \, \muv\drd\\
= \sum_{i=1}^n h_n^i(x_i) - n \int_{\R^{d+k}} h_n \muv\drd.
\end{multline*} In view of  \eqref{Wsplit11}, we conclude that the claim holds.

The final step is to blow up and note that, using the relation $\gamma= s+2-d-k$ and a change of variables, we have
$$ \int_{\R^{d+k}}\yg |\nab h_{n,\ell}|^2= n^{\frac{s}{d}}  \int_{\R^{d+k}}\yg |\nab h'_{n,\eta}|^2
 $$ with the convention $s=0$ in the cases \eqref{wlog}--\eqref{wlog2d}, 
where $\ell= \eta n^{-1/d}$.
We may thus write 
\begin{multline*} \lim_{\eta\to 0}\left(\int_{\R^{d+k}}\yg |\nab \hne|^2 -  n\c  \g(\eta)\right)
= \lim_{\eta \to 0} \left(\int_{\R^{d+k}}\yg |\nab h_{n,  n^{-1/d} \eta}|^2 - n\c \g(n^{-1/d} \eta) \right)
\\ = \lim_{\eta \to 0} \left( n^{\frac{s}{d}} \int_{\R^{d+k}}\yg |\nab h'_{n,   \eta}|^2 -  n\c \g(n^{-1/d} \eta)   \right) 
\end{multline*}
and, combining with the above and rearranging terms, this completes the proof.

\end{proof}

\subsection{Coercivity of $\W$}
In this subsection, we show that $\W$ controls the discrepancy between the number of points in a ball and the (suitably rescaled) volume of the ball.
We start with two lemmas 
\begin{lem}[Average density of points]\label{limdensity} If $\nab h \in \mathcal A_m$ and $\W_\eta(\nab h)<+\infty$ for some $\eta<1$ then 
$$\lim_{R\to \infty} \frac{\sum_{p\in \Lambda\cap K_R} N_p}{R^d} =m.$$
\end{lem}

\begin{proof}
We note that in the case $k=0$, the result is proven in \cite[Lemma 3.1]{rs}.
We denote $K_R= [-R/2,R/2]^d\times \{0\}$ and $\tilde K_R= [-R/2,R/2]^{d+k}$.
We also recall that $E_\eta =\nab h_\eta $ satisfies 
\begin{equation}\label{delp}
-\div (\yg E_\eta)= \c \Big(\sum_{p\in \Lambda } N_p \delta_p^{(\eta)}  - m\drd\Big) ,\end{equation} and denote $\nu= \sum_{p\in \Lambda} N_p \delta_p$.
   Since $\W_\eta(E) <+\infty$ we have $\int_{K_R\times \R^k} \yg|E_\eta|^2 \le C_\eta R^d$ for any $R>1$.
Thus, by a mean value argument, we   may find  
$t \in [R-1,R]$   such that
\begin{equation}\label{bbord}
\int_{\p  K_t  \times\R^k   } \yg |E_\eta|^2 \le  \int_{ K_R\times \R^k  }  \yg  |E_\eta|^2\le CR^d;\end{equation} and then $L \in [1, \sqrt{R}]$ such that 
\begin{equation}\label{bbord2}
\int_{ K_R \times \p [-L, L]^k}\yg |E_\eta|^2 \le R^{-1/2} \int_{K_R \times \R^k }  \yg |E_\eta|^2\le CR^{d-\frac{1}{2}}.
\end{equation}
Let us next  integrate \eqref{delp} over $ K_t \times [-L,L]  $ and use Green's theorem to find
\begin{equation}
\int_{  K_t\times [-L,L]^k} \sum_{p \in \Lambda } N_p \delta_p^{(\eta)}  - m |K_t| =  - \frac{1}{\c}\int_{\p  (K_t\times [-L,L]) } \yg E_\eta\cdot \vec{\nu},
\end{equation}where $\vec{\nu}$ denotes the outer unit normal.
Using the Cauchy-Schwarz inequality and \eqref{bbord}--\eqref{bbord2}, we deduce that
\begin{multline}\label{feg}
\left|\int_{  K_t\times [-L,L]^k} 
\sum_{p \in \Lambda } N_p \delta_p^{(\eta)}  - m |K_t| \right|\le  C
\left( \int_{\p K_t\times [-L,L]^k}     \yg \right)^{1/2}  \left( \int_{\p  K_t\times [-L,L]^k} \yg |E_\eta|^2\right)^{1/2}\\
+ C\left( \int_{ K_t\times \p [-L,L]^k }     \yg \right)^{1/2}  \left( \int_{  K_t\times \p [-L,L]^k } \yg |E_\eta|^2\right)^{1/2}   \\
\le    CL^{\frac{\gamma+1}{2} }  R^{\frac{d-1}{2}} R^{\frac{d}{2}} + CL^{\frac{\gamma}{2}} R^{\frac{d}{2}}  R^{\frac{d}{2}-\frac{1}{4} }=o(R^d) \end{multline}
as $R\to \infty$, in view of the bound on $L$ and the fact that $\gamma<1$.
Since $\eta< 1$, by definition of $\nu$ and since the $\delta_p^{(\eta)}$ are supported in $B(p,\eta)$, we have $\nu(K_{R-2})\le \int_{ K_t \times [-L,L]^k} \sum_{p \in \Lambda } N_p \delta_p^{(\eta)} \le \nu(K_{R+1}) $.  The result thus follows from \eqref{feg}, after  dividing  by $R^d$ and letting $R \to \infty$.

\end{proof}

\begin{lem}[Controlling the  discrepancy]\label{lemdiscrepance}
Assume  $E$ satisfies  a relation of the form $$-\div(\yg E)= \c \Big(\sum_{p\in \Lambda} N_p \delta_p- m(x)\drd \Big)   $$
in some subset $U\subset \R^{d+k}$ for some $m\in L^\infty(U)$, and let $E_\eta$ be associated as in \eqref{defeeta}. Then for any $0<\eta<1$, $L>2$ and $a\in \R^d \times \{0\}$, denoting $\tilde B_L(a)= B_L(a)\times [-L/2,L/2]$, if $\tilde B_{2L}(a)\subset U$ we have \footnote{with the convention $s=0$ in the cases  \eqref{wlog}--\eqref{wlog2d}}  
\begin{equation}\label{estdisc}
 \int_{\tilde B_{2L}(a)} \yg |E_\eta|^2\ge C \frac{D(a, L)^2}{L^s}\min \left(1, \frac{D(a,L)}{L^d}\right) ,
\end{equation}
for some $C$ depending only on $d,s$ and $\|m\|_{L^\infty}$,
where $D(a,L)$ denotes the discrepancy
\[
\sum_{p\in \Lambda\cap B_L(a)} N_p - m|B_L(a)|.
\]
\end{lem}
\begin{proof}
The proof follows \cite[Lemma 3.8]{rs}. For simplicity of notation we denote $\tilde B_t(a)=\tilde B_t$. We first consider the case that $D:= D(a, L)>0$. We first  note that if 
\begin{equation}\label{defLe}
L+\eta\le t \le T:= \min \left(  2L, \Big((L+\eta)^d + \frac{D}{2C}\Big)^{\frac1d}\right)\end{equation}
with $C$ well-chosen,  we have 
\begin{eqnarray*}
 -\int_{\p \tilde B_t}\yg E_\eta\cdot\vec{\nu} &=&-\int_{\tilde B_t}\div(\yg E_\eta) = \c \int_{\tilde B_t}\Big(\sum_{p\in\Lambda}N_p\delta_p^{(\eta)} - m(x)\drd\Big)\\
 &\ge &c_{d,s}\left(D  -   \int_{B_t\backslash B_L } m(x)\, dx\right)\ge \c D- C \left( t^{d}-L^{d}\right)\ge \frac{\c}{2} D   ,
\end{eqnarray*} if we choose the same $C$ in \eqref{defLe}. 
By the Cauchy-Schwarz inequality there holds (with the convention $s=0$ if \eqref{wlog}--\eqref{wlog2d}) 
\begin{eqnarray*}
 \int_{\tilde B_{2L}}\yg|E_\eta|^2&\ge&\int_{L+\eta}^{T}\left(\int_{\p \tilde K_t}\yg\right)^{-1}\left(\int_{\p \tilde B_t}\yg E_\eta\cdot\nu\right)^2dt\\
 &=&C D^2 \int_{L+\eta}^{T}t^{-(\gamma+d+k-1)}  \, dt=  CD^2 \left( \g( L+\eta)-\g(T)\right)   \end{eqnarray*}
using the previous estimate and \eqref{gs}. Inserting the definition of $T$ and rearranging terms, one easily checks that we obtain \eqref{estdisc}.
There remains to treat the case where $D\le 0$.  This time, we let 
$$T \le  t \le L-\eta, \qquad    T:= \Big( (L-\eta)^d - \frac{D}{2C}\Big)^{\frac1d}$$
and if $C$ is  well-chosen  we have \begin{eqnarray*}
 -\int_{\p \tilde B_t}\yg E_\eta\cdot\vec{\nu} &=&-\int_{\tilde B_t}\div(\yg E_\eta) = \c \int_{\tilde B_t}\Big(\sum_{p\in\Lambda}N_p\delta_p^{(\eta)} - m(x)\drd\Big)\\
 &\le &c_{d,s}\left(D  -   \int_{B_L\backslash B_t } m(x)\, dx\right)\le \frac{\c}{2} D   ,
\end{eqnarray*}and the rest of the proof is analogous, integrating from $T$ to $L-\eta$. 

\end{proof}

\subsection{Monotonicity of $\W_\eta$}

Next, we prove that applying the $\eta $ truncation to the energy is essentially  monotone  in $\eta$ (this is natural if we recall that it is almost truncating the kernel at level $\eta$). 

\begin{lem}\label{prodecr}
For any $x_1, \dots, x_n \in \R^d$, and any $1>\eta>\alpha>0$ we have\footnote{with the convention $s=0$ in cases \eqref{wlog}--\eqref{wlog2d}}
\begin{multline*}
-C n \|\muv\|_{L^\infty} \eta^{\frac{d-s}{2}}\le \left(\int_{\R^{d+k}}\yg |\nab h'_{n,\alpha}|^2 - n \c\g(\alpha) \right) -
\left(\int_{\R^{d+k}}\yg |\nab h'_{n,\eta}|^2 - n\c  \g(\eta) \right)
\\ \le  C n \|\muv\|_{L^\infty} \eta^{\frac{d-s}{2}} +
\c \sum_{i\neq j,  |x_i-x_j|\le 2\eta} \min(\g(\alpha), \g(|x_i-x_j|-\alpha)) - \g(|x_i-x_j|+\eta).
\end{multline*}where  $C$ depends only of $d$ and $s$. 
%Thus 
%$$\lim_{\eta\to 0} \left(\int_{\R^{d+k}}\yg |\nab h'_{n,\eta}|^2 - n\c  \g(\eta) \right)$$
%exists.
\end{lem}

\begin{proof} As in the introduction we use the notation $\g_\eta(x)= \min (\g, \g(\eta) $, and we note that  since $\g_\eta=\g- f_\eta $ (recall \eqref{feta} and \eqref{divf}) we have 
\begin{equation}\label{divge}
-\div (\yg \nab \g_\eta) = \c \delta_0^{(\eta)}.
\end{equation}
We then let $\fae:= f_\alpha-f_\eta$.  We note that  $\fae$  vanishes outside $B(0,\eta)$, and  
 \begin{equation}\label{faee}
\g(\eta)- \g(\alpha)\le \g_\eta- \g_\alpha = \fae\le 0 \end{equation}  and $\fae$ solves (cf. \eqref{divf})
\begin{equation}\label{eqfae}
-\div (\yg \nab \fae)= \c (\delta_0^{(\eta)}- \delta_0^{(\alpha)}).
\end{equation}
In view of
\eqref{defheta}, we have
$\nab \hne'= \nab h_{n, \alpha}' + \sum_{i=1}^n\nab \fae(x-x_i)$
and hence
\begin{multline}
\int_{\R^{d+k}} \yg |\nab \hne'|^2 = \int_{\R^{d+k}} \yg |\nab h_{n, \alpha}'|^2 + \sum_{i, j} \int_{\R^{d+k}} \yg \nab \fae (x-x_i) \cdot \nab \fae (x-x_j) 
\\ +  2 \sum_{i=1}^n \int_{\R^{d+k}} \yg \nab \fae(x-x_i)\cdot \nab h_{n, \alpha}'.\end{multline}
We first examine
\begin{align}\nonumber
&  \sum_{i, j} \int_{\R^{d+k}} \yg \nab \fae (x-x_i) \cdot \nab \fae (x-x_j) 
\\ \label{faeterm}
& = - \sum_{i,j} \int_{\R^{d+k}}  \fae (x-x_i) \div (\yg \nab \fae(x-x_j))= \c\sum_{i,j} \int_{\R^{d+k}} \fae (x-x_i) (  \delta_{x_j}^{(\eta)}- \delta_{x_j}^{(\alpha)}) 
.
\end{align}
Next, 
\begin{multline}\label{fae2}
2 \sum_{i=1}^n \int_{\R^{d+k}} \yg \nab \fae(x-x_i)\cdot \nab h_{n, \alpha}'
 = -2 \sum_{i=1}^n \int_{\R^{d+k}} \fae(x-x_i) \div (\yg \nab h_{n,\alpha}') 
\\=  2
\c  \sum_{i=1}^n\int_{\R^{d+k}} \fae(x-x_i) \Big( \sum_{j=1}^n \delta_{x_j}^{(\alpha)} - \muv'\drd\Big) . 
\end{multline}
These last two equations add up  to give a right-hand side equal to 
\begin{equation}\label{rhs}
\sum_{i\neq j} \c \int_{\R^{d+k}} \fae (x-x_i) (  \delta_{x_j}^{(\alpha)}+ \delta_{x_j}^{(\eta)}) 
- 2\c \sum_{i=1}^n \int\fae(x-x_i) \muv'\drd
+ n\c \int_{\R^{d+k}} \fae ( \delta_{0}^{(\alpha)} + \delta_0^{(\eta)} )\, . \end{equation}
We then note that $\int\fae( \delta_0^{(\alpha)} + \delta_0^{(\eta)} ) = -\int f_\eta \delta_0^{(\alpha)}=-( \g(\alpha)-\g(\eta))  $ by definition of $f_\eta$ and the fact that $\delta_0^{(\alpha)}$ is a measure supported on $\p B(0, \alpha)$ and of mass $1$.
Secondly, we bound  $\int_{\R^{d+k}}\fae(x-x_i) \muv'\drd$ by $$\|\muv\|_{L^\infty} \int_{\R^{d+k}} |f_\eta| \le  C \|\muv\|_{L^\infty} \max(\eta^{d-s}, \eta^d |\log \eta|) $$
according to the cases, as seen in the proof of the splitting formula.
Thirdly, we observe that in view of \eqref{faee}, the first term in \eqref{rhs} is nonpositive, and that only the terms for which $|x_i-x_j|\le 2 \eta$ contribute. We now bound its absolute value using \eqref{faee} and \eqref{divge}
\begin{multline*} 
\left|\sum_{i\neq j} \c \int_{\R^{d+k}} \fae (x-x_i) (  \delta_{x_j}^{(\alpha)}+ \delta_{x_j}^{(\eta)}) \right|\\=-  \sum_{i\neq j} \int_{\R^{d+k}} (\g_\alpha- \g_\eta)(x-x_i)  \div (\yg (\nab \g_\alpha(x-x_j) + \nab \g_\eta(x-x_j) )) \\
= \sum_{i\neq j} \int_{\R^{d+k}}\yg( \nab \g_\alpha -\nab \g_\eta )(x-x_i) \cdot( \nab \g_\alpha + \nab \g_\eta) (x-x_j)\\
= \sum_{i\neq j} \int_{\R^{d+k}}\yg \nab \g_\alpha (x-x_i) \cdot\nab \g_\alpha(x-x_j) - \yg \nab \g_\eta(x-x_i) \cdot   \nab \g_\eta(x-x_j)\end{multline*}
where we noted that the other terms cancel out in the sum when exchanging the roles of $i$ and $j$.
Integrating by parts again using \eqref{divge}, and using the fact that $\delta_x^{(\eta)}$ is a measure of mass $1$ supported on $\p B(x, \eta)$,  we are led to 
\begin{multline*} 
\left|\sum_{i\neq j} \c \int_{\R^{d+k}} \fae (x-x_i) (  \delta_{x_j}^{(\alpha)}+ \delta_{x_j}^{(\eta)}) \right|=  \c
\sum_{i\neq j} \int_{\R^{d+k}}\g_\alpha (x-x_i) \delta_{x_j}^{(\alpha)}   - \g_\eta(x-x_i) \delta_{x_j}^{(\eta)}\\
\le \c \sum_{i\neq j, |x_i-x_j|\le 2\eta} \g_\alpha(|x_i-x_j|-\alpha) - \g_\eta (|x_i-x_j|+\eta)\\
\le \c \sum_{i\neq j,  |x_i-x_j|\le 2\eta} \min(\g(\alpha), \g(|x_i-x_j|-\alpha)) - \g(|x_i-x_j|+\eta).
\end{multline*} where we used the fact that $\g$ and $\g_\eta$ are radial decreasing.

We conclude that 
\begin{multline*}  
- C n \|\muv\|_{L^\infty} \max( \eta^{d-s}, \eta^d |\log \eta|)
\\ \le \(\int \yg |\nab h_{n, \alpha}'|^2  -n\c  \g(\alpha)\) -\( \int_{\R^{d+k}}\yg |\nab \hne'|^2 - n\c \g(\eta)\)\\ \le 
 C n \|\muv\|_{L^\infty} \max( \eta^{d-s}, \eta^d |\log \eta|)\\+ 
 \c \sum_{i\neq j,  |x_i-x_j|\le 2\eta} \min(\g(\alpha), \g(|x_i-x_j|-\alpha)) - \g(|x_i-x_j|+\eta),
\end{multline*}
 and this finishes the proof, noting that in all cases we have $\max(\eta^{d-s}, \eta^d  |\log \eta|) \le \eta^{\frac{d-s}{2}}$ if one takes the convention $s=0$ in the logarithmic cases.

\end{proof}
The next proposition expresses the same fact at the level of the limits.
\begin{pro}\label{Wbb}
Let  $E \in \mathcal{A}_m$. For any $1>\eta>\alpha>0$ such that $\W_\alpha(E) <+\infty$  we have 
$$\W_\eta(E) \le \W_\alpha (E )+ Cm\eta^{\frac{d-s}{2}}
$$
where $C$ depends only on $s$ and $d$; and thus   $\lim_{\eta\to 0}\W_\eta(E)= \W(E)$ always exists. Moreover,
$\W_\eta$ is bounded below on $\mathcal A_m$ by  a constant depending only on $s,d$ and $m$.
\end{pro}
We note that this proves items 1 and 2 in Proposition \ref{prow}.

\begin{proof} 
Let  us consider $E$ satisfying  a relation of the form $$-\div(\yg E)= \c \Big(\sum_{p\in \Lambda} N_p \delta_p- m(x)\drd \Big) \quad \text{in} \ K_R \times \R^k   $$
 for some $m\in L^\infty(K_R)$.  Let $E_\eta$ be associated via \eqref{defeeta}. Assume $1>\eta>\alpha>0$. 
Let $K_R= [-R/2,R/2]^d \times\{0\} $ and $\tilde K_R= [-R/2,R/2]^{d+k}$.
Let $\chi_R$ denote a smooth cutoff function equal to $1$ in $\tilde K_{R-3}$  and vanishing outside $\tilde K_{R-2}$. As in the previous proof we we note that 
$$E_\eta= E_\alpha + \sum_{p\in \Lambda} N_p \nab \fae(x-p)$$
and insert to expand
\begin{multline}
\int_{\R^{d+k}} \chi_R  \yg |E_\eta|^2- \int \chi_R \yg |E_\alpha|^2 
\\=\sum_{p, q\in \Lambda} N_p N_q \int_{\R^{d+k}}\chi_R \yg \nab \fae(x-p)\cdot \nab \fae(x-q) + 2\sum_{p\in \Lambda} N_p \int_{\R^{d+k}} \yg \chi_R \nab \fae(x-p)\cdot E_\alpha.\end{multline}
We separate here between the integral over the set  where $\chi_R\equiv 1$ and the rest, which we will control by 
\begin{multline*}
Error:= \sum_{p, q\in K_{R-1} \backslash  K_{R-4}} N_p N_q \int_{\R^{d+k}} \yg|\nab \fae(x-p)||\nab \fae(x-q)|
\\+ \sum_{p \in K_{R-1} \backslash  K_{R-4}}N_p  \int_{\R^{d+k}}  \yg |E_\alpha| |\nab \fae(x-p)|.\end{multline*}
We will work on controlling $Error$ just below, and for now, using the fact that $\fae $ is supported in $B(0,\eta)$,  we may write 
\begin{equation}\label{errormain}
\int_{\R^{d+k}}\chi_R  \yg |E_\eta|^2- \int_{\R^{d+k}} \chi_R \yg |E_\alpha|^2 \le Error+ Main\, , 
\end{equation} where \begin{multline*}
Main:=
\sum_{p, q \in\Lambda\cap K_{R-4}}N_pN_q  \int_{\R^{d+k}} \yg \nab \fae(x-p)\cdot \nab \fae(x-q) \\+ 2\sum_{p\in\Lambda\cap K_{R-4s} }N_p \int_{\R^{d+k}} \yg  \nab \fae(x-p)\cdot \nab h_\alpha.\end{multline*}
We may now evaluate $Main$ just as we did in the proof of Lemma \ref{prodecr}:
\begin{multline*}
Main= 
- \sum_{p, q\in K_{R-4}\cap \Lambda}N_p N_q  \int  \fae(x-p)\div (\yg \nab \fae(x-q))  \\
-  2\sum_{p\in K_{R-4}\cap \Lambda}  N_p \int \fae(x-p)\div (\yg E_\alpha)\\
= \c \sum_{p,q\in K_{R-4}\cap \Lambda} N_pN_q\int  \fae (x-p) ( \delta_q^{(\eta)}- \delta_q^{(\alpha)}) 
+ 2 \c \sum_{p,q\in K_{R-4}\cap \Lambda}N_p N_q  \int \fae(x-p) (\delta_q^{(\alpha)}- m(x)\drd)\end{multline*} thus, bounding $\fae$ from below by \eqref{faee}, we get  
\begin{multline}
- C \|m\|_{L^\infty} \eta^{\frac{d-s}{2}} + 2\c(\g(\eta)- \g(\alpha)) \sum_{p\neq q, |p-q|\le 2\eta} N_p N_q
\\ \le 
Main - \c \sum_{p \in\Lambda\cap K_{R-4}} N_p  (\g(\eta)-\g(\alpha))\le  C \|m\|_{L^\infty} \eta^{\frac{d-s}{2}}.
\end{multline}

Next, to control $Error$, we partition $K_{R-1}\backslash K_{R-4}$ into disjoint cubes $\mathcal C_j$  of sidelength centered at points $y_j$ and we denote by $\mathcal N_j =\sum_{p  \in \Lambda \cap \mathcal C_j} N_p$.
By Lemma \ref{lemdiscrepance},     we have that 
$\mathcal N_j^2  \le C + C e_j$ where $$e_j:= \int_{B_2(y_j)\times (-R/2, R/2)^k} \yg |E_\alpha|^2 .$$
Using that the overlap of the $B_2(y_j)$ is bounded, we may write 
$$ \sum_j \mathcal N_j^2  \le C R^{d-1} + \sum_j e_j \le CR^{d-1} + \int_{\tilde K_{R}\backslash \tilde K_{R-5}} \yg |E_\alpha|^2 .$$
We then may deduce, by separating the contributions in each $\mathcal C_j$  and using the Cauchy-Schwarz inequality and 
$\int |\nab f_{\alpha, \eta} |^2  \le C \g(\alpha)$, 
that
\begin{multline}
Error\le  C \g(\alpha) \sum_j \mathcal N_j^2 + C \sum_j \g( \alpha)^{1/2}\mathcal  N_j e_j^{1/2}
 \le C \g( \alpha) \left( \sum_j \mathcal  N_j^2 + \sum_j e_j\right) \\
\le C \g( \alpha)      \left( R^{d-1} + \int_{\tilde K_{R}\backslash \tilde K_{R-5}} \yg |E_\alpha|^2 \right)  .\end{multline}
Returning to \eqref{errormain}  we have found that 
\begin{multline}\label{estm} 
- C\|m \|_{L^\infty} \eta^{\frac{d-s}{2}}  \sum_{p \in\Lambda\cap K_{R-4}} N_p - C \g( \alpha)     \left( R^{d-1} + \int_{\tilde K_{R}\backslash \tilde K_{R-5}} \yg |E_\alpha|^2 \right)\\
\le 
\(\int   \chi_R \yg |E_\alpha|^2 - \c \sum_{p \in\Lambda\cap K_{R-4}} N_p  \g(\alpha)\)-\( \int \chi_R  \yg |E_\eta|^2 - \c \sum_{p \in\Lambda\cap K_{R-4}} N_p  \g(\eta)\)\le 
\\
 2\c(\g(\alpha)- \g(\eta)) \sum_{p\neq q, |p-q|\le 2\eta} N_p N_q
+C\|m \|_{L^\infty} \eta^{\frac{d-s}{2}}  \sum_{p \in\Lambda\cap K_{R-4}} N_p + C \g( \alpha)     \left( R^{d-1} + \int_{\tilde K_{R}\backslash \tilde K_{R-5}} \yg |E_\alpha|^2 \right)
\end{multline}and it easily follows that 
\begin{multline}      
0 \le\( \int_{K_R\times \R^k }\yg |E_\alpha|^2 -\c \g(\alpha) \sum_{p\in \Lambda \cap K_R}  N_p \right)- 
\left( \int_{K_R\times \R^k }\yg |E_\eta|^2 -   \c\g(\eta) \sum_{p\in \Lambda \cap K_R}  N_p \right)  +error \\ \le  2\c(\g(\alpha)- \g(\eta)) \sum_{p\neq q, |p-q|\le 2\eta} N_p N_q\end{multline}with 
\begin{multline*}|error|\le
 C\eta^{\frac{d-s}{2}}\|m\|_{L^\infty}  \sum_{p \in \Lambda \cap K_{R } }N_p + C \g(\alpha)  \sum_{p\in \Lambda \cap (K_R\backslash K_{R-5}   )} N_p  \\+C (1+ \g(\alpha))  \int_{(K_R\backslash K_{R-5})\times \R^k} \yg |E_\eta|^2 +\yg |E_\alpha|^2,
\end{multline*} where $C$ depends only on $s$ and $d$. 

Let us now specialize to $E\in \mathcal{A}_m$.
In view of Lemma \ref{limdensity} we have that $ \lim_{R\to \infty} R^{-d} \sum_{p\in \Lambda \cap K_R} N_p =m$ and 
$\lim_{R \to \infty}R^{-d} \sum_{p\in  \Lambda \cap (K_R\backslash K_{R-5})}N_p =0$.
In addition,  since $\W_\alpha(E)  <\infty$ and by definition of $\W_\alpha$,  we must have $\lim_{R\to \infty} R^{-d} \int_{\tilde K_{R}\backslash \tilde K_{R-5}} \yg |E_\alpha|^2=0$. We deduce that 
$$- C m \eta^{\frac{d-s}{2}}\le  \W_\alpha(E)- \W_\eta(E)  \le    Cm \eta^{\frac{d-s}{2}} + 2\c  (\g(\alpha)- \g(\eta))  \limsup_{R\to \infty}\frac{1}{R^d} \sum_{p\neq q, |p-q|\le 2\eta} N_p N_q  .$$
It then immediately follows that $\W_\eta$ has a limit (finite or infinite) as $\eta \to 0$, and that $\W_\eta(E)$ is bounded below by, say, $\W_{1/2}(E) - Cm$, which in view of its definition is obviously bounded below by $-\c m- Cm$. 
\end{proof}

\section{The periodic case}\label{sec3}

We consider now the case of periodic configurations of charges $\Lambda$ an prove Proposition \ref{periodic} and Theorem \ref{abrikolat}. % We also compare the functional $\W$ to the functional $W$ of the following definition:
% \begin{defi}\label{defw}
%  Let $d,s,k,\gamma$ be as above, let $x_1,\ldots,x_n\in\mathbb R^d\times\{0\}\subset\mathbb R^{d+k}$. For $\chi\in C^0_0(\mathbb R^{d+k})$ and $E:\mathbb R^{d+k}\to\mathbb R^{d+k}$ satisfying
%  \[
% -\div\left(\yg E\right)=c_{d,s}\left(\sum_{i=1}^n\delta_{x_i} - n\mu_0\right)\ ,
%  \]
% we define
% \[
%  W(E,\chi):=\lim_{\eta\to 0}\left(\int_{\mathbb R^{d+k}\setminus\cup B(x_i,\eta)}\yg\chi|E|^2  - c_pg(\eta)\sum_i\chi(x_i)\right)\ .
% \]
% \end{defi}
\subsection{Proof of  Proposition \ref{periodic}}
 We start by proving the second point, by a modification of the calculations of the proof of Proposition \ref{splitting}. By using the periodicity we see that we just have to compute
 \begin{equation}\label{wetaper}
  \lim_{\eta\to 0}\left[\int_{\mathbb T\times\mathbb R^k}\yg|\nab H_\eta|^2 -\c N\g(\eta)\right]\ .
 \end{equation}
We may take $\eta\le\frac{1}{2}\min_{i\neq j}|a_i-a_j|$. We remark that
\[
 H_\eta(x) = \sum_{i=1}^N \left(c_{d,s}\bar G(x-a_i) - f_\eta(x-a_i)\right)\ ,
\]where $\bar G$ is the  solution in $\mathbb{T}\times \R^k$ of 
\begin{equation} 
\label{defG1}-\div (\yg \nab \bar G)= \delta_0 - \frac{1}{|\mathbb{T}|}\drd,
\end{equation} 
with $\int_{\mathbb T\times \R^k} \bar G \drd=0$.  By the extension representation, one may check that the trace $G$ on $\R^d$ of $\bar G$ solves \eqref{defG}. 
Inserting into  \eqref{wetaper}, we  can  compute via an integration by parts:
\begin{eqnarray*}
\int_{\mathbb T\times\mathbb R^k}\yg|\nab H_\eta|^2 &=&- \int_{\mathbb T\times\mathbb R^k}H_\eta \div(\yg\nab H_\eta)  \\
 &=&c_{d,s}\int\left(\sum_{i=1}^N (c_{d,s}H(x-a_i) - f_\eta(x-a_i))\right) \Big(\sum_{j=1}^N\delta_{a_j}^{(\eta)} - \drd\Big)dx\ .
\end{eqnarray*}
Now we may continue exactly like in the proof of Proposition \ref{splitting}, with $\sum_{i=1}^n \bar G(x-a_i)$ in the place of $h_n(x)$ and with $1$ in the place of $\mu_V$. As an analogue of $h_n^i(x)$ we thus have
\[
c_{d,s}\sum_{j\neq i} \bar G(x-a_j) + N\lim_{|x|\to 0}\left(c_{d,s}\bar G(x)-\g(x)\right),
\]
therefore, since $\bar G $ has zero average, we obtain the following analogue of \eqref{Wsplit11}, where we write $G$ instead of $\bar G$ since all the quantities can be computed on the trace:
\[
 \lim_{\eta\to 0}\frac{1}{c_{d,s}^2}\left[\int_{\mathbb T\times\mathbb R^k}\yg|\nab H_\eta|^2 - N\c\g(\eta)\right]=\sum_{i\neq j} G(a_i- a_j) + N\lim_{x\to 0}\left(G(x) - \frac{\g(x)}{c_{s,d}}\right).
\]
This is equivalent to \eqref{minper}. Now assume that  $E=\nab h$ is another periodic vector field (in the first $d$ coordinates) which is compatible with the points  $a_1,\ldots,a_N$. We have that $\nab u=\nab h-\nab H$ solves $\div (\yg \nab u)=0$ and  is periodic. 
We may then write 
$$ \int_{\mathbb T \times \R^k} \yg |(\nab u + \nab H)_\eta|^2= \int_{\mathbb T \times \R^k} \yg |\nab u|^2 +\int_{\mathbb T \times \R^k} |\nab H_\eta|^2 $$
where the cross-term has vanished by using the periodicity of $H_\eta$ and of $\nab u$, and $\div (\yg \nab u)=0$.
It is then straightforward to deduce that $\W(\nab h) \ge \W (\nab H)$, which finishes the proof of item 2. 

We now turn to the proof of item 1. 
Suppose that the points $a_1,\ldots, a_N$ are not distinct. Without loss of generality, we have $k$ points $a_1,\ldots, a_k$ with multiplicities $N_1,\ldots, N_k$ with $\sum_{i=1}^k N_i=N, \max_i N_i\ge 2$. The left hand side of \eqref{estdistinct} in our periodic case then contains a first term equal to $c_{d,s}\g(\eta)\sum_{i=1}^kN_i^2$, which is not  cancelled by the negative $c_{d,s}N\g(\eta)$ term any more. The remaining terms in \eqref{estdistinct} are respectively positive, zero, and $o_\eta(1)$, therefore we have
\[
 \int_{\mathbb T\times\mathbb R}\yg|\nab H_\eta|^2=c_{d,s}\g(\eta)\left(\sum_{i=1}^kN_i^2 - N\right) + o_\eta(1)\to+\infty\ ,\quad (\eta\to 0)\ ,
\]
as desired.

For the third item, we  generalize the proof of  \cite[Prop. 2.10]{BS}. We solve the equation \eqref{defG} satisfied by $G$ by Fourier transform. 
We choose the following normalization for Fourier transforms and series:
$$\widehat{f}(\xi)= \int f(x) e^{-2i\pi x \cdot \xi }\, dx$$
$$c_k(f)= \int_{\mathbb T _N} f(x) e^{-\frac{2i\pi kx}{N}} \, dx.$$
Then  the Fourier inversion formula is $f(x) =\int \widehat{f}(\xi) e^{2i\pi x\cdot \xi} \, d\xi$ and
$$f(x)= \frac{1}{N} \sum_{k \in \mz} c_k(f) e^{\frac{2i\pi k}{N} x} .$$
If $G$ solves \eqref{defG} then $\widehat{G}$ has to satisfy
$$
\(\frac{2\pi m}{N}\)^{2\alpha}\widehat{G}(m)= \widehat{\delta_{(0,0)}} - \delta_m^0\frac{1}{N}\ , 
$$
with $\delta_m^0$ by definition equal to $1$ if $m =0$ and $0$ otherwise.
Combining these facts, we obtain
$$ \widehat{G}(m) = \frac{1 - \delta_m^0}{\(\frac{2\pi |m|}{N}\)^{2\alpha}}
\qquad for\ m\neq 0.$$
The undetermination of $\widehat{G(0)}$ corresponds to the fact that $G$ is only determined by \eqref{defG} up to a constant.
By Fourier inversion with the above normalization, it follows that
\begin{equation}
\label{Gcos}
G(x)= 2\frac{N^{2\alpha-1}}{(2\pi)^{2\alpha}}\sum_{k=1}^\infty \frac{\cos (2\pi\frac{ k}{N}x) }{k^{2\alpha}}+c.
\end{equation}
We see that the condition $\int_{\mathbb T_N}G=0$ gives $c=0$. We next use  the following formula found in \cite[p. 726]{prud}:
$$
\sum_{k=1}^\infty \frac{\cos ky}{k^a} = \frac{1}{\Gamma(a)}\int_0^\infty \frac{t^{a-1} (e^t \cos y  - 1) }{1-2e^t \cos y + e^{2t}} \, dt\ .
$$
Applying  it to  $a=2\alpha$ and $ y=\frac{2\pi}{N}x$ we obtain
$$
 G(x) = 2\frac{N^{2\alpha-1}}{(2\pi)^{2\alpha}\Gamma(2\alpha)}\int_0^\infty \frac{t^{2\alpha-1} \(e^t \cos \(\frac{2\pi}{N}x\)  - 1\) }{1-2e^t \cos \(\frac{2\pi}{N}x\) + e^{2t}} \, dt\ ,
$$
which is the desired formula \eqref{G1d}.

\subsection{The two-dimensional case}
In the case $N=1$ of Proposition \ref{periodic}, the first term in the expression of $\W(\nab H)$ does not appear and the lattice $\Lambda$ is such that $\mathbb T$ has volume $1$. Then the minimum $\W$-energy corresponding to the configuration $\Lambda$ can be expressed as follows, using the Fourier series expression for the function $G=G_\Lambda$ obtained like in the above proof:
\begin{equation}\label{glambda}
 \W(\Lambda)=c_{d,s}\lim_{x\to 0}\left(c_{d,s}G_\Lambda(x) - \g(x) \right)=c_{d,s}\lim_{x\to 0}\left(c_{d,s}\sum_{q\in\Lambda^*\setminus \{0\}}\frac{e^{2i\pi\,  q\cdot x}}{|2\pi q|^{2\alpha}} -\g(x) \right)\ ,
\end{equation}where $\Lambda^*$ denotes the dual lattice to $\Lambda$.
In dimension $d=2$ we may use like in \cite{gl13} the number theory results from \cite{cassels,ennolacassels,rankin,montgomery} to characterize the lattice of smallest $\W$-energy, allowing to prove Theorem \ref{abrikolat}.
% From the formula \eqref{minper} with $N=1, d=2$, which is the same as \eqref{glambda}, we have for a volume-$1$ lattice $\Lambda$ that
%\begin{equation}\label{diffglambda}
%\W(\Lambda):=c_{d,s}^2\lim_{x\to 0}\left(G_\Lambda(x) - \frac{g(x)}{c_{d,s}}\right)\ 
%\%end{equation}
Recalling that $\g(x)$ is the solution to $(-\Delta)^\alpha \g =\c \delta_0$ in $\mathbb R^2$, we have that $U_\Lambda:=G_\Lambda- \c^{-1}\g(x)$ is regular near $0$ and we have $\W(\Lambda)=\c^2U_\Lambda(0)$. 

Instead of the function $\g(x)$ we will use a function $\g_*(x)$ which has the following properties (where $\mathcal F$ denotes the Fourier transform):
\begin{equation}\label{gstar}
\left\{\begin{array}{l} 
        \mathcal F \g_*(\xi) = \mathcal F \g(\xi)\quad\text{outside }B_1,\text{ and }\mathcal F \g_*\text{ is finite in  }B_1\,\\
        \g(x) - \g_*(x)\text{ is }C^\infty\, .
       \end{array}\right.
\end{equation}
To construct $\g_*$ it suffices to note that both $\g(x)$ and $\mathcal F \g(\xi)$ are in $L^1_{\mathrm{loc}}(\R^2 \backslash \{0\})$ for our range of exponents, since $s,2\alpha\in(0,2)$. It follows that for $\eta$ smooth with support in $B_1$, equal to $1$ in a neighborhood of the origin, $\eta(\xi)\mathcal F \g(\xi)$ is a compactly supported distribution therefore its inverse Fourier transform is $C^\infty$ (see \cite[2.1.3]{hormander1}). We then define $\g_*(x):= \g(x)-\mathcal F^{-1}(\eta \mathcal F \g)(x)$ and the above properties are easily verified.\\

By replacing $\g_*$ for $\g$ in the definition \eqref{glambda} of $\W$, we define the functional
\begin{equation}\label{wstar}
  \W_*(\Lambda):=\c^2\lim_{x\to 0}\left(G_\Lambda(x) - \frac{\g_*(x)}{\c}\right)=\c^2 \lim_{x\to 0}\left(\sum_{q\in\Lambda^*\setminus \{0\}}\frac{e^{2i\pi\,  q\cdot x}}{|2\pi q|^{2\alpha}} -\frac{\g_*(x)}{\c}\right)\ .
\end{equation}
Note that $\W_*$ and $\W$ differ by a constant thus have the same minimizers, since $g(x)-g_*(x)$ is a regular function near zero. It thus suffices to prove the result for $\W_*$.

\subsection{Minimum of the Epstein zeta function in $2$ dimensions}
 We now recall a result from \cite{cassels} regarding the analytic continuation of the Epstein zeta function. 
Without loss of generality, we may just consider lattices of the canonical form 
\begin{equation}\label{lambdatau}
 \Lambda_{\tau}:= y^{-1/2}((x,y)\mathbb Z + (1,0)\mathbb Z)\text{, with }\tau=x+iy, y>0, x\in[0,1)\, .
\end{equation}
Let us denote the corresponding Epstein zeta function with exponent $2\alpha$ by
\begin{equation}\label{realpha1}
 Z_\tau (\alpha):=\sum_{q\in \Lambda_{\tau}\setminus\{0\}}\frac{1}{|q|^{2\alpha}}= \sum_{(m,n)\in\mathbb Z^2\setminus \{0\}}\frac{y^{\alpha}}{((mx +n)^2+ m^2y^2)^\alpha}\ .
\end{equation}
This series converges only for $\op{Re}(\alpha)>1$. We note that $Z_\tau$ is formally periodic in $x$ of period $1$. The following formula  is stated in \cite{chowlaselberg}, proved in \cite{batemangrosswald}, and obtained by manipulating the Fourier series in $x$ of $Z$:
\begin{equation}\label{selchow}
Z_\tau (\alpha) = 2y^\alpha \zeta(2\alpha) + 2 y^{1-\alpha}\zeta(2\alpha -1) \frac{\Gamma(\tfrac{1}{2})\Gamma(\alpha-\tfrac{1}{2})}{\Gamma(\alpha)} + Q(x,y,\alpha)\ ,
\end{equation}
where $\zeta(u):=\sum_{k=1}^\infty \frac{1}{k^u}$ is the Riemann zeta function,
\begin{equation*}
Q(x,y,\alpha):=\frac{8\pi^\alpha y^{-1/2}}{\Gamma(\alpha)}\sum_{r=1}^\infty r^{\alpha-\frac{1}{2}}\sigma_{1-2\alpha}(r) K_{\alpha - \frac{1}{2}}(2\pi\,ry)\cos(2\pi\,rx)\ ,
\end{equation*}
and $\sigma_\beta(k)=\sum_{d|k}d^\beta$ and $K_\nu(z)=\int_0^\infty e^{-z\cosh t}\cosh (\nu t) dt$ is the so-called modified Bessel function on the second kind, which for positive real $\nu, z$ decays exponentially. However if we consider the expression for $Z_\tau$ we see that the terms involving the Riemann zeta function again converge only for $\op{Re}(\alpha)>1$, i.e. precisely outside our interest range. But we may use the functional equation 
\[
 \zeta(s)=2^s\pi^{s-1}\sin\left(\frac{\pi s}{2}\right)\Gamma(1-s)\zeta(1-s)
\]
to extend the Riemann zeta function to $\op{Re}(\alpha)<1$. The formula \eqref{selchow} then gives the analytic continuation of $Z_\tau(\alpha)$ to a meromorphic function on all $\alpha\in\mathbb C\setminus\{1\}$, which has a pole of residue $\pi$ at $1$. Note that the value given by formula \eqref{realpha1}
% \cm{remplacer par version avec reseaux en dim. generale} 
coincides with the one given by \eqref{selchow} as long as it is defined, by uniqueness of the analytic continuation and since there are no branch points. The useful point to note is that for $(x',y')\neq (x,y)$ the difference $Z_\tau (\alpha) - Z_{\tau'}(\alpha)$ extends by continuity in $\alpha$ over $1$, the two poles cancelling each other.

We are then able to  use the following result
%for dimension $2$ and of \cite{sarnakstormbergsson} for dimensions $4,8,24$:
\begin{theo}[\cite{cassels,ennolacassels}]\label{casselsth}
 Let $\alpha\in \mathbb C, \op{Re}(\alpha)>0$. Under the same assumptions and notations on $\tau,\Lambda_\tau$ as above, if $Z_\tau (\alpha)$ is the analytic continuation in $\alpha$ as above then for all $x\in(0,1)$ and $y>0$ there holds
 \begin{equation}\label{zetaser}
  Z_\tau(\alpha) - Z_{e^{i\pi/3}}(\alpha)\ge 0\ ,
 \end{equation}
with strict inequality for $\tau \neq e^{i\pi/3}$.
\end{theo}
The lattice $\Lambda_{tri}$ corresponding to $\tau = e^{i\pi/3}$ or  $(x,y)=(1/2,\sqrt3/2)$ is the triangular lattice. 
\subsection{Relation between $\W_*$ and $Z_\tau$ and proof of Theorem \ref{abrikolat}}
Theorem \ref{abrikolat} will follow from Theorem \ref{casselsth} once we prove the following proposition. 
\begin{pro}\label{unifconv}
Let $d=2, \alpha=\frac{d-s}{2}\in(0,1)$ and assume that $\Lambda_\tau \subset\mathbb R^2$ is a unit volume lattice. There holds 
\begin{equation}\label{wg}
 \W_*(\Lambda_\tau)- \W_*(\Lambda_{tri})=\frac{\c^2}{(2\pi)^{2\alpha}}\left(Z_\tau (\alpha) -Z_{e^{i\pi/3}} (\alpha)\right)\ .
\end{equation}
 \end{pro}
\begin{proof}
The left-hand side and right-hand side of \eqref{wg} agree for $\alpha=\frac{d-s}{2}>1$, because in view of \eqref{wstar} the following holds
\begin{eqnarray*}
 \W_*(\Lambda_\tau)-\W_*(\Lambda_{tri})&=&c_{d,s}^2\lim_{x\to 0}\left(\sum_{q\in\Lambda_\tau^*\setminus \{0\}}\frac{e^{2i\pi\,  q\cdot x}}{|2\pi q|^{2\alpha}} - \sum_{q\in\Lambda_{tri}^*\setminus \{0\}}\frac{e^{2i\pi\,  q\cdot x}}{|2\pi q|^{2\alpha}}\right)\\
 &=&\c^2\left(\sum_{q\in\Lambda^*\setminus \{0\}}\frac{1}{|2\pi q|^{2\alpha}} - \sum_{q\in\Lambda_{tri}^*\setminus \{0\}}\frac{1}{|2\pi q|^{2\alpha}}\right)\\
 &=&\frac{\c^2}{(2\pi)^{2\alpha}}\left(Z_\tau(\alpha)- Z_{e^{i\pi/3}} (\alpha)\right)\, .
\end{eqnarray*}
This is justified since the above sums giving the traditional definition of $ Z(\alpha)$ converge absolutely. For the same reason, the above formula for $\W_*(\Lambda)-\W_*(\Lambda_{tri})$ can be extended to the complex half-plane $\op{Re}(\alpha)>1$. We also have that the right hand side of \eqref{wg} is analytic in $\alpha$ for $\alpha\neq 1$, thus we only need to show that the left hand side is analytic in $\alpha$ too. To prove this we use a smoothed sum method, via a very basic Euler-Maclaurin type error estimate valid in two dimensions.

To that aim, let us consider a positive Schwartz function $\eta:\mathbb R^2\to [0,1]$ such that $\eta=1$ on $B_{1/2}$ and $\eta=0$ outside $B_1$. Then the inverse Fourier transform $\varphi$ of $\eta$ is a Schwartz function with integral equal to $1$ and the functions $\varphi_N(x):=N^2\varphi(Nx)=(\mathcal F^{-1}\eta)(\xi/N)$ approximate a Dirac mass at the origin. The following holds:
\begin{eqnarray*}
 \lim_{x\to 0}\left(G_\Lambda(x) - c_{d,s}^{-1}\g_*(x)\right)&=&\lim_{N\to\infty}\int_{\mathbb R^2} \varphi_N(x)(G_\Lambda(x)-\c^{-1}g_*(x))dx\\
 &=&\lim_{N\to\infty}\int_{\mathbb R^2} \eta(\xi/N)(\mathcal F G_\Lambda(\xi) - \c^{-1}\mathcal F g_*(\xi))d\xi\ .
\end{eqnarray*}
We may write the integral of the last line as a sum over the Voronoi cells $K_p:=\{x\in\mathbb R^2:\,|x-p|=\min_{p'\in \Lambda^*}|x-p'|\}$ for $p\in\Lambda^*$ and use \eqref{gstar}, to get
\begin{multline*}
 \int_{\mathbb R^2} \eta(\xi/N)(\mathcal F G_\Lambda(\xi) - \c^{-1}\mathcal F \g_*(\xi))d\xi =\\
  \sum_{p\in(\Lambda^*\setminus\{0\})\cap B_1}\frac{\eta(p/N)}{|2\pi p|^{2\alpha}}-\c^{-1} \sum_{p\in \Lambda^*\cap B_1}\int_{K_p}\eta(\xi/N)\mathcal F \g_*(\xi)d\xi + \sum_{p\in\Lambda^*\setminus B_1}\left(\frac{\eta(p/N)}{|2\pi p|^{2\alpha}} - \int_{K_p}\frac{\eta(\xi/N)}{|2\pi \xi|^{2\alpha}}d\xi\right)\\
 :=I_1(\Lambda, \alpha) + I_2(\Lambda,\alpha)\, .
\end{multline*}
We note that the term $I_1(\Lambda,\alpha)$ is just the sum of a finite (depending on $\Lambda=\Lambda_\tau^*$ but locally bounded with respect to $\tau$) number of terms, in particular it is analytic in $\alpha$ and converges uniformly in $N$ for $\tau$ in a fixed compact set and $\alpha$ in a fixed compact subset within $\op{Re}(\alpha)>0$. Therefore it suffices to focus on the term $I_2(\Lambda,\alpha)$. The terms in the sum defining it are of the form
\[
M_p^N:=\int_{K_p}\left(\frac{\eta(p/N)}{|2\pi p|^{2\alpha}}-\frac{\eta(\xi/N)}{|2\pi\xi|^{2\alpha}}\right)d\xi ,\ p\in\Lambda^*\setminus B_1\ .
\]
For real $\alpha>0$, define  $f_N(\xi):=|2\pi\xi|^{-2\alpha}\eta(\xi/N) $, $\vp_N= f_{N}-f_{N/2}$ which by choice of $\eta$ is supported in $B_N\setminus B_{N/4}$. Its Taylor expansion is
\[
  \vp_N(p) - \vp_N(\xi) + D\vp_N(p)\cdot(\xi-p) = E_N(\xi) \]
with
\begin{equation}\label{taylorf}
\|E_N\|_{L^\infty(K_p)}\le C\|D^2\vp_N\|_{L^\infty(K_p)}\max_{\xi \in K_p} |\xi-p|^2= C_\Lambda\|D^2\vp_N\|_{L^\infty(K_p)}\, ,
\end{equation} and by a direct computation we have
\begin{equation}\label{d2f}
 \|D^2\vp_N\|_{L^\infty(K_p)}\le C_{\eta,\alpha,\Lambda}\left(\frac{1}{N^2|p|^{2\alpha}} + \frac{1}{N|p|^{2\alpha+1}} + \frac{1}{|p|^{2\alpha+2}}\right)\, .
\end{equation}

Choosing now $N=2^h$, we may  write $f_N= f_{2^{k_\Lambda}}+  \sum_{k=k_\Lambda+1}^{h} f_{2^k}-f_{2^{k-1}}= f_{2^{k_\Lambda}}+ \sum_{k=k_\Lambda+1}^h  \vp_{2^k} $ and thus 
\begin{equation}\label{kptaylorerror}
 |M_p^N|\le \int_{K_p}|f_{2^{k_\Lambda}}(\xi)-f_{2^{k_\Lambda}} (p)|  + C_\Lambda  \sum_{k=k_\Lambda+1}^h   \int_{K_p}|E_{2^k}(\xi)|d\xi\ ,
\end{equation}where we used the fact that since  the domain $K_p$ is symmetric with respect to $p$ the integral of the first order term in the Taylor expansion  is zero.
We now sum the bounds \eqref{taylorf}--\eqref{d2f} for the contributions \eqref{kptaylorerror} over $p\in\Lambda^*$ such that $K_p$ intersects $B_{2^k}\setminus B_{2^{k-2}}$. In this case $2^{k-2}-C_\Lambda\le|p|\le 2^{k}+C_\Lambda$ and since each $K_p$ has volume $1$, the number of such $p$'s is $\le (2^k+C_\Lambda)^2$. Choosing $k_\Lambda$ minimal such that $2^{k_\Lambda}\ge 8C_\Lambda$, we have
\begin{eqnarray*}
 \sum_{p\in\Lambda^*\setminus B_1}|M_p^{2^h}|&\le& C_{\alpha,\eta,\Lambda} + C_\Lambda \sum_{k=k_\Lambda}^h\sum_{2^{k-2}-C_\Lambda\le|p|\le 2^k+C_\Lambda} \|D^2\vp_{2^k}\|_{L^\infty(K_p)}\\
 &\le& C_{\alpha,\eta,\Lambda} + C_\Lambda \sum_{k=k_\Lambda}^h (2^k+C_\Lambda)^2\frac{1}{(2^k)^{2+2\alpha}}\\
 &\le& C_{\alpha,\eta,\Lambda} + C_\Lambda \sum_{k=k_\Lambda}^\infty 2^{2(k+1)-(2+2\alpha)k}\\
 &=& C_{\alpha,\eta,\Lambda} + 2^{-2k_\Lambda\alpha}C'_{\Lambda,\alpha}\, .
\end{eqnarray*} In the last estimate we used the fact that $\alpha>0$ thus the series $\sum_{k=1}^\infty 2^{-2k\alpha}$ converges. The same reasoning easily extends to the case $\op{Re}(\alpha)>0$.

We thus see that the series in $M_p^N$ converges absolutely and uniformly in $N$ for $\tau,\alpha$ inside a fixed compact set respecting the given constraints. Since $\eta$ had bounded support, just finitely many terms contribute to each sum and by Morera's theorem the series gives a function which is  analytic in $\alpha$. By local uniformity of the convergence the limit in $N$ remains analytic in $\alpha$ for $\op{Re}(\alpha)>0$.

As discussed at the beginning of the proof, the left-hand side and the right-hand side of \eqref{zetaser} agree for $\op{Re}(\alpha)>1$, while they are both analytic on $\op{Re}(\alpha)>0$,  so by analytic continuation they must also agree on $\op{Re}(\alpha)>0$. This completes the proof.
\end{proof}

\section{Separation of points}\label{sec3b}
The proofs in this section follow \cite{rs,rns}, which themselves are inspired by \cite{lieb}. 

The first lemma expresses that if a configuration minimizes $\w$, each point must be  a minimum for the potential generated by the rest. 
\begin{lem}\label{lemreste}
Let $x_1, \dots , x_n\in \R^d$ minimize $H_n$. Let $x_i'$ be the blown-up points and $h_{n}'$ be associated as in \eqref{rescalh} and let $U(X)= h_{n}'(X)- \g(X-x_1') $.
Then, for any $\bar x \in \R^d$ and $\bar x'= n^{1/d}\bar x$, we have 
\begin{equation}\label{compU}
n^{s/d}U(x_1') + 2n \zeta(x_1)\le n^{s/d} U(\bar x')+ 2 n \zeta(\bar x)\end{equation}
\end{lem}
\begin{proof}
Let us denote $\tilde{h}_n= \g * \left( \sum_{i=2}^n \delta_{x_i'} + \delta_{\bar x '} - \muv' \drd\right)$.   We note that $h_n$ and $\tilde{h}_n$ correspond to the potential generated by a zero total charge, so they both decay like $ |x|^{-s-1} $ as $x \to \infty$ while their gradients decay like $|x|^{-s-2}$.

We also let $\bar{h}(X) = \tilde{h}_n(X)- h_n(X)= \g( X-\bar x')-\g( X- x_1') $, and note that it decreases in the same way at infinity. 

Expanding the square, we write
\begin{multline}\label{decW1}
\int_{\R^{d+k}}  \yg |\nab \tilde {h}_{n,\eta}|^2 = \int_{\R^{d+k}}  \yg |\nab h_{n,\eta} |^2 + \int_{\R^{d+k}}  \yg  |\nab \bar h_\eta|^2 +  2\int_{\R^{d+k} } \yg \nab  
h_{n, \eta} \cdot \nab \bar h_\eta. \end{multline}
Using Green's theorem and the decay of $\bar{h}$ and $\nab \bar{h}$ at infinity, we have, if $\bar x' \neq x_1'$
$$ \int_{\mr^{d+k}}\yg |\nab \bar{h}_\eta|^2
= \c\int_{\R^{d+k}}  (\delta_{\bar x'}^{(\eta)} - \delta_{x_1'}^{(\eta)} ) (\g(X-\bar x') - \g(X- x_1') ) = 2 \c \g(\eta)  - 2\c \g(\bar x'- x_1') +o(1)$$ as $ \eta \to 0.  $
On the other hand, still if $\bar x' \neq x_1'$,  
\begin{multline*}
\int_{\mr^{d+k}}\yg \nab  
h_{n, \eta} \cdot \nab \bar{h}_\eta = \c \int (\g (X-\bar x')- \g(X-x_1') )  ( \sum_{i=1}^n \delta_{x_i'}^{(\eta)} - \muv' \drd) 
\\ =   \c\left( \sum_{i=2}^n \g(x_i'- \bar x') - \g(x_i'- x_1') -  \g(\eta) + \g(x_1' - \bar x')\right)\\ - 
\c \int (\g (X-\bar x')- \g(X-x_1') ) \muv' \drd+o(1), 
\end{multline*}as $\eta \to 0$.
We deduce that if $\bar x'\neq x_1'$, 
$$\lim_{\eta\to 0} 
\int_{\R^{d+k}}  \yg |\nab \tilde {h}_{n,\eta}|^2 - \int_{\R^{d+k}}  \yg |\nab h_{n,\eta} |^2 = U(\bar x')-U(x_1') ,$$
and of course there is equality as well if $\bar x '=x_1'$.
Combining with  the splitting formula Proposition \ref{splitting}, it follows that \eqref{compU} holds for minimizers.
\end{proof}

We now state a maximum principle proven in \cite[Theorem 2.2.2]{fks} for operators with weights in the  Muckenhoupt $A_2$ class, which contain the operator $-\div (\yg \nab \cdot)$ that we are using.
 
\begin{lem}\label{pcpmax}
Assume that  $h$ satisfies  
$$-\div (\yg \nab h) \ge 0$$ and $\int_U \yg |\nab h|^2 <+\infty$,
in some open subset $U $ of $\R^{d+k}$, and $\int_U \yg |\nab h|^2 <+\infty$. Then $h$ does not have any local minimum in $U$.
\end{lem}

We are now in a position to deduce the 
\begin{proof}[Proof of Theorem \ref{separation}]
First we prove that all the points are in $\Sigma' \times \{0\}$. Take a minimizing configuration, and let $U$ be as in Lemma \ref{lemreste}. 
%Let also $U_\eta= h_{n, \eta}' (X) - \g_\eta (X-x_1')$ where $\g_\eta= g- f_\eta$. We note that $U_\eta $ coincides with $U$ except in the $B(x_i', \eta)$. 
 Since $h_{n}' $ satisfies \eqref{hnp}, we see that $U$ satisfies 
 $-\div(\yg \nab U)=\g*  \mu $, where $\mu$  is a measure of total mass $-\c$ and compactly supported in $ \R^d \times \{0\}$. Thus $U$ must be asymptotic to $- \g$ at infinity.
In the cases \eqref{wlog}--\eqref{wlog2d}, this implies that $U\to +\infty$ as  $|X|\to +\infty$, and in the case \eqref{kernel}, this implies that $U(X)<0$ for $|X|$ large enough while $U(X) \to 0$ as $|X|\to +\infty$.  In addition $U (X) \to + \infty $ as $X\to x_i'$ for some $i\neq 1$.
In both cases, it  follows that $U$ must achieve a minimum somewhere in $\R^{d+k}\backslash \cup_{i=2}^n \{x_i'\}$. 
But from \eqref{hnp}, we have that $-\div (\yg \nab U) \ge 0 $ in $\R^{d+k } \backslash (\Sigma'\times \{0\}\cup \cup_{i=2}^n  \{x_i'\})$ which is an open set.  Also $\yg |\nab U|^2$ is locally integrable away from the points $x_i'$. The maximum principle, Lemma \ref{pcpmax} thus applies to $U$ away from the points $x_i'$, and implies that the  minimum must be in $\Sigma' \times \{0\}$ (because it can't be at any $x_i'$ for $i \ge 2$). Call $\bar x' $ this point and $\bar x= n^{-1/d}\bar x'$. In view of Lemma \ref{lemreste}, since $\zeta(\bar x)=0$,  we have 
$$ U(\bar x') \le U(x_1') \le U(\bar x') - 2n^{1-s/d} \zeta(x_1)$$
which implies (since $\zeta\ge 0$) that $U(x_1')= U(\bar x') $ and thus $x_1'$ is also a point of minimum of $U$. Thus $x_1'\in \Sigma'$ by the maximum principle.
Since the system is invariant under relabelling, this shows that all the points $x_i$ are in $\Sigma$, as claimed.

We next prove the separation result. Now that we know that all the points are in $\Sigma$, \eqref{compU} gives that for any $\bar x'\in \R^d$,
\begin{equation}\label{compU2}
U(x_1') \le U(\bar x') +2 n^{1-s/d} \zeta(\bar x).
\end{equation}

Let $\ro= \left(\frac{1}{\omega_d \overline{m}}\right)^{\frac{1}{d}}$ where $\omega_d $ is the volume of the unit ball in dimension $d$, and $\overline{m}$ is the constant in \eqref{assmu2}.
Let $x_2'$ be some point  in the collection and let us assume  that $x_1'\in B(x_2', \ro)$ (up to relabelling).
 Note that the  choice of $\ro$ ensures that $\overline{m} \mathcal H^d(B(x_2', \ro)\cap \R^d\times \{0\}))<1$.
We next  split  $U$  into 
\begin{equation}
\label{splitu}
U= U^{near}  + U^{rem}
\end{equation}
where
\begin{align*}
& U^{near}=  \g* \Big( \delta_{x_2'}- \overline{m}   \indic_{B(x_2', \ro)}\drd\Big)
\\ &
U^{rem}= \g * \Big( \sum_{i=3}^n \delta_{x_i'}-  \muv'  \drd  +\overline{ m} \indic_{B(x_2', \ro) }\drd \Big).\end{align*}

One may observe that 
$$-\div (\yg \nab U^{rem}) \ge \c \left( \overline{m} \indic_{B(x_2', \ro)}- \muv' \right) \drd \ge 0\quad \text{in} \ (\Sigma')^c\cup B(x_2', \ro)$$
because we always have $\overline{m} \ge \muv'$.
 Moreover, since $- \int\div (\yg \nab U^{rem})=\c( -3 + \overline{m} \omega_d \ro^d) $ is a negative number $C\le -2 \c$ by choice of $\ro$,  $U^{rem}  $ is asymptotic to $\frac{C}{\c} \g$ at infinity.  By the same reasoning as above it follows that $U^{rem}$ achieves a minimum somewhere. 
 By Lemma \ref{pcpmax} this minimum is at some point $\bar{x}'$ which cannot be in $(\Sigma')^c \cup B(x_2', \ro)$ hence is in the closure of  $(\Sigma' \backslash B(x_2', \ro))\times \{0\}$.
 Thus \begin{equation}\label{Urem1}
 U^{rem}(\bar{x}') \le U^{rem} (x_1')\qquad \zeta(\bar x)=0.\end{equation}

We next turn to $U^{near}$ and show that if $|x| \overline{m}^{1/d}$ is smaller than some constant $r>0$ depending only on $d$ and $s$, then 
\begin{equation}\label{compunear}
U^{near} (x) \ge \max_{B(x_2', \ro)^c} U^{near} .\end{equation}
By change of origin, we may  assume that $x_2'=0$.  By scaling, it also suffices to prove this for $\overline{m}= 1$ and $\ro = \omega_d^{-1/d}$.
 We note that then $U^{near} \le \g +C $, where the constant depends only on $s,d$ and can be taken to be $0$ in the Riesz cases (it is there to account for the possible negativity of $-\log |x|$ in the logarithmic case), so that $\max_{B(0, \ro)^c} U^{near} \le \g(\ro)+C= \g(\omega_d^{-1/d})+C$. On the other hand, in the Riesz case, 
 $$U^{near} (x) = \g(x) - \int_{y \in B(0, \ro)} \g(y-x)\, dy  \ge \g(x) - \int_{y\in B(0, |x|+\ro) }\g(y) \, dy \ge \g(x) - C (\ro+ |x|)^{d-s}$$ and the conclusion follows immediately. One easily checks that the conclusion holds as well in the logarithmic cases, and one can also note that an explicit estimate of the best $r$ is possible.

 If $r$ is taking as above, then if $x_1' \in B(x_2', r \overline{m}^{-1,d})$, we have from \eqref{compunear} that 
This implies that $U^{near} (\bar x') < U^{near} (x_1')$. Combining this with \eqref{Urem1} and   \eqref{splitu}, it follows that
$
U(\bar{x}')< U (x_1')$  and $\zeta(\bar{x})=0$, a contradiction with \eqref{compU2}. Thus, for any pair of points $x_1', x_2'$ (up to relabelling), we must have $|x_1'-x_2'|\ge r {\overline m}^{-1/d}$. Scaling down  concludes the proof.   \end{proof}

\begin{pro}\label{pro53}
Assume $E $ is a minimizer of $\W$ over the class of vector-fields in $\mathcal A_1$ which are $R$-periodic, for some given $R$ such that $|K_R|\in \mathbb N$. Then, letting $\Lambda$ be the associated set of points,  we have $$\min_{p, q \in \Lambda, p\neq q } |p-q|\ge r $$ where $r$ is a positive constant depending only on $s$ and $d$ (the same as in Theorem \ref{separation}).
\end{pro}
Note that such a minimizer exists by the explicit formula in Proposition \ref{periodic}.
\begin{proof}
The argument is the same as above.  First we note that the points of $E $ must have single-multiplicities, otherwise $\W(E)$ would be infinite by Proposition \ref{periodic}.
Second,  we let $h$ be the periodic potential generated by the $N=|K_R|$ minimizing points $x_1, \dots x_N$,  and $G$ be the periodic Green function of the operator $-\div (\yg \nab \cdot )$ as in \eqref{defG1}. Let  $U(X)= h(X) - G(X-x_1)$ and let $\tilde h= G* \left( \sum_{i=2}^N \delta_{x_i}+ \delta_{\bar x}- \drd\right) $ and $\bar h(X)= \tilde h (X)- h(X) =G(X-\bar x)- G(X-x_1)$. 
Computing exactly as in 
 Lemma \ref{lemreste}, we find that  $\W (\nab \tilde h) - \W(\nab h) = U(\bar x)-U(x_1')$, hence since $(x_1, \dots, x_N)$ minimize $\W$, we must have that $x_1'$ is a minimum of $U$. 
 
 Next we  assume by contradiction  that  there is a point in the collection, say  $x_2$, such that   $x_1\in B(x_2, r)$ where $r$ is as in Theorem \ref{separation}, and write 
\begin{equation}
\label{splitu2}
U= U^{near}  + U^{rem}\end{equation}
where
\begin{align*}
& U^{near}= G * \left( \delta_{x_2}-   \indic_{B(x_2, \ro)}\drd\right)
\\ &
U^{rem}= G * \Big( \sum_{i\neq 1, 2}\delta_{x_i}+ (   \indic_{B(x_2, \ro) } - 1) \drd\Big).\end{align*} By definition of $G$, we have
\begin{multline*}-\div (\yg \nab U^{rem})=\c\left( \sum_{i\neq 1, 2}\delta_{x_i}+ \left(   \indic_{B(x_2, \ro) } - 1- \frac{- 2 + |B(x_2, \ro)|}{N}\right) \drd\right)\\
\ge \c\left( (\indic_{B(x_2, \ro)} -1 )\drd\right)\end{multline*} where we used that $|B(x_2, \ro)|<1$, i.e. $-\div(\yg \nab U^{rem})\ge 0 $ in $B(x_2, \ro)\cup (\R^d)^c$. Moreover, since $G$ has average $0$ in each periodicity cell of $\R^d$, $U^{rem}$ too, so $U^{rem}$ takes negative values unless it is constant. But  $U^{rem} $ is periodic in the $x\in \R^d $ direction and tends to $0$ as $|y|\to \infty$ (like $G$), so $U^{rem}$ must achieve a minimum at some point. 
By Lemma \ref{pcpmax}, this minimum   $\bar x$ can only  be in $(\R^d \backslash B(x_2, \ro)) \times \{0\}$. On the other hand, we saw in the proof of Theorem \ref{separation} that $U^{near}(x_1) >\max_{B(x_2, \ro)^c}U^{near}$.  
It follows that if $|x_1-x_2|<r$, we have $U(\bar x)<U(x_1)$, a contradiction.

\end{proof}

\section{Lower bound with the ergodic theorem approach}\label{sec4}

In this section, we turn to obtaining a lower bound for the energy of arbitrary (non necessarily minimizing) configurations. 
From Proposition \ref{splitting} and Lemma \ref{prodecr} 
we already have the following: for all $\eta<1$, 
\begin{multline}\label{binspl}
H_n(x_1, \dots, x_n) - n^2 \I(\muv) \\
\ge 2n \sum_{i=1}^n \zeta(x_i) + 
n^{1+\frac{s}{d}}   \left( \frac{1}{n} \int_{\R^{d+k}}\yg |\nab h'_{n,\eta}|^2 -   \c\g(\eta) - C\eta^{\frac{d-s}{2}} \right) ,
\end{multline}
respectively in the cases \eqref{wlog}--\eqref{wlog2d},
\begin{multline}\label{binspllog}
H_n(x_1, \dots, x_n) - n^2 \I(\muv)+ \frac{n}{d}\log n \\
\ge 2n \sum_{i=1}^n \zeta(x_i) + 
n   \left( \frac{1}{n} \int_{\R^{d+k}}\yg |\nab h'_{n,\eta}|^2 -   \c\g(\eta) - C\eta^{d/2} \right) ,
\end{multline}
where $C$ depends only on $d, s, V$.  We recall that from Lemma \ref{prodecr} there is equality if $\min_{i\neq j} |x_i-x_j|>2\eta.$

In this section we will take the limits $n\to\infty$ and $\eta\to 0$  in the above relations to  provide the lower bound for the energy.
This is done as in  \cite{rs} according to the method for ``lower bounds for 2-scale energies" initiated in \cite{gl13,ss2d} and inspired by Varadhan.
 The idea is to rewrite the energy as an average of energies computed on finite size balls after blow-up.
More precisely,   consider a radial smooth probability density $\chi:\R^d\to \R^+$ supported in the unit ball of $\R^d$. We may  rewrite  the energy of  $\hne'$ (which is the function defined using \eqref{defheta} and \eqref{rescalh}) as:
\begin{eqnarray*}
 \int_{\R^{d+k}}\yg|\nab \hne'|^2&=&\int_{\R^{d+k}}\left(\int_{\R^d}\chi(x-\tilde x)d\tilde x\right)\yg|\nab \hne'(X)|^2dX\\
 &=&\int_{\R^{d+k}}\int_{\R^d}\chi(\tilde x)\yg|\nab\hne'(X+(\tilde x,0))|^2d\tilde x\ dX\\
 &\ge&\int_{n^{1/d}\Sigma}\int_{\R^k}\int_{\R^d}\chi(\tilde x)\yg\left|\nab\hne'\left(X+(\tilde x,0)\right)\right|^2d\tilde x\ dy\ dx,
 \end{eqnarray*}
where we discarded the integral over the complement of $n^{1/d} \Sigma$ which we guessed to be unimportant. Changing variables, we obtain
\begin{equation}\label{crucialerg}
 \int_{\R^{d+k}}\yg|\nab \hne'|^2
 \ge
 n|\Sigma| \dashint_{\Sigma} \int_{\R^{d+k}} \chi(\tilde x) \yg \left|\nab\hne'\left(xn^{1/d}+\tilde x,y\right)\right|^2d \tilde x\, dy\ dx\, .
\end{equation}
The method then consists in examining the local energies thus defined, i.e. 
$$\int_{\R^{d+k}} \chi(\tilde x) \yg \left|\nab\hne'\left(xn^{1/d}+\tilde x,y\right)\right|^2d\tilde x \, dy$$
which have natural limits and to rewrite \eqref{crucialerg} as an average over $x\in \Sigma$ of these energies. 

More precisely, given some configuration of points and $h_{n,\eta}' $ its associated truncated blown-up potential,   the local energy is  defined as follows  based on \eqref{crucialerg}, for $(x,\Y)\in \mathcal X$ ($\mathcal X $ is the space $\Sigma \times L^p_{\loc}(\R^{d+k}, \R^{d+k})$ as specified  in the introduction):
\begin{equation*}
\mathbf f_{n,\eta}(x,\Y):=\left\{\begin{array}{ll}\int_{\R^{d+k}}\chi\yg|\Y|^2&\text{ if }\ \Y(X)=\nab\hne'(X+(xn^{1/d},0))\, ,\\[3mm] +\infty&\text{ else}\, ,\end{array}\right.
\end{equation*}
where $\theta_\lambda,\lambda\in\R^d$ is the group of translations of $\R^d$, which acts on $\mathcal X$ by 
\[
 \theta_\lambda \Y(X):=\Y(X+(\lambda,0))\, . 
\]To separate scales we also consider scale-$n^{1/d}$ coupled actions of the $\R^d$-translations on $\R^d\times\mathcal X$ defined as follows:
\[
 T_\lambda^n(x,\Y):=\left(x+n^{1/d}\lambda,\theta_\lambda\Y\right)\, .
\]
The global energy is defined as an average of the local ones by: 
$$\mathbf F_{n,\eta}(\Y):=\dashint_{\Sigma}\mathbf f_{n,\eta}\left(x, \theta_{xn^{1/d}}\Y\right)dx$$
and \eqref{crucialerg} translates into the upper bound
\begin{equation}\label{crucialerg2}
 \mathbf F_{n,\eta}(\Y)\le\frac{1}{|\Sigma|n}\int_{\R^{d+k}}\yg|\nab \hne'|^2\ ,
\end{equation}
if $\mathbf F_{n,\eta}(\Y) \neq +\infty$.
In view of \eqref{crucialerg2} and \eqref{binspl}, to bound the energy from below, it suffices to bound from below $\mathbf F_{n,\eta}$. Theorem 7 in \cite{ss2d} is precisely designed to obtain lower bounds on such energies from input at the microscopic scale (i.e. on $\mathbf f_n$). The idea is that 
$\mathbf F_{n,\eta}$ is roughly $\int \mathbf f_{n,\eta} (x, \mathcal Y) \, dP_{n,\eta} (x,\mathcal Y)$ where $P_{ n,\eta}$ is the push-forward of $P_{\nu_n}$, as defined in \eqref{pnun} by the map $\Phi_\eta$, or in other words
the push-forward of the normalized Lebesgue measure on $\Sigma$ by 
$$x\mapsto \left(x, \nab h_{n,\eta}'(n^{1/d} (x,0) + \cdot)\right).$$

Then it suffices to obtain some tightness for $\{P_{n,\eta}\}$ and pass to the limit in this average  $\int \mathbf f_n \, dP_{n, \eta}$ to obtain a lower bound by $\int \mathbf f_\eta \, dP_\eta $ where $\mathbf f_\eta$ is identified as the $\liminf $ of $\mathbf f_{n,\eta}$.
More precisely, if $(\tilde x_n, \mathcal Y_n ) \to (\tilde x, \mathcal Y)$ in $\mathcal X$ we have 
$$\liminf_{n\to \infty} \mathbf  f_{n,\eta}(x_n, \Y_n)\ge \mathbf f_\eta(\tilde x, \Y)\ ,
 $$
where
\[
 \mathbf f_\eta (\tilde x,\Y):=\begin{cases}
\D\int_{\mr^d}\chi\yg|\mathcal{Y}|^2& \text{if } \tilde x \in \Sigma \text{ is a point of continuity of }  \muv  \\
& \text{and} \ \mathcal{Y} = \Phi_\eta(E)  \text{ for some } E \in{\mathcal A}_{\mu_V (\tilde x)}\\[3mm]
0 & \text{if } \ \tilde x \text{ is a point of discontinuity of $\muv$},\\[3mm]
+\infty & \text{otherwise} .
\end{cases}
\]

 Showing that the limit $P=\Phi_\eta^{-1}(P_\eta)$ is admissible will ensure in particular that it is translation-invariant, and that we may apply the multi-parameter ergodic theorem to conclude with the desired lower bound.

The first step is to obtain the tightness of $\{P_{n,\eta}\}_n$. This will be obtained from the following compactness result for the local energies, analogous to \cite[Lemma 4.2]{rs}.
\begin{lem}\label{rslem52}\mbox{}\\
Let $h_{n,\eta}'$ be the truncated blown-up potential generated by a configuration of points and let $\nu_n'= \sum_{i=1}^n \delta_{x_i'}$. 
Assume that  for every $R>1$ and for some  $\eta\in(0,1)$, we have
\begin{equation}\label{eq:asum h}
\sup_n  \int_{B_R\times\R^k} \yg\left|\nab \hne'\left(\left(n^{1/d}\tilde x_n,0\right)+\cdot\right) \right|^2 \le C_{\eta, R},
\end{equation}
and that the centering point $\tilde x_n\to \tilde x\in\mr^d$ as $n\to \infty$.
Then $\{\nu'_n((n^{1/d}\tilde x_n,0)+\cdot)\}_{n}$ is locally bounded and up to extraction converges weakly as $n\to \infty$, in the sense of measures, to
$$\nu = \sum_{p\in\Lambda}N_p \delta_p$$
where $\Lambda $ is a discrete set and $N_p\in \mathbb{N}^*$.
In addition, there exists $E, E_\eta \in L^p_{\loc}(\R^{d+k} ) $ for  $1<p<\min(2,\frac{2}{\gamma+1}, \frac{d+k}{s+1})$,  with $E_\eta=\Phi_\eta(E)$
such that up to further extraction of a subsequence,
\begin{equation}\label{convweak1}
\nab h_n'\left(\left(n^{1/d} x_n,0\right)+\cdot\right)  \rightharpoonup   E \ \text{weakly in } \ L^p_{ \loc}  \ \text{as } \ n\to \infty ,
\end{equation}
and
\begin{equation}\label{convweak2}
\nab \hne'   \left(  \left(n^{1/d}x_n,0\right)+\cdot\right)   \rightharpoonup  E_\eta \ \text{ weakly in } \ L^p_{\loc}  \ \text{and } L^2_{\yg, \loc} \ \text{as } \  n \to \infty.
\end{equation}
Moreover $E$ is a gradient, and if $\tilde x\in \Sigma$ is a point of continuity of  $\muv$, we have
\begin{equation}\label{eqhl}
- \op{div}(\yg E) = c_{d,s}(\nu -\mu_V(\tilde x))\quad \text{in } \ \R^{d+k} \end{equation} hence $E \in{\mathcal A}_{\mu_V (\tilde x)}$.
\end{lem}

\begin{proof}
Following the proof of Lemma \ref{limdensity}, we easily deduce from the bound \eqref{eq:asum h}  that 
there exists $t\in [R-1,R]$ such that 
\[
\left| \int_{ B_t\times\R^k}    \sum\delta_{(x_i',0)}^{(\eta)} ((n^{1/d} \tilde x_n,0) + X) dX   - \int_{B_t} \mu_V(\tilde x_n  +  n^{-1/d}  x)\, dx \right|\ \le C_{\eta, R},
\] for some constant that depends on $\eta, R$ and the constant in \eqref{eq:asum h}.
Since $\mu_V$ is bounded it follows that, letting $\underline{\nu_n'}:=\nu_n'(n^{1/d}\tilde x_n,0) + \cdot) $, we have
\[
\underline{\nu_n'}(B_{R-1}\times\R^k) \le C_d \|\mu_V\|_{L^\infty}  R^d  +C_{\eta,R}\,. 
\]
This establishes that $\{\underline{\nu_n'}\}$ is locally bounded independently of $n$. In view of the form of $\underline{\nu_n'}$, its limit can only be of the form $\nu= \sum_{p\in \Lambda} N_p \delta_p$, where $N_p$ are positive  integers, and $\Lambda$ is a discrete set contained in $\R^d\times\{0\}$.

From the bound \eqref{eq:asum h}, up to a further extraction, we have that $\underline{\nab h}_{n,\eta}':= \nab h_{n, \eta}' ((n^{1/d} \tilde x_n,0)+\cdot) $ is locally weakly convergent in the weighted space $L^2_{\yg}$, and converges (locally) to some vector field $E_\eta$. 
Using H\"older's inequality, we note that $L^2_{\yg} (B_R) $ (where $B_R$ denotes the ball of radius $R$ in $\R^{d+k}$) embeds continuously into $L^q(B_R)$ for $1<q<\min(2, \frac{2}{\gamma+1})$. It thus follows that $\underline{\nab h}_{n,\eta}' $ is bounded in $L^q_{\loc}$ and converges to $E_\eta$ also in the sense of distributions. 
We next deduce that  $\underline{\nab h}_{n}':=  \nab h_n' ( (n^{1/2} \tilde x_n,0)+\cdot)$ is bounded in $L^p(B_R)$ for $1< p<  \min(2, \frac{2}{\gamma+1}, \frac{d+k}{s+1})$. Indeed, $\nab h_n'= \nab h_{n,\eta}'- \sum_{i=1}^n \nab f_\eta(X- (x_i', 0))$, 
$\nab f_\eta$ is in $L^p(B_R)$ for any $p<\frac{d+k}{s+1}$, and the number of $x_i'$'s in $(n^{1/d}\tilde x_n,0)+ B_R$ is bounded by $ \underline{\nu_n'}(B_{R}\times\R^k) $, hence bounded. 
We thus deduce that, up to a further extraction, $\underline{\nab h}_n' $ converges weakly in $L^p(B_R)$ for such $p$'s, to some vector-field $E$, which must be a gradient.
Moreover,  $E_\eta=\Phi_\eta(E)$ because $\Phi_\eta$ commutes with the weak convergence in $L^p_{\loc}$ of  $\underline{\nab h}_n'$. Indeed by definition
\[
\Phi_\eta (\underline{\nab h}_n') =\underline{ \nab h}_n'  + \sum_{p \in \Lambda_n} N_p \nab f_\eta (. - p)\ ,
\]
where $\Lambda_n$ is the set of points associated with $\underline{\nu_n'} $. Since  all these points have limits, one may check that the sum in the right-hand side converges to $\sum_{p \in \Lambda} N_p \nab f_\eta (. - p)$, at least weakly in $L ^p_{\loc}$. Taking the limit, we deduce
\[
 E_\eta = E + \sum_{p \in \Lambda}N_p \nab f_\eta (. - p),
\]
i.e. $E_\eta = \Phi_\eta(E)$ as desired.

There remains to show that \eqref{eqhl} holds. For that we start from \eqref{hnpe} and translating  the equation by $n^{1/d}(\tilde x_n,0)$ and integrating  against a smooth compactly supported test function $\vp$, we find 
$$\int_{\R^{d+k}} \nab \vp \cdot \underline{\nab h}_{n,\eta}'\yg = \c\int_{\R^{d+k}} \vp  \left( \sum \delta^{(\eta)}_{x_i'- n^{1/d} \tilde x_n  }    -\muv'( n^{1/d}(\tilde x_n,0)+ \cdot) \drd\right). $$
In view of the $L^2_{\yg}$ convergence of $\underline{\nab h}_{n,\eta}'$, the weak convergence in measures of $\underline{\nu}_n'$ and the continuity of $\muv $ at $\tilde x=\lim_{n\to \infty} \tilde x_n$, taking the $n\to \infty $ limit in this relation yields
$$\int_{\R^{d+k}} \nab \vp \cdot E_\eta  \yg= \c\int_{\R^{d+k}} \vp  \left(\sum_{p\in\Lambda} N_p \delta_p^{(\eta)}-\muv(\tilde x) \drd\right), $$
where $\Lambda, N_p$ are associated to $\nu$. 
It thus follows that 
$$-\div (\yg E_\eta)= \c\sum_{p\in \Lambda} N_p \delta_p^{(\eta)}- \muv(\tilde x) \drd \quad \text{in} \ \R^{d+k}$$
in the sense of distributions.
Since $E=\Phi_\eta^{-1}(E_\eta)$, the relation \eqref{eqhl} follows.

\end{proof}

At this point, the rest of the proof is identical to \cite{rs}. We thus only state the main steps. 
The assumptions of the abstract Theorem 7 of \cite{ss2d} are satisfied thanks to Lemma \ref{rslem52}. 
This theorem then ensures 
that $P_{n,\eta}$ as defined above is tight and converges up to extraction to a Borel probability measure $P_\eta $ on $\mathcal X$, and $P:=\Phi_\eta^{-1}(P_\eta)$ is admissible (for this we need that $P$-a.e. $\mathcal(x,E)$,  $E\in\mathcal{A}_{\muv(x)}$. This is ensured by Lemma \ref{rslem52} assuming that a.e. point in $\Sigma$ is a point of continuity of $\muv$.
 Theorem 7 of \cite{ss2d} also  yields that 
 \begin{equation}
\label{rg1} \liminf_{n \to 0} \mathbf F_{n,\eta}(\Y) \ge \int \mathbf f_\eta( x, \Y)\, dP_\eta(x,\Y)=\int 
\left(\lim_{R\to
+\infty}\dashint_{K_R} \mathbf f_\eta(x,\theta_\lambda  \Y  )\,d\lambda\right)\, dP_\eta(x, \mathcal Y).
\end{equation}
The second relation is an application of Wiener's multiparameter ergodic theorem as in \cite{becker}, in view of the translation-invariance of $P$ and thus  of  $P_\eta$, and it is part of the result that the limit exists. 
Finally, combining this with the definition of $\mathbf f_\eta $ and \eqref{crucialerg2}, we are led to 
$$\liminf_{n\to \infty} \frac{1}{|\Sigma|n} \int_{\R^{d+k}} \yg |\nab h_{n,\eta}'|^2 \ge \lim_{R\to \infty} \dashint_{K_R} |\Phi_\eta(E)|^2 \, dP(x, E).$$
Inserting into \eqref{binspl} and using the fact that $\int \muv=1$ and the first marginal of $P$ is the normalized Lebesgue measure on $\Sigma$, we obtain 
$$\liminf_{n\to \infty} 
 n^{-1-\frac{s}{d}} \left(  H_n(x_1, \dots, x_n) - n^2 \I(\muv)\right) \ge
\frac{|\Sigma|}{\c} \int \W_\eta(E)\, dP(x, E)- C \eta^{\frac{d-s}{2}},$$ and the analogous statement in the cases \eqref{wlog}--\eqref{wlog2d}.
It then remains to let $\eta \to 0$. Since $P$-a.e. $(x,E)$ is in $\mathcal A_{\muv(x)}$ and $\muv $ is bounded, in view of \eqref{scalingW}--\eqref{scalinglog} and Proposition \ref{Wbb}, we have that  $\W_\eta$ is bounded below independently of $\eta$ for $P$-a.e. $(x, E)$. We may thus apply Fatou's lemma to take the $\eta \to 0$ limit. 
In view of the definition of $\widetilde{\W}$ \eqref{Wtilde} we  thus obtain the following general lower bound result.

% We then make use of the following theorem (see \cite{chatard}, \cite{Becker}).
% \begin{theo}[Multiparameter ergodic theorem]\label{ergodicabstr}
%  Let $X$ be a Polish space and consider a continuous action $\mathbb Theta:\mathbb R^d\to C^0(X)$ of the group $(\mathbb R^d, +)$ on $X$. Assume that $P$ is a $\mathbb Theta$-invariant probability measure on $X$. Then for $f\in L^1(X;P)$ the limit
%  \[
%   f^*(u)=\lim_{R\to\infty}\dashint_{K_R} f(\mathbb Theta_\lambda u)d\lambda\ \text{ exists, a.e. }u\in X\ ,
%  \]
% provided that $K_R$ is a Vitali family of sets of $\mathbb R^d$. Moreover for all $f\in L^1(X;P)$ there holds
% \[\ell
%  \int f\ dP=\int f^*\ dP\ .
% \]
% \end{theo}

\begin{pro}\label{proergodic} Assume $V$ satisfies \eqref{assv1}--\eqref{assv2}--\eqref{assv3}, and that $\muv$ is a measure with an $L^\infty$ density and a.e. continuous.
Let $x_1, \dots, x_n \in \R^d$ and define $P_{\nu_n}$ as in \eqref{pnun}.
Then up to extraction of a subsequence, $P_{\nu_n}$ converges weakly in the sense of probability measures to a measure $P\in \mathcal{P}(\mathcal X)$ which is admissible, and 
 \begin{equation}\label{binf}
\liminf_{n\to \infty} n^{-1-\frac{s}{d}} \left(  H_n(x_1, \dots, x_n) - n^2 \I(\muv)\right) \ge
\widetilde{\W}(P)\  \mathrm{in case \eqref{kernel}}\, ,\end{equation}
respectively 
\begin{equation}\label{binflog}
\liminf_{n\to \infty} \frac{1}{n} \left(  H_n(x_1, \dots, x_n) - n^2 \I(\muv)+\frac{n}{d} \log n \right) \ge
\widetilde{\W}(P)\  \mathrm{in cases \eqref{wlog}--\eqref{wlog2d}}.\end{equation}
\end{pro}

\section{Screening}\label{sec5}
Starting with this section, we turn to the upper bound part of the proof. We now need to use the further assumptions on $\muv$, \eqref{assumpsigma}--\eqref{assmu2}.

In this section, we prove the screening result.  More precisely, 
we consider an ``electric" vector field in a strip $K_R\times \R^k$ for $R$ large. 
 We would like 
to  prove that up to errors on the energy which can be made negligible as $R\to \infty$, we may modify $E$ in a neighborhood of the boundary of the strip
 in such a way as to obtain $E\cdot\nu=0$ on $\p  K_R\times \R^k $. 
 This  allows to later patch together several copies of the so-constructed $E$ without creating new divergence. We also want to ensure that the construction preserves the ``good separation" of the points.

Constants $C$ will mean positive constants which may only depend on $d$ and $s$.

\begin{pro}[screening]\label{screening}Let $0<\ep<1/2$ be given. 
Assume $K_R \subset \R^d\times \{0\}$ is a hyperrectangle whose sidelengths are in $[2R, 3R]$ and such that $|K_R|$ is an integer, and that  $E$ is a vector field defined in $\check{K}_R\times \R^k$ with $\check {K}_R:= \{x\in K_R, \dist(x, \p K_R)\ge \hal \ep R\}$ and satisfying 
$$-\div (\yg E) = \c \Big( \sum_{p \in \Lambda } N_p \delta_p -\drd\Big) \quad \text{in} \ \check{K}_{R}\times \R^k$$
for some discrete set $\Lambda\subset \R^d$ and positive integers $N_p$. There exists $\eta_0>0$ depending only on $d$ such that for any $0<\eta, \eta'<\eta_0$, the following holds. 
 Let  $E_\eta, E_{\eta'}$ be associated to $E$ as in \eqref{defeeta} and let  
$$\M:=\frac{1}{R^d}\int_{\check{K}_R\times \R^k}\yg |E_{\eta}|^2 ,$$ and in the case $k=1$
\begin{equation}\label{decrvert}
e_{\ep,R}:=  \frac{1}{ \ep^4 R^d} \int_{\check{K}_R\times \(\R\backslash (-\hal \ep^2 R,\hal \ep^2 R)\)} \yg |E|^2.\end{equation}
 There exists $R_0>0$ such that if
 \begin{equation}\label{condR}
 R> \max \left( \frac{R_0}{\ep^{2}},  \frac{R_0\Mp}{\ep^3}\right),   \quad 
R>   \begin{cases}R_0 \M^{1/2} \ep^{-d-3/2} & \text{if } k =0\\ 
\max(R_0\M^{1/(1-\gamma)} \ep^{\frac{-1-2d+\gamma}{1-\gamma}} , R_0\ep^{\frac{2\gamma}{1-\gamma} }  e_{\ep,R}^{1/(1-\gamma)}   )  & \text{if} \ k=1
\end{cases}, 
 \end{equation} 
then
  there exists a 
 vector field $\hat E\in L^p_\loc(\mathbb R^{d+k},\mathbb R^{d+k})$ (with $p$ as in the introduction) such that
 \begin{itemize}
  \item $\hat E \cdot \vec{\nu}=0$ on $\p  K_R\times \R^k $, where $\vec{\nu}$ is the outer unit normal, and $\hat E=0$ outside $ K_R\times \R^k$.
  \item There exists a subset $\hat \Lambda\subset K_R$ and positive integers $N_p$ such that $\hat E$ satisfies
   \[
   -\op{div}\left(\yg\hat E\right)=c_{s,d}\Big(\sum_{p\in\hat \Lambda}N_p\delta_{p}-\drd\Big) \  \text{in} \ K_R\times \mr^k .
  \]
 \item $\hat E=E$ in a hyperrectangle $K_R'\times[-\ep^{2} R, \ep^{2} R ]^k$ where $K_R'$ contains $\{x\in K_R, \dist (x, \p K_R)> 2 \ep R\}$. In particular $\hat \Lambda\cap K_R'=\Lambda \cap K_R'$.
 \item The minimal distance between the points in $\hat\Lambda \cap (K_R\backslash K_{R'})$ and between them and $\p K_R$ is bounded below by $\eta_0$.
The minimal distance between points in $\hat \Lambda$ counted with multiplicity, and between points in $\hat\Lambda$ and $\p K_R$ is bounded below by the minimum of $\eta_0$ and of the minimal distance between the original points of $\Lambda$ counted with multiplicity. 
 In other words, if the points of $\Lambda$ are simple and well-separated, so are those of $\hat \Lambda$. 

 \item Letting $\hat E_\eta$ be $\Phi_\eta(\hat E)$, we have
 \begin{equation}\label{C3}
  \int_{ K_R\times\mathbb R^k}\yg|  \hat E_\eta|^2 \le \left(\int_{ \check{K}_R\times\mathbb R^k}\yg|E_\eta|^2 \right) (1+ C \ep )  +
  C \g(\eta) (1+\Mp)   \ep R^d + C e_{\ep, R} \ep R^d\, ,
 \end{equation} 
 where  $C$ depends only on $s $ and $d$.
 \end{itemize}
\end{pro}
\begin{remark}
In this result,  one should think of $\M$   as being bounded above by a constant times $\g(\eta)$, and $e_{\ep, R}$ as bounded by a constant.   Then the conditions of $R$ are that it  has to be large enough, so much so as $\ep$ and $\eta$ are small. In  a first reading, one may also take $\eta'=\eta$ in the statement and obtain 
$$ \int_{ K_R\times\mathbb R^k}\yg|\hat E_\eta|^2 \le \left(\int_{ \check{K}_R\times\mathbb R^k}\yg|E_\eta|^2 \right) (1+ C\g(\eta) \ep )+ C \ep R^d. 
$$ Since $\int \yg |E_\eta|^2 $ blows up like $\g(\eta)$ as $\eta \to 0$, this gives an additive error in $\g(\eta)^2$.
But one may also prefer to choose say $\eta'=1/8$ and $\eta$ small and meant to tend to $0$. For a  vector field $E\in \mathcal A_1$ of finite energy, $\Mp$ will be bounded  by a constant depending on the choice $1/8$ and $\W(E)$. The formulation then  gives in that case a bound 
$$\int_{ K_R\times\mathbb R^k}\yg|\hat E_\eta|^2 \le \left(\int_{ \check{K}_R\times\mathbb R^k}\yg|E_\eta|^2 \right) (1+ C \ep ) + C \g(\eta) \ep R^d.$$ 
The additive
 error term  then blows up like $\g(\eta)$ as $\eta\to 0$  instead of $\g(\eta)^2$.
\end{remark}

\subsection{Preliminary lemmas}
We start with a series of preliminary results which will be the building blocks  for the construction of $\hat E$. 
\begin{lem}[Subdivision of a hyperrectangle]
 \label{subdivision}
 Let  $H=[0,\ell_1]\times\cdots\times[0,\ell_d]$ be a $d$-dimensional hyperrectangle of sidelengths $\ell_i$.    Fix a face $F$ of $H$. Let $m>0$ be such that $m|H|\in \mathbb N$, and for all $i$, $\ell_i >2m^{-1/d}$. 
Then  there is a partition of $H$ into $m|H|$  subrectangles $\mathcal R_j$,  such that the following hold
 \begin{itemize}
  \item  all rectangles have volume $1/m$,   \item the sidelengths of each $\mathcal R_j$ lie in the interval $[2^{-d} m^{-1/d},2^dm^{-1/d}]$,
  \item all the $\mathcal R_j$'s which have a face in common with $F$ have the same  sidelength in the direction perpendicular to $F$. 
 \end{itemize}
\end{lem}Note that in general, even if the area of $H_i$ is assumed to be an integer, it is not possible to divide it into \emph{congruent} rectangles of sidelengths in $[C^{-1},C]$ with a constant $C$ independent of $H_i$, see the lower bound in Theorem 8 of \cite{chan} in case $d=2$, but we may however divide it into rectangles of comparable sizes.
 A question which we don't answer is whether one can achieve a bound on their sizes independent of the dimension. 
\begin{proof}
The statement is obvious in dimension $1$. We prove the statement for $d\ge 2$ by induction on the dimension.
% In case $d=2$ we may divide the $\ell_2$-dimension into $n_2=[\ell_2/s_1]-1$ equal intervals of length $s_2=[\ell_1]/\ell_1$ and then we divide each of the stripes congruent to $[0,\ell_1]\times[0,s_2]$ into $[\ell_1]$ equal rectangles, congruent to $[0,s_1]\times[0,s_2]$ for $s_1=\ell_1/[\ell_1]$. To complete the subdivision of $[0,\ell_1]\times[0,\ell_2]$ we are left with a strip congruent to $[0,\ell_1]\times[0,s_2']$, with $s_2'=\ell_2-n_2s_2\in[1,2[$. Since the total remaining area is an integer, we may subdivide this strip into equal rectangles of area $1$, and their remaining side will have length $s_1'=1/s_2'\in]1/2,1]$.\\
 Up to relabeling the directions, we may suppose that $F$ is one of the faces where all the first $d-1$ coordinates are constant.

We use the induction hypothesis for  the hyperrectangle $H':=[0,\ell_1]\times\cdots\times[0,\ell_{d-1}]$ with the choice   $$ m'=  \frac{[ m^{\frac{d-1}{d} } |H'|]  } { |H'|}$$ with does satisfy $ m' |H'|\in \mathbb N$.  We note that from the assumption $\ell_i >2 m^{-1/d}$ we have $ m'\in [\hal m^{\frac{d-1}{d}}, m^{\frac{d-1}{d}}]$.
We  thus   obtain a subdivision into rectangles $\mathcal R_j'$ of volumes all equal to $1/ m'$ and sidelengths belonging to $[2^{-(d-1)} (m')^{-\frac{1}{d-1}} , 2^{d-1} (m')^{-\frac{1}{d-1}}]$. 
  We then split the interval $[0, \ell_d]$ into  $[\ell_d m /m' ]-1$ adjacent equal intervals of length $m'/m $, plus a remaining interval of length $\ell'\in[ m'/m, 2 m'/m]$.  By taking the cartesian product of the first group of intervals with the $\mathcal R_j'$,      we obtain hyperrectangles of volume  $1/m$, and of sidelengths in $[ 2^{-d} m^{-1/d}, 2^d m^{-1/d} ]$

The remaining strip is of the form $H'\times[0, \ell']$ and has again  volume in $\frac{1}{m}\mathbb N$. We may apply the induction hypothesis to $H'$, this time with $m'' =m\ell'$, since $m'' |H'|=m\ell' |H'| \in \mathbb N$. We thus     obtain a subdivision of  $H'$ into hyperrectangles $\mathcal R_j''$ of volume $1/m''$ and sidelengths in $[ 2^{-(d-1) } (m'')^{-\frac{1}{d-1}} ,  2^{d-1} (m'')^{-\frac{1}{d-1}}    ]$. Taking the cartesian products of these rectangles with $[0, \ell']$ concludes the proof. 
 The last point of the thesis is also true, because the rectangles $\mathcal R_j$ with a face in common with $F$ are all making up the same strip.

% For $i=3,\ldots,d$ divide the side $[0,\ell_i]$ into $[\ell_i]$ equal intervals $I_k^i$ of length $s_i:=\ell_i/[\ell_i]$. All the $\mathcal R_j$ in the end will be congruent to one of
% \[
%  [0,s_1']\times[0,s_2']\times\prod_{i=3}^d[0,s_i]\quad\text{or}\quad[0,s_1'']\times[0,s_2'']\times\prod_{i=3}^d[0,s_i]\ ,
% \]
% for specially chosen $s_1', s_1'', s_2', s_2''$. For the first two coordinates we proceed as follows. Performing the Euclidean division algorithm we may find two integers $q,r$ such that
% \[
%  \frac{\prod_{i=1}^d\ell_i}{\prod_{i=3}^d[\ell_i]}=q\ [\ell_1] + r\ ,\quad r<[\ell_1]\ .
% \]
% We then obtain
% \[
%  [\ell_1]^2\le\ell_1\ell_2\le q[\ell_1]+ r\ \Rightarrow\ q[\ell_1]\ge[\ell_1]([\ell_1]-1)\ \Rightarrow\ q\ge[\ell_1]-1\ge r\ ,
% \]
% thus we find that for integers $a=q-r,b=r$ there holds
% \[
%  a\frac{[\ell_1]}{\ell_1}\prod_{i=3}^d\frac{[\ell_i]}{\ell_i} +b\frac{[\ell_1]+1}{\ell_1}\prod_{i=3}^d\frac{[\ell_i]}{\ell_i}:=as_2'+ bs_2''=\ell_2\ .
% \]
% We may thus divide $[0,\ell_2]$ into $a$ segments of length $s_2'$ and $b$ segments of length $s_2''$, then along the $\ell_1$-dimension divide the stripes congruent to
% \[
%  [0,\ell_1]\times[0,s_2']\times\prod_{i=3}^d[0,s_i]
% \]
%  into $[\ell_1]$ congruent rectangles and the stripes congruent to
% \[
%  [0,\ell_1]\times[0,s_2'']\times\prod_{i=3}^d[0,s_i]
% \]
%  into $[\ell_1]+1$ congruent rectangles, obtaining 

\end{proof}

\begin{lem}[Correcting fluxes on rectangles]\label{lem57}
Let $H$ be a hyperrectangle of $\R^d$ with sidelengths in $[\ell/2, 3\ell]$. Let $\tilde H:=H\times [-\ell,\ell]^k$.  
 Let $g\in L^2_{\yg}(\p \tilde H) $  and $m$ be a function on $H$ of average $m_0$  such that 
%$$c_d(m-1) |\mathcal{K} | = - \int_{\p \mathcal{K}} g $$ 
\begin{equation}\label{eq:flux m}
 - \c m_0 |H| =    \int_{\p  \tilde H} \yg g. 
\end{equation}
Then the mean zero  solution to
\begin{equation}\label{eqnu}
\left\{\begin{array}{ll}
-\div (\yg \nab  h)  =\c m \drd & \text{in} \ \tilde H\\
\partial_\nu h =g & \text{on} \ \p \tilde H \end{array}\right.
\end{equation}
satisfies
\begin{equation}\label{estlcs2}
\int_{\tilde H} \yg |\nab h|^2 \le  C \ell  \int_{\p \tilde  H}\yg |g|^2+ C \ell^{d+1- \gamma} \|m-m_0\|_{L^\infty (H)}^2,
\end{equation}
where $C$  depends only on $d$ and $s$.
\end{lem}

\begin{proof} First we note that in the case where $k=0$ (which goes along with $\gamma=0$), this was proven in \cite[Lemma 5.8]{rs}. We may thus focus on the case $k=1$.

We may  split $h$ as $h_1 + h_2+h_3$ where $h_1$ solves 
\begin{equation}\label{eqh1}
 \left\{\begin{array}{ll}
   -\op{div}\left(\yg\nab h_1\right)=   \c m_0 \drd
  & \text{ in }\tilde H\ ,\\[3mm]
    \p_\nu h_1=c:=-\frac{\c m_0}{2  \ell^\gamma}& \text{ on }H\times\p[-\ell,\ell]\ ,\\[3mm]
   \p_\nu h_1=0& \text{ on the rest of }\p \tilde H\ ,
\end{array}
\right.
\end{equation}
 $h_2$ solves 
\[
 \left\{\begin{array}{ll}
   -\op{div}\left(\yg\nab h_2\right)=   0 
  & \text{ in }\tilde H\ ,\\[3mm]
    \p_\nu h_2= g- c& \text{ on }H\times\p[-\ell,\ell]\ ,\\[3mm]
   \p_\nu h_2=g & \text{ on the rest of }\p \tilde H\ ,
\end{array}
\right.
\]
and $h_3$ solves \begin{equation}\label{eqh3}
 \left\{\begin{array}{ll}
   -\op{div}\left(\yg\nab h_3\right)=   \c (m-m_0) \drd
  & \text{ in }\tilde H\ ,\\[3mm]
    \p_\nu h_3= 0& \text{ on }H\times\p[-\ell,\ell]\ ,\\[3mm]
   \p_\nu h_3=0 & \text{ on the rest of }\p \tilde H\ ,
\end{array}
\right.
\end{equation}

We note that the first equation  has a unique solution up to constants, and an explicit solution is
$h_1=\lambda  m_0|y|^{1-\gamma}$, with 
$$\lambda= \frac{c}{m_0(1-\gamma)} \ell^{\gamma} =  \frac{-c_{s,d}}{2(1-\gamma)} .$$
Then $|\nabla h_1|= (1-\gamma)\lambda m_0  |y|^{-\gamma}$ and with straightforward computations
\begin{equation}\label{estenh1}
 \int_{\tilde H}\yg|\nab h_1|^2\le C\, m_0^2 \ell^{d+1 -\gamma },
\end{equation}with $C$ depending only on $d$ and $s$.
%and
%\begin{equation}\label{estbh1}
 %\int_{\p \tilde H}\yg|\p_\nu h|^2\le C\, m^2 \ell^d\ell^{d-s-1}\ .
%\end{equation}
Next  we note that $h_2$ can be  obtained as the minimizer with average $0$ of the functional $\frac12\int_{\tilde H}\yg|\nabla h|^2-\int_{\p \tilde H}\yg h \bar g$ where $\bar g$ is the boundary condition.  Comparing $h_2$ with the choice $h=0$ we obtain, using Cauchy-Schwarz,
\begin{eqnarray}
 \frac12\int_{\tilde H}\yg|\nab h_2|^2&\le&\int_{\p \tilde H}\yg |h_2||\bar g| \nonumber\\[3mm]
 &\le&\left(\int_{\p \tilde  H}\yg |h_2|^2\right)^{\frac12} \left(\int_{\p \tilde H}\yg |\bar g|^2\right)^{\frac12}  ,\label{esth3a}
\end{eqnarray}
and we estimate this in weighted spaces via
\begin{equation}\label{wsi}
\int_{\p \tilde H}\yg |h_2|^2\le C \ell \int_{\tilde H} \yg |\nab h_2|^2.\end{equation}
We postpone the proof of this inequality to the end. 
Inserting into \eqref{esth3a}, we obtain
$$\int_{\tilde H} \yg |\nab h_2|^2 \le C \ell  \int_{\p \tilde  H}\yg |\bar g|^2,$$
and inserting the definition of $\bar g $ and $c$, we are led to 
\begin{equation}\label{esth2}\int_{\tilde H} \yg |\nab h_2|^2 \le C \ell  \int_{\p\tilde  H}\yg |g|^2 + C m_0^2 \ell^{d+1-\gamma}.\end{equation} 
Next, we apply Cauchy-Schwarz again to obtain from \eqref{eq:flux m}  that  $$m_0^2 =\frac{1}{\c^2|H|^2} \left(\int_{\p \tilde H}  \yg g\right)^2\le  C \ell^{-d+\gamma} \int_{\p \tilde H} \yg |g|^2$$ and combining with  \eqref{esth2}, we deduce
\begin{equation}\label{resh2}
\int_{\tilde H} \yg |\nab h_2|^2 \le C \ell \int_{\p\tilde  H}\yg |g|^2 .\end{equation}

For $h_3$, we first assume the following trace inequality, whose proof is postponed to the end: 
\begin{equation}\label{trace0}
\int_H |h(x, 0)|^2 \le C \ell^{1-\gamma} \int_{\tilde H} \yg |\nab h|^2.\end{equation}
Assuming this, let us multiply \eqref{eqh3} by $h_3$ and integrate by parts to obtain 
$$\int_{\tilde H}\yg |\nab  h_3|^2 = \c \int_{H \times \{0\}} (m-m_0) h_3 .$$
Combining with the Cauchy-Schwarz inequality and \eqref{trace0}, we easily deduce 
$$\int_{\tilde H}\yg |\nab  h_3|^2 \le C \ell^{1-\gamma} \|m-m_0\|_{L^\infty(H)}^2 |H| = C \ell^{d+1-\gamma}\|m-m_0\|_{L^\infty(H)}^2,$$
which, combined with \eqref{estenh1} and \eqref{resh2}, gives the result.

Let us now prove \eqref{wsi} and \eqref{trace0}.
First we may again reduce to the case $k=1$ (otherwise the inequality is standard). For any function $h$ let us apply for each $y>0$ the standard trace inequality to $h $ on $H\times \{y\}$ to  obtain 
$$\int_{\p H\times \{y\}} |h-h_y|^2 \le C \ell \int_{H\times \{y\}} |\nab_x h|^2$$
where $h_y$ is the average of $h$ on $H\times \{y\}$.
Integrating against $\yg$ and using the triangle inequality in $L^2$, we deduce that 
\begin{multline*}
\int_{\p  H\times [-\ell, \ell]}\yg h^2 \le  2C \ell \int_{\tilde H} \yg |\nab h|^2 + 2\int_{\p H \times [-\ell,\ell]} \yg |h_y|^2\\ \le C \ell \int_{\tilde H} \yg |\nab h|^2 + C \frac{|\p H|}{|H|} \int_{\tilde H} \yg h^2\end{multline*}
where we used that by Jensen's inequality,  $|h_y|^2\le \frac{1}{|H|}\int_{H\times \{y\}} h^2$.
Next,  we note that the Sobolev inequality in weighted spaces for functions of zero average 
\begin{equation}\label{sobo}
\int_{\tilde H} \yg |h|^2 \le C \ell^2 \int_{\tilde H} \yg |\nab h|^2\end{equation}
holds by \cite{fks}, since $\yg$ is an $A_2$ weight.
This yields 
\begin{equation}\label{cote}\int_{\p H \times [-\ell,\ell]} \yg h^2 \le C \ell \int_{\tilde H} \yg |\nab h|^2.\end{equation}
Next, we need to prove the same relation in $H\times \{-\ell, \ell\}$.
For any $x \in H$, let us denote $h_x=\dashint_{[-\ell, \ell]} h(x, \cdot)$.
By Cauchy-Schwarz, we have
\begin{equation}\label{hx}
|h_x|^2\le \frac{1}{4\ell^2}\int_{-\ell}^\ell\yg | h(x, \cdot)|^2 \int_{-\ell}^\ell \frac{dy}{\yg}\le C \ell^{-1-\gamma} \int_{-\ell}^\ell \yg |h(x, \cdot)|^2.\end{equation}
In addition, for each $x \in H$, we may write 
\begin{equation}\label{tracel}
|h(x,\ell) -h_x|^2\le\left( \int_{-\ell}^\ell |\p_y h(x, \cdot) |\right)^2\le \int_{-\ell}^\ell \yg |\p_y h(x, \cdot) |^2 \int_{-\ell}^\ell \frac{dy}{\yg}= C\ell^{1-\gamma} \int_{-\ell}^\ell \yg |\nab h(x, \cdot)|^2.\end{equation}
Integrating this over $x\in H$, we obtain 
$$\int_{H\times \{\ell\}} \ell^\gamma |h-h_x|^2 \le C \ell \int_{\tilde H} \yg |\nab h|^2.$$
On the other hand, integrating \eqref{hx} over $x\in H$ yields 
$$\int_H \ell^\gamma |h_x|^2   \le \ell^{-1}\int_{\tilde H} \yg |h|^2 \le C \ell\int_{\tilde H}\yg |\nab h|^2$$  where we used \eqref{sobo}.
With the triangle inequality, we deduce that 
$$\int_{H\times \{\ell\}} \yg |h|^2 \le C \ell\int_{\tilde H} \yg |\nab h |^2.$$
Combining with \eqref{cote} and using the symmetry, we conclude that \eqref{wsi} holds.
We then note  that we may slightly modify the last proof to obtain \eqref{trace0}: instead of \eqref{tracel} we can write 
$$|h(x, 0)-h_x|^2 \le \left( \int_{-\ell}^\ell |\p_y h(x, \cdot) \right) ^2 \le C \ell^{1-\gamma} \int_{-\ell}^\ell \yg |\nab h(x, \cdot)|^2.$$
On the other hand, as seen above 
$\int_H |h_x|^2 \le C \ell^{1-\gamma} \int_{\tilde H} \yg |\nab h|^2 $, 
so integrating over $H$ and using  the triangle inequality yields \eqref{trace0}.
\end{proof}

The following lemma is straightforward, there we omit the proof. 
 \begin{lem}[Adding a point without flux creation]\label{lemcs1}\mbox{}\\
Let $\mathcal{R}$ be a hyperrectangle in $\R^d$ of barycenter $0$ and sidelengths in $[a,b]$ with $a,b>0$,  let $\tilde {\mathcal R}= \mathcal R\times [-1,1]$, $m=1/|\mathcal{R}|$ and let $X \in B(0, \hal \min(a,b))$.
The mean zero solution to
\begin{equation*}
\left\{\begin{array}{ll}
-\div (\yg \nab h)  = \c\left( \delta_X -m \drd\right)& \text{in} \ \tilde{\mathcal R} \\
\p_\nu h =0 & \text{on} \ \p \tilde{\mathcal{R}}
\end{array}\right.
\end{equation*}
satisfies
\begin{equation*}
\lim_{\eta\to 0}\left| \int_{\tilde{\mathcal{R}} }\yg  |\nab h_\eta|^2 - \c \g(\eta) \right|\le C
\end{equation*}
where $C$ depends only on $d,a,b$. \end{lem}

\begin{lem}[Completing charges near the boundary]\label{chargesnearbdry}Let $\mathcal{R}$ be a hyperrectangle in $\R^d$ of center $0$ and sidelengths in $[a,b]$ with $a,b>0$,  let $\tilde {\mathcal R}= \mathcal R\times [-1,1]^k$. Let $F$ be a face of $\mathcal R$ and $\tilde F:=F\times[-1,1]^k$. Let $\Lambda_{\mathcal R} \subset \mathbb R^d\times\{0\}$ be a discrete set of points  contained in an $\eta$-neighborhood of $\tilde F$, with $\eta< 1$, and $N_p \in \mathbb N^*$ for $p \in \Lambda_{\mathcal R}$. Let $c$ be a constant such that
 \begin{equation}\label{CR}
 c \int_{\tilde F}\yg=\c\int_{\tilde{\mathcal R}}\sum_{p\in\Lambda_{\mathcal R}}N_p\delta_p^{(\eta)}\ .
 \end{equation}
The mean-zero solution to \[
 \left\{\begin{array}{ll}
   -\op{div}\left(\yg\nabla h\right)=\c \sum_{p\in\Lambda_{\mathcal R}}N_p\delta_p^{(\eta)}& \text{ in }\tilde{\mathcal R}\ ,\\[2mm]
   \p_\nu h=0& \text{ on }\p \tilde{\mathcal R}\setminus \tilde F\ ,\\[2mm]
   \p_\nu h=c& \text{ on }\tilde F\ 
\end{array}
\right.
\] 
satisfies 
\begin{equation}\label{estboundary}
 \int_{\tilde{\mathcal R}}\yg|\nab h|^2\le C(\g(\eta)+1)\Big(    \sum_{p  \in \Lambda_{\mathcal R}} N_p\Big) ^2\, 
\end{equation} where $C$ depends only on $d, s, a, b$.
\end{lem}
\begin{proof}
We may write $h=h_1+h_2$ where
\[
 \left\{\begin{array}{ll}
   -\op{div}\left(\yg\nabla h_1\right)=\c\sum_{p\in\Lambda_{\mathcal R}}N_p\delta_p^{(\eta)}-{c_{s,d}m}\drd& \text{ in }\tilde{\mathcal R} \\[3mm]
   \p_\nu h_1=0& \text{ on }\p \tilde{\mathcal R}
\end{array}
\right.
\]
and
\begin{equation}\label{eqcN}
\left\{\begin{array}{ll}
   -\op{div}\left(\yg\nabla h_2\right)={c_{s,d}m}\drd& \text{ in }\tilde{\mathcal R}\\[2mm]
    \p_\nu h_2=0& \text{ on }\p \tilde{\mathcal R}\setminus \tilde F \\ [2mm]
    \p_\nu h_2=c& \text{ on }\tilde F
\end{array}\right.
\end{equation} where $m$ is chosen so that 
$$c_{s,d} m|\mathcal{R}| = \c\int_{\tilde{\mathcal R}}\sum_{p\in\Lambda_{\mathcal R}}N_p\delta_p^{(\eta)} =c   \int_{ \tilde F } \yg. $$ 
In view of Lemma \ref{lem57} (or its proof) we have that 
$$\int_{\tilde{\mathcal R}} \yg |\nab h_2|^2 \le C c^2 \le C (\sum_{p \in \Lambda_{\mathcal R}} N_p )^2 $$
where $C$ may depend on $d,a,b$ and we used  \eqref{CR} to bound $c$.
There remains to control $h_1$. First we note that  $h_1(X)= \sum_{p\in \Lambda_{\mathcal R}} N_p G_\eta(X,p)$ where for $p\in \mr^d\times \{0\}$,  $G_\eta(X,p) $ denotes the solution with zero average on $\mathcal R \times \{0\}$ of 
\begin{equation}\label{eqGeta}
\left\{\begin{array}{ll}
   -\op{div}\left(\yg\nabla G_\eta \right)=c_{s,d} \left(\delta_p^{(\eta)} - (\dashint_{\tilde{\mathcal{R}}} \delta_p^{(\eta)})  \drd\right)   & \text{ in }\tilde{\mathcal R}\\[2mm]
    \p_\nu G_\eta=0& \text{ on }\p \tilde{\mathcal R}\,  .
\end{array}\right.\end{equation}
The desired estimate will thus follow provided we show that 
\begin{equation}\label{estGeta}\int_{\tilde{\mathcal R}} \yg |\nab G_\eta|^2 \le C (\g(\eta)+1).\end{equation}

We recall that the truncated function $\g_\eta(X)= \min (\g(X ), \g(\eta))= \g(X)- f_\eta(X) $ is a solution to 
$-\div (\yg \nab \g_\eta)= \delta_0^{(\eta)}$. 
We may thus compare $G_\eta$ with $\g_\eta(x-p)$. 
If $p$ is at distance $>\frac{a}{4} $ from   the boundary of $\mathcal R$, then $\g_\eta(x-p)$ has Neumann derivative $\vp$ which is bounded on $\p\tilde{ \mathcal R}$, and $G_\eta- \g_\eta(x-p)$ is solution to the equation 
$$\left\{\begin{array}{ll} -\op{div}\left(\yg\nabla u \right)=- c_{s,d} (\dashint_{\tilde{\mathcal{R}} }\delta_p^{(\eta)})     \drd& \text{ in }\tilde{\mathcal R} \\[2mm] 
\p_\nu u =\vp & \text{ on }\p \tilde{\mathcal R}\end{array}\right. 
$$ with $\vp $ bounded. In view of Lemma \ref{lem57} we thus have 
 $$\int_{\tilde {\mathcal R}} \yg |\nab (G_\eta- \g_\eta)|^2 \le C $$
    where $C$ depends only on $d,a,b$.  In that case, the result follows easily since one may compute that $\int_{\tilde {\mathcal R}}  \yg |\nab g_\eta|^2 = \c \g(\eta) + C$.
    
    We now turn to the case where $p$ is  close to the boundary of $\mathcal R$.\\
Without loss of generality we may assume that one vertex of $\mathcal R$ is at the origin in $\R^d$ and that it is the one that $p$ is closest to. We then consider the $2^d-1$ points $ p_1, \dots p_{2^d}$ obtained by symmetry of $p_1:=p$ with respect to the coordinate axes in $\R^d$. In other words if $(p^1, \dots p^d)$ are the coordinates of $p$, we let the $p_i$ be all  the points  with each $j$th coordinate equal  $\pm p^j$.  We define $\mathcal R'$ the rectangle obtained by taking all the same reflections of $\mathcal R$ (its sidelengths are thus double of those of $\mathcal R$, and $\tilde {\mathcal R}'= \mathcal R' \times [-1,1]^k$.  We may write $\mathcal R'= \cup_{i=1}^{2^d} \mathcal R_i$ where $\mathcal R_i$ is the image of $\mathcal R$ by the same symmetry that mapped $p$ to $p^i$. We also let $\tilde{\mathcal R}_i= \mathcal R_i \times [-1,1]^k$.
  \\  
 We can extend $G_\eta$ by multi-reflection to a function on $\tilde {\mathcal R'}$, and  write it as  $u_\eta + v_\eta$ where 
\begin{equation}\label{ueta}
\left\{\begin{array}{ll}
   -\op{div}\left(\yg\nabla u_\eta \right)=c_{s,d} \left(\sum_{i=1}^{2^d} \delta_{p^i}^{(\eta)} \indic_{\tilde{\mathcal R}_i} -  2^d \indic_{\tilde{K}_a } \frac{1}{|\tilde K_a|}  (  \int_{\tilde{\mathcal{R}}} \delta_p^{(\eta)})  \drd\right)   & \text{ in }\tilde{\mathcal R'}\, ,\\[2mm]
    \p_\nu u_\eta=0& \text{ on }\p \tilde{\mathcal R'}\, ,
\end{array}\right.
\end{equation}
where $\tilde K_a$ is a cube centered at the origin of sidelengths $\min(a/2,\hal)$ (recall that we assume that the point $p$ is included in $\tilde K_a$)
\begin{equation}\label{veta}
\left\{\begin{array}{ll}
   -\op{div}\left(\yg\nabla v_\eta \right)=c_{s,d} \left(\ (2^d \indic_{\tilde{K}_a } \frac{1}{|\tilde K_a|}  - 2^d \indic_{\tilde{ \mathcal R'}} \frac{1}{|\tilde{\mathcal R}'|}) (    \int_{\tilde{\mathcal{R'}}} \delta_p^{(\eta)}) \right)  \drd
     & \text{ in }\tilde{\mathcal R'}\, ,\\[2mm]
    \p_\nu v_\eta=0& \text{ on }\p \tilde{\mathcal R'}\, . 
\end{array}\right.\end{equation}
By the estimate in Lemma \ref{lem57}, we have $\int_{\tilde {\mathcal R'}} \yg |\nab v_\eta|^2 \le C$, with $C$ depending only on $d$ and $s$.
Thus there remains to estimate the same for $u_\eta$. For this we build sub and super-solutions to the equation \eqref{ueta}.
\\
For a supersolution, we take 
$\bar{G}_\eta:=\sum_{i=1}^{2^d} \g_\eta(x-p_i)$ which  satisfies 
$$-\div (\yg \nab \bar {G}_\eta) = \c \sum_{i=1}^{2^d} \delta_{p_i}^{(\eta)}  \ge -\div (\yg \nab u_\eta) \quad \text{in} \ \tilde{\mathcal R'}. $$
  For a subsolution we use the explicit solution  $h_M= \frac{\c}{2(1-\gamma)}  M |y|^{1-\gamma}$ as in the proof of Lemma \ref{lem57} but  with $M=-\c 2^d /|\tilde K_a|  \ge \frac{\c}{|\mathcal R|}$, which is  bounded since $\gamma<1$ and  satisfies 
 $$-\div (\yg \nab (u_\eta  -h_M) ) \le 0.$$
 By 
 the maximum principle Lemma \ref{pcpmax}, we deduce   that $\bar{G_\eta} - u_\eta $ and $u_\eta-h_M$  achieve their minimum on the boundary of $\tilde{\mathcal R'}$.  Up to adding a constant to $u_\eta$, we also assume that its minimum in $\tilde {\mathcal R'}$ is $0$. We may thus write 
 $$  h_M + \min_{\p \tilde {\mathcal R'}} (u_\eta - h_M) \le u_\eta \le \bar G_\eta -\min_{  \p \tilde {\mathcal R'}}   (\bar{G}_\eta -u_\eta)$$
 which yields, in view of the properties of $\bar G_\eta $ and $h_M$, \begin{equation}\label{pmp}
  0 \le u_\eta \le C +  2^d \g(\eta) + \max_{\p  \tilde {\mathcal R'}}  u_\eta - \min_{\p  \tilde {\mathcal R'}} u_\eta. \end{equation}
 But $u_\eta$ being the solution of \eqref{ueta} whose right-hand side is $0$ in $\tilde {\mathcal R'} \backslash \tilde K_a$ satisfies a Harnack principle (cf. \cite[Theorem 2.3.8]{fks}), which asserts that $\max_{K} u_\eta \le A \min_K u_\eta$ for each compact set $K$, and for some constant $A>0$ independent of $u_\eta$. Since we have zero Neumann boundary condition we can in principle extend the solution by reflection across the boundary, so this relation holds in fact up to the boundary of $\tilde{ \mathcal R' }$, so for some $K$ containing $\p \tilde {\mathcal R'}$, and we now consider such a $K$.    
 Standard arguments then imply that 
 \begin{equation}\label{mpp2}\max_{\p  \tilde {\mathcal R'}}  u_\eta - \min_{ \p \tilde {\mathcal R'}} u_\eta\le (1-\theta) (\max_{  \tilde {\mathcal R'}}  u_\eta - \min_{ \tilde {\mathcal R'}} u_\eta) = (1-\theta) \max_{  \tilde {\mathcal R'}}  u_\eta   \end{equation}
 for $\theta= \frac{1}{1+A}<1$ where $A$ is the constant in the Harnack inequality.  Indeed, either $u_\eta \ge \theta \max_{\tilde{\mathcal R}}  u_\eta$ in $K$,
  in which case we deduce  $ \max_K u_\eta -\min_K u_\eta \le (1-\theta)  \max_{\tilde{\mathcal R}}  u_\eta$ and this implies \eqref{mpp2};  or 
 else,  $\min_K u_\eta \le   \theta \max_{\tilde{\mathcal R}}  u_\eta$, in which case  the  above Harnack inequality yields  $\max_K u_\eta \le A \theta \max_{\tilde{\mathcal R}}  u_\eta = (1-\theta) \max_{\tilde{\mathcal R}}  u_\eta $ by choice of $\theta$, which again implies \eqref{mpp2}.

Inserting then \eqref{mpp2} into \eqref{pmp}, we deduce 
 $\theta \max_{  \tilde {\mathcal R'}}    u_\eta \le C + 2^d \g(\eta)$ i.e. $$0\le u_\eta \le C(1+ \g(\eta)) \quad \text{in} \ \tilde{\mathcal R'} $$ for some positive constant $C$ depending only on $d$ and $s$.
Multiplying \eqref{ueta} by $u_\eta$, integrating by parts  and inserting this bound, we find\begin{equation*}\int_{\tilde{\mathcal R'}} \yg |\nab u_\eta|^2 =\c \int_{\tilde{\mathcal R'}} u_\eta   \left(\sum_{i=1}^{2^d} \delta_{p^i} ^{(\eta)} -  2^d \indic_{\tilde{K}_a } \frac{1}{|\tilde K_a|}  (  \int_{\tilde{\mathcal{R}}} \delta_p^{(\eta)})  \drd\right) 
\le C (1+ \g(\eta)) ,
\end{equation*} where again $C$ depends only on $d$ and $s$. 
Combining with  the estimate on $v_\eta$, we deduce $\int_{\tilde{ \mathcal R}}\yg |\nab G_\eta|^2\le C(1+ \g(\eta))$ with $C$ depending only on $s$ and $d$, which is the desired estimate \eqref{estGeta}.

\end{proof}

\subsection{Proof of Proposition \ref{screening}}
\noindent
{\bf Step 1.}  \emph{Finding a good boundary. }\\
 We note that there exists a constant $C>0$ depending only on $d$ and a hyperrectangle $K_R'\subset \check{K}_R\subset K_R$ whose faces are  at distance $t \in [\hal \ep R, \ep R]$ from those of $\check{K}_R$ and parallel to them,   such that, denoting by $(\p K_R')_L$ the $L$-tubular neighborhood of its boundary and assuming $\hal \ep R>  2^{d+2}$ (which may be included in the condition \eqref{condR})  we have
 \begin{eqnarray}
  \int_{\partial K_R'\times\mathbb R^k}\yg|E_\eta|^2 &\le& C  \M  \ep^{-1}R^{d-1}\ ,\label{enbd1}\\[3mm]
 \int_{(\partial K_R')_{ 2^{d+2}} \times\mathbb R^k}\yg|E_{\eta'} |^2&\le& CM_{R, \eta'}  \ep^{-1} R^{d-1}\, . \label{enbd2} %\sum_{p\in\Lambda\cap  (\partial K_R')_{2\ell}   }N_p & \le & C \ep R^d \, . \label{enbd3}
 \end{eqnarray}
Indeed,
 we   obtain that $$\int_{(\partial K_R')_{ 2^{d+2}} \times\mathbb R^k}\yg|E_{\eta'} |^2 \le CM_{\eta', R}  \ep^{-1} R^{d-1}$$ hold simultaneously for $\eta$ and $\eta'$   by a pigeonhole principle  on a subdivision of $[\hal \ep R,\ep  R]$ into $ 2^{-(d+4)} \ep R$ pieces of size $2^{d+2}$. We then use the mean value principle to get \eqref{enbd1} (translating the faces of $\p K_R'$ if necessary).  We note that $K_R'$ contains $\{x\in K_R, \dist(x, \p K_R) \ge 2\ep R\}$ as desired. 

In the case $k=1$, by a mean-value argument on \eqref{decrvert}, we find that there exists $\ell \in [\hal \ep^2 R, \ep^2 R]$ such that 
\begin{equation}\label{decv}
\int_{K_R' \times \{-\ell, \ell\}} \yg |E|^2 <  2 e_{\ep,R} \ep^{2} R^{d-1}.\end{equation}
In the sequel, we choose this $\ell$, or in the case $k=0$ we choose $\ell=\ep^2 R$. In all cases we have $\ell \le \ep^2 R$. 
We also  note that by the assumption \eqref{condR}, we have $R>\ep R> \ell>1$.

% We note that writing
% \[
%  D(\mu,\mu):=\int_{\mathbb R^{d+k}}\yg |\nabla h_\mu|^2
% \]
% for $h_\mu$ such that $-\div(\yg\nabla h_\mu)=\mu$, where $\mu$ is a Radon measure in $\mathbb R^{d+k}$, we then have by polari\ellation
% \[
%  D(\mu, \nu)=\int_{\mathbb R^{d+k}}\yg\nabla h_\mu\cdot\nabla h_\nu\ .
% \]
% We next prove the following.

\medskip 
\noindent
\textbf{Step 2.} \emph{Subdividing the domain.}\\
We consider the following regions, on each of which we will perform different constructions to build $\hat E$: 
\begin{eqnarray*}
 D_0&=&K_R'\times[-\ell,\ell]^k\ ,\\
 D_\p&=&(K_R\times[-\ell,\ell]^k)\setminus D_0\ ,\\
 D_1&=& (K_R\times [-R,R])\setminus(D_0\cup D_\p)\ .
\end{eqnarray*}
Of course the set $D_1$ does not exist in the case $k=0$. 
We let $\Lambda_0\subset \Lambda$ be given by those points whose smeared charges touch $\p K_R'$, i.e.
\begin{equation}\label{defl0}
 \Lambda_0=\left\{p\in\Lambda:\ B(p,\eta)\cap (\p K_R')\neq\emptyset\right\}\ .
\end{equation}
We may also denote $\Lambda_{int} = (\Lambda \cap K_R')\backslash \Lambda_0$.
The goal of the construction is to ``complete" outside $D_0$ the smeared out charges whose centers belong to $\Lambda_0$, and   place an additional  
\begin{equation}\label{defnp}
N_\p :=|K_R|-|K_R'| +\frac{1}{\c} \int_{\p D_0} \yg E_\eta\cdot\vec{ \nu} - \int_{D_0^c}  \sum_{p\in \Lambda_0  }N_p \delta_p^{(\eta)}\end{equation}
points  in the set $K_R \backslash K_R'$.
Integrating the relation satisfied by $E_\eta$ over $D_0$ yields that 
\begin{multline*}
-\int_{\p D_0} \yg E_\eta \cdot \vec{\nu} = \c \int_{D_0} \sum_{p \in \Lambda} N_p \delta_p ^{(\eta)}- \c |K_R'|\\= 
\c \sum_{p \in \Lambda_{int}\cup \Lambda_0 } N_p -\c \int_{D_0^c}  \sum_{p\in \Lambda_0} N_p \delta_p^{(\eta)} - \c |K_R'|,\end{multline*}with $\vec{\nu}$ the outer normal,  hence 
$$N_\partial = |K_R|- \sum_{p \in  \Lambda_{int}\cup \Lambda_0 } N_p$$ 
which proves, since $|K_R|\in \mathbb{N}$ that $N_\p$ is indeed an integer.

In the case $k=1$ we also define the constant $C_0$ to be 
\begin{multline}\label{defC0}
C_0=  \frac{\ell^{-\gamma}}{2(|K_R|-|K_R'|)} \int_{D_0 \times \p [-\ell,\ell]} \yg E_\eta \cdot \vec{\nu}\\= 
\frac{\ell^{-\gamma}}{2(|K_R|-|K_R'|)}\left( \c N_\p - \int_{\p K_R' \times [-\ell,\ell]} \yg E_\eta \cdot \vec{\nu} +\c \int_{D_0^c}  \sum_{p\in \Lambda_0} N_p \delta_p^{(\eta)}\right) -\hal \ell^{-\gamma}
\end{multline}
where for the second equality we have used \eqref{defnp}.
Whenever $C_0$ appears below, we will mean $C_0=0$ in the case $k=0$.

Next, we split $K_R \backslash K_R'$ into hyperrectangles $H_i$ with sidelengths $\in [\ell/2,2\ell]$ (this can be done by constructing successive strips, in the style of the proof of Lemma \ref{subdivision}). We then let $\tilde H_i=H_i \times [-\ell,\ell]^k$.
We also let 
$$n_i = \c \int_{\tilde H_i}\sum_{p\in \Lambda_0} N_p \delta_p^{(\eta)}, $$ and let $m_i$ be such that 
 \begin{equation}\label{defmi}
 \c(m_i-1) |H_i| =  \int_{\p D_0 \cap \p \tilde H_i}  \yg E_\eta \cdot\vec{ \nu} - 2 C_0\ell^{\gamma}  |H_i|- n_i.\end{equation}We note that $n_i=0$ unless $H_i$ has a face in common with $\p D_0$.
 We will check below that $|m_i-1|<\hal$. Since the sidelengths of $H_i$ are of order $\ell$, as soon as $\ell $ is large enough, modifying the boundaries of the $H_i$ a little bit, we can ensure that in addition each $m_i |H_i|\in \mathbb{N}$.
 Indeed, combining \eqref{defmi} and \eqref{defnp}, we may check that $\sum_i m_i |H_i|= N_\p$, hence is an integer. Once we prove later that $|m_i-1|<\hal$, this will imply that $N_\p >0$.
 
Since $m_i\in [\hal, \frac{3}{2}]$ and $m_i |H_i|\in \mathbb N$, if $\ell >\ell_0= 2 (\frac{1}{2})^{-1/d}$, we may apply Lemma \ref{subdivision}.   This places a condition on $\ep^2R$ being large enough, which is fulfilled by taking $R_0$ large enough in \eqref{condR}.  For the $\tilde H_i$ which have some codimension-$1$ face $F_i$ in common with $\partial D_0$, we choose such an $ F_i$ to play the role of $ F$ in the notation of Lemma \ref{subdivision}. Each $H_i$ is   then divided into rectangles $\mathcal R_\alpha, \alpha\in I_i$ of sidelengths bounded above and below by positive constants which depend only on $d$, and  volumes $1/m_i$.
We let $\tilde{\mathcal R}_\alpha= \mathcal R_\alpha \times [-1,1]^k$.    For each $\mathcal R_\alpha$, we denote 
\begin{equation}\label{defna}
\bar n_\alpha= \c \int_{\tilde{\mathcal R}_\alpha}\sum_{p\in \Lambda_0} N_p \delta_p^{(\eta)} \end{equation} 
so that $n_i = \sum_{\alpha \in I_i} \bar n_\alpha.$  
 \medskip 
 
 \noindent
 {\bf Step 3.} \emph{Defining $\hat E_\eta$}.\\
 Over each $\tilde H_i$ we define $\hat E_\eta $ as a sum $E_{i,1}+E_{i,2}+E_{i,3}+E_{i,4}$, some of these terms being zero except for $H_i$ that has some boundary in common with $\p D_0$. 
 
The first vector field contains the contribution of the completion of the smeared charges belonging to $\Lambda_0$.
We let
 $$E_{i,1}:=\sum_{\alpha\in I_{i,\p}}\indic_{\tilde{\mathcal{R}}_\alpha} \nab h_{1,\alpha}$$ where $h_{1,\alpha}$ is the solution of
\begin{equation}\label{defh1}
\left\{\begin{array}{ll}
 -\op{div}\left(\yg\nabla h_{1,\alpha}\right)=\c \sum_{p\in\Lambda_\alpha}N_p\delta_p^{(\eta)}& \text{ in } \tilde{\mathcal R}_\alpha 
,\\[3mm]
 \p_\nu h_{1,\alpha}=0& \text{ on }\p   \tilde{\mathcal R}_\alpha  \setminus \p D_0\ ,\\[3mm]
 \p_\nu h_{1,\alpha}=\frac{-\bar n_\alpha }{\int_{ \p \tilde{\mathcal{R}}_\alpha\cap \p D_0}     \yg }    &\text{ on }\p \tilde{\mathcal{R}}_\alpha\cap \p D_0 ,
\end{array}
\right.
\end{equation}
We note that the definition of $\bar n_\alpha$ makes this equation solvable and that by construction of the $\mathcal R_\alpha$, the constant $\int_{ \p \tilde{\mathcal{R}}_\alpha\cap \p D_0}     \yg $ is bounded above and below by positive constants depending only on $d$ and $s$. 

The second vector  field is defined to be $E_{i,2}= \nab h_2$ with 
\begin{equation*}
\left\{\begin{array}{ll}
   -\op{div}\left(\yg\nabla h_2\right)=- 2\ell^{\gamma} C_0\drd& \text{ in }\tilde H_i\ ,\\[3mm]
   \p_\nu h_2=C_{0}&\text{ on }H_i\times\p[-\ell,\ell]\ , \\[3mm]
   \p_\nu h_2=0&\text{ on }\text{the rest of }\p \tilde H_i\ ,
\end{array}
\right.
\end{equation*} Of course, $E_{i,2}=0$ in the case $k=0$.

The third vector field is defined to be $E_{i,3}=\nab h_3$ with  
\begin{equation*}
\left\{\begin{array}{ll}
   -\op{div}\left(\yg\nabla h_3\right)= (\c( m_i-1) + 2 \ell^{\gamma} C_0) \drd& \text{ in }\tilde H_i\ ,\\[3mm]
   \p_\nu h_3=0&\text{ on }H_i\times\p[-\ell,\ell]^k\ , \\[3mm]
   \p_\nu h_3= g_i
  &\text{ on }\p  H_i\times [-\ell,\ell]^k ,
\end{array}
\right.
\end{equation*} where we let $g_i=0$ if $\tilde H_i$ has no face in common with $\p D_0$ and otherwise 
$$ g_i= - E_\eta \cdot \vec{\nu} +  \sum_{\alpha \in I_i} \indic_{\tilde{\mathcal{R}_\alpha}  }\frac{ \bar n_\alpha}{ \int_{ \p \tilde{\mathcal{R}}_\alpha\cap \p D_0}     \yg }$$ 
with $E_\eta \cdot \vec{\nu}$ taken with respect to the outer normal to $\p D_0$.  We note that this is solvable in view of \eqref{defmi} and the definition of $\bar n_\alpha$ and $n_i$.

Finally, the fourth vector field consists of contributions from almost equally spaced and screened charges over the $\mathcal R_\alpha$. To define it, let $p_\alpha$ be the barycenter of  each hyperrectangle $\tilde{\mathcal R}_\alpha, \alpha\in I_i$ and define a function $h_{4,\alpha}$ as solving \begin{equation}\label{defh4}
 \left\{\begin{array}{ll}
   -\op{div}\left(\yg\nabla h_{4,\alpha}\right)=c_{s,d} \left(\delta_{p_\alpha}^{(\eta)}- m_{i}\drd\right)& \text{ in }\tilde{\mathcal R}_\alpha
   \\[3mm]
   \p_\nu h_{4,\alpha}=0& \text{ on }\p \tilde{\mathcal R}_\alpha\ .
\end{array}
\right.
\end{equation}
We note that this equation is solvable because we have chosen $|\mathcal R_\alpha|=1/m_i$.
We then define in $\tilde H_i$
\[
 E_{i,4}=
         \sum_{ \alpha \in I_i }\indic_{\tilde{\mathcal{R}}_\alpha}     \nabla h_{4,\alpha} \, .
\]
This finishes defining $\hat E_\eta $ over $D_\p$. 
To define $\hat E_\eta$ over $D_1$ (in the case $k=1$) we  let $\hat E_\eta=\nab h$ with 
\begin{equation}\label{defh}
\left\{\begin{array}{ll}
   -\op{div}\left(\yg\nabla h\right)=0& \text{ in }D_1\ ,\\[3mm]
   \p_\nu h=- \phi& \text{ on }\p D_1\ .\\
\end{array}
\right.
 \end{equation}
where 
\begin{equation}\phi:=\indic_{\p D_1\cap\p D_0} E_\eta \cdot \vec{\nu}+\indic_{\p D_\p\cap\p D_1}C_0,
\end{equation}
with the outer normal taken to be outer to $D_0$. Again this equation is solvable in view of \eqref{defC0}.
This completes the construction of $\hat E_\eta$. 

We note that  the normal components are always constructed to be continuous across interfaces,  so that no divergence is created there, and so    $\hat E_\eta$ thus defined (and extended by $0$ outside of $D_0 \cup D_\p \cup D_1$) satisfies 
\begin{equation}
\left\{\begin{array}{ll}
   -\op{div}\left(\yg \hat E_\eta \right)=   c_{s,d} \left(\sum_{p \in \hat\Lambda } N_p \delta_p^{(\eta)}   - \drd\right) & \text{in } K_R \times \R^k  ,\\[3mm]
   \hat E_\eta \cdot \nu = 0  
   & \text{ on } \p K_R \times \R^k \ ,
\end{array}
\right.
 \end{equation}
where 
$$\hat \Lambda=\left( \Lambda \cap K_R' \backslash \Lambda_0 \right)\cup \Lambda_0 \cup \left( \cup_{i} \cup_{\alpha \in I_i}  \{p_\alpha\}\right). $$ We note that by construction  the distance between the new  points of $\hat \Lambda$ and between them and $\p K_R $ and them and the original points  is bounded below by a constant depending only on the sidelengths of $\mathcal R_i$, hence on  $d$, call it $2\eta_0$. So if the original points of $\Lambda$ are simple and well separated, so are those of $\hat \Lambda$. 
We then define 
$$\hat E= \Phi_\eta^{-1}( \hat E_\eta) = \hat E_\eta+\sum_{p\in \Lambda_{R, \eta}}N_p \nab f_\eta(x-p)$$ where $f_\eta $ is as in \eqref{feta}. Since $\hat \Lambda$ is at distance $\ge  \eta_0 $ from $\p K_R$, and since 
$f_\eta$ is supported in $B(0, \eta)$ with $\eta<\eta_0$,  we have that $\hat E =\hat E_\eta$ at distance $\ge  \eta$ from $\hat \Lambda$ and in particular on $\p (D_0\cup D_{\p})$.  In particular $\hat E$ solves
\begin{equation}\label{C2}
\left\{\begin{array}{ll}
   -\op{div}\left(\yg \hat  E \right)=   c_{s,d} \left(\sum_{p \in\hat  \Lambda  } N_p \delta_p   - \drd\right) & \text{in } K_R \times \R^k  ,\\[3mm]
  \hat  E \cdot\vec{ \nu} =  0 
   & \text{ on }\p K_R \times \R^k 
\end{array}
\right.
 \end{equation}
as desired.  We note that we may as well choose for $p_\alpha$, instead of the barycenter of the $\widetilde{\mathcal R}_\alpha$, any point at distance $<\eta_0/2$ from it, and have the same conclusions verified if $\eta<\eta_0$.
\medskip 

\noindent
{\bf Step 4.} {\it Controlling the constants.}\\
To control $C_0$, we use \eqref{decv} to obtain with the Cauchy-Schwarz inequality,  from \eqref{defC0},
\begin{equation}\label{controleC0}
|C_0| \le  C \ell^{-\gamma} R^{-d} \ep^{-1}  e_{\ep,R}^{1/2} \ep R^{\frac{d-1}{2}} \ell^{\gamma/2} R^{d/2}  = C e_{\ep,R}^{1/2} \ep^{-\gamma} R^{- \frac{\gamma+1}{2}}.\end{equation}

To control $\bar n_\alpha$, we note that $|\bar n_\alpha|\le C n_\alpha$ where $n_\alpha:=  \sum_{p \in \Lambda_0, \dist (p, \mathcal R_\alpha)\le \eta} N_p$.  
Since the sidelengths of $\mathcal{R}_\alpha$ are bounded by $2^d |m_i|^{-1/d}$ and $m_i >\hal$,   we may choose $L\ge 1$ such that the $L$-fattened 
rectangles $(\mathcal R_{\alpha})_L$ are contained in $(\p K_R')_{ 2^{d+2}}$. Combining Lemma \ref{lemdiscrepance} with  \eqref{enbd2} and recalling that 
$p_\alpha$ is the barycenter of $\mathcal R_\alpha$, we obtain
\begin{multline} \sum_i\sum_{\alpha\in I_{i,\p}}(\bar n_\alpha)^2\le C 
 \sum_i\sum_{\alpha\in I_{i,\p}}n_\alpha^2 \le C\sum_i\sum_{\alpha\in I_{i,\p}}(n_\alpha - L^d)^2
 \le C\sum_{i,\alpha}D(p_\alpha, L)^2\\ \le C\sum_{i,\alpha}\int_{(\mathcal R_\alpha)_L}\yg|E_{\eta'}|^2
  \le C \int_{(\p K_R')_{2\cdot 3^d}  \times \mathbb R}\yg|E_{\eta'}|^2
 \le C M_{R,\eta'}\ep^{-1}R^{d-1}\ .\label{estna}
\end{multline}

Finally, to control $m_i$ we write that in view of \eqref{defmi}
\begin{equation}
|m_i-1| \le C \ell^{-d} \left(   \int_{\p D_0 \cap \p \tilde H_i} \yg |E_\eta|   \right) + 2\ell^{\gamma} |C_0|+ \ell^{-d} \sum_{\alpha \in I_i} \bar n_\alpha.\end{equation}
We note that the last term can be controlled by \eqref{estna} using that $n_\alpha\le n_\alpha^2$ since $n_\alpha$ is an integer, by $ C\Mp \ep^{-3} R^{-1} $. This can be made small if \eqref{condR} holds, choosing $R_0$ large enough there.
The second term (when it exists) can be bounded by $ e_{\ep,R}^{1/2}R^{\frac{\gamma-1}{2}}   \ep^{\gamma}$  which is small as soon as $R_0$ is chosen large enough in \eqref{condR}.
The first term can be bounded by Cauchy-Schwarz and \eqref{enbd1}, first in the case $k=0$ by 
$$C\ell^{-d} \M^{1/2} \ep^{-1/2} R^{\frac{d-1}2} \ell^{\frac{d-1}{2}}\le C \M^{1/2} \ep^{-d-\frac{3}{2}} R^{-1}$$
and in the case $k=1$ by
$$\ell^{-d} \M^{1/2} \ep^{-1/2} R^{\frac{d-1}2} \ell^{\frac{d-1}{2}} \ell^{\frac{1+\gamma}2}\le \M^{1/2}  R^{\frac{\gamma-1}{2}}\ep^{-\hal -d+\gamma}. $$
These terms are all small as soon as \eqref{condR} holds with $R_0$ chosen large enough.
\medskip

\noindent
{\bf Step 5.} {\it Estimating the energy of $\hat E_\eta$}.
\\
For $l=1, \cdots, 4$, using the previous notation, we define $E_l:= \sum_i E_{i,l}$. 
We control successively the energy of each $E_l$.
For $\alpha \in I_{i, \p D}$,
we use Lemma \ref{chargesnearbdry} and from the estimate \eqref{estboundary} therein we obtain,\begin{equation}\label{esthna}
\int_{\tilde R_\alpha}\yg|\nab h_{1,\alpha}|^2 \le C(\g(\eta)+1) n_\alpha^2 .
\end{equation}  From \eqref{estna} we thus obtain (absorbing the $1$ into the $\g(\eta)$) that the total contribution is \begin{equation}\label{este1}
 \int_{K_R\backslash K_R' \times [-\ell, \ell]^k }\yg|E_1|^2=\sum_{i}\int_{\tilde H_i}\yg|\nab h_{i,1}|^2\le C\Mp \g(\eta)\ep^{-1}R^{d-1}\le C \Mp\g(\eta)\ep R^d
\end{equation}where for the last inequality we used \eqref{condR}.

For $E_2$, we use Lemma \ref{lem57}, and control the number of $H_i$'s by $C\ep R^d \ell^{-d}$ to obtain, with \eqref{controleC0} 
\begin{equation}\label{cE2}
\int_{(K_R\backslash K_R') \times [-\ell, \ell]^k}   \yg |E_2|^2 \le C \ep R^d \ell^{-d}    C_0^2 \ell^{\gamma+1+d}  \le C e_{\ep,R} \ep^{3} R^d.\end{equation}

For $E_3$ we use Lemma \ref{lem57} to get 
\begin{multline}\label{cE3}
\int \yg |E_3|^2 \le C \ell \sum_i \int_{\p  H_i \times [-\ell,\ell]^k}  \yg g_i^2 \\
\le  C\ell \int_{\p K_R' \times \R^k}   \yg |E_\eta|^2 +  C\ell \sum_{i, \alpha \in I_i} (\bar n_\alpha)^2 
\le C \ell (\M +\Mp) \ep^{-1} R^{d-1}\le C (\M + \Mp) \ep R^d.
\end{multline}
where we have used \eqref{enbd1},  \eqref{estna}, the geometric properties of $\tilde{\mathcal R}_\alpha$.

For $E_4$ we use Lemma \ref{lemcs1} and multiply by the number of  $\mathcal R_\alpha$, which is  proportional to the volume of the region, hence  bounded by $C \ep R^{d}$, to obtain 
\begin{equation}\label{este4}
 \int_{\tilde K_R}\yg|E_4|^2 = \sum_{i, \alpha\in I_i}\int_{\tilde R_\alpha}\yg|\nab h_{4,\alpha}|^2\le C\g(\eta) \ep R^d\, .
\end{equation}

Finally,  in $D_1$, we use the obvious analogue of Lemma \ref{lem57} to obtain 
\begin{multline}\label{cD1}
\int_{D_1} \yg |\hat E_{\eta}|^2\le  C R  \int_{\p D_1}  \yg |\phi|^2 \le  CR \left(\int_{\p D_0 \cap \p D_1} \yg |E_\eta \cdot \nu|^2 +   C_0^2 \ell^\gamma
(|K_R|-|K_R'|) \right) \\ \le  Ce_{\ep,R} \ep^2    R^d    + C e_{\ep,R}  \ep R^{d}  \le C e_{\ep,R} \ep R^d 
\end{multline}
where we used  \eqref{decv} and \eqref{controleC0}.

Combining \eqref{este1}--\eqref{cD1} we conclude that 
$$\int_{K_R\times \R^k} \yg |\hat E_\eta|^2 \le  \int_{K_R' \times \R^k} \yg |E_\eta|^2 + C ((\Mp+1)  \g(\eta) \ep R^d + ( \M +e_{\ep,R}) \ep R^d)
$$
and since it is clear by construction that $\hat E_\eta= \Phi_\eta (\hat E)$, 
the desired result holds. This concludes the proof of Proposition \ref{screening}.

\begin{remark}It follows from the discussion in the proof that if the points in $\hat\Lambda \cap (K_R\backslash K_R')$ are each displaced by a distance $\le \eta_0/4$ (with $\eta_0$ as in the statement of the proposition), there exists a vector field $\hat E$ compatible with the modified configuration and  satisfying  the exact same conclusions.  Indeed one may displace the barycenters $p_\alpha$ of the $\tilde{\mathcal R}_\alpha$ in the proof by a quarter of their minimal distance without changing  the conclusions. 

\end{remark}

\section{The upper bound on the energy}\label{sec6} We are now in a position to use Proposition \ref{screening} to construct a test-configuration approximating $\min \w$. 

We will need the following result along the lines of \cite[Lemma 5.4]{ss1d}.
\begin{lem}\label{ss1d}\mbox{}\\ Assume we are in one of the cases for which $k=1$. 
There exists $E$ a minimizer of $\W_\eta$ over $\mathcal A_1$ which satisfies 
$$
 \lim_{z\to\infty}\lim_{R\to \infty}\frac{1}{R^d}\int_{K_R\times(\mathbb R\setminus(-z,z))}\yg|E|^2 =0\ .
$$
 \end{lem}
\begin{proof}
We claim that there exists  $P$, an $\R^d$-translation invariant probability measure on vector fields $P$ which is concentrated on minimizers of $\mathcal{A}_1$.
Assuming this, the result follows by applying the multi-parameter ergodic theorem as in \cite{becker}  on the cubes $K_R= [-R/2,R/2]^d $ to the function 
$f_z(E)=\ \int_{[-1,1] \times \{|y|>z\} } \yg |E|^2 $ with $z>1$. 
Indeed, the multi-parameter ergodic  theorem yields that 
$\int f_z(E) \, dP(E) = \int f_z^* (E) \, dP(E)$ with 
$$f_z^*(E):=  \lim_{R\to \infty} \frac{1}{|K_R|} \int_{ K_R }  f( E(x+ \cdot) ) \, dx  $$
and we may check that $$f_z^*(E)\ge C   \lim_{R \to \infty}   \frac{1}{R^d} \int_{ K_{R} \times \{|y|>z\} }  \yg |E|^2. $$ 
But for any $E \in \mathcal A_1$,  the family of function $\{f_z(E)\}_{z>1}$ decreases to $0$ as $z$ increases to $  + \infty$ and is dominated by $f_1$, thus 
by dominated convergence, we have $\lim_{z\to 0} \int f_z(E)\, dP(E)=0$.
Combining with the above, it follows that 
$$\lim_{z\to \infty} \int \left(  \lim_{R \to \infty}   \frac{1}{R^d} \int_{ K_{R} \times \{|y|>z\} }  \yg |E|^2 \right) dP(E)=0 $$ 
and by Fatou's lemma, for $P$-a.e. $E$ we must have 
$$\lim_{z\to \infty}  \lim_{R \to \infty}   \frac{1}{R^d} \int_{ K_{R} \times \{|y|>z\} }  \yg |E|^2=0.$$ The result follows, since $P$ is concentrated on minimizers of $\W_\eta$ on $\mathcal A_1$.

For the existence of $P$, we start from some $E$ minimizing $\W_\eta$ over $\mathcal A_1$. The existence of such an $E$ can be proven exactly as in \cite[Appendix]{rs}.   We let  $K_n $  be the cubes $[-n/2,n/2]^d$ in $\R^d \times \{0\}$, and  $P_n$ be the push forward of the normalized Lebesgue measure on $K_n$ by the map 
$x \mapsto E(  (x,0) + \cdot)$.  We also let  as in Section \ref{sec4}
$$\mathbf f_{n,\eta} ( \mathcal Y)= \int_{\R^{d+1}} \chi \yg |\mathcal Y|^2  $$
if $\mathcal Y(X) = E_\eta (X+ (x , 0))$ for some $x \in K_n$, and $+\infty$ otherwise. 
Then the exact same arguments as in Section \ref{sec4} apply to this setting (the only difference is that we have from the beginning an infinite number of points), in particular we have the obvious compactness result analogous to Lemma  \ref{rslem52}. This ensures that the abstract result  \cite[Theorem 3]{gl13}, can be applied and it yields at the end 
$$\liminf_{n\to \infty} \frac{1}{|K_n|} \int_{K_n \times \R} \yg |E_\eta|^2 \ge \int \W_\eta(E) \, dP(E)$$
with $P$ an $\R^d$-translation invariant probability measure. The left-hand side in this relation is $\W_\eta(E)=\min \W_\eta $ by definition.  It follows that
 $P$ is concentrated on minimizers of $\W_\eta$.
\end{proof}
\begin{remark} With the same reasoning and Fatou's lemma, we could prove the same for a minimizer of $\W$ itself.\end{remark} 
This proposition implies that  we can find a minimizer $E$ of $\W_\eta$ over  $\mathcal A_1$ for which given $\ep$, we have  $e_{\ep, R}\le 1$ as soon as $R$ is large enough (in terms of $\ep$).  We may now  conclude  the proof of Proposition \ref{prow}. For the existence of minimizers of $\W$ and $\W_\eta$ over $\mathcal A_1$, as mentioned we may argue as in \cite[Appendix]{rs}. 
The existence of sequences of periodic minimizers then follow exactly as in \cite{ss2d,rs} (so we don't give all details): we take a minimizer of $\W_\eta$ over $\mathcal{A}_1$ satisfying the results of  Lemma  \ref{ss1d}, apply Proposition \ref{screening} to it with some $R$ large enough, then multi-symmetrize  it  by reflexion and periodize the result to obtain some periodic $E_{R, \eta}$ with almost the same energy, i.e. given $\ep>0$,  if  $R$ is large enough depending on $\eta$ and $\ep$, we have $\W_\eta(E_{R, \eta}) \le \min_{\mathcal A_1} \W_\eta+\ep$.  This implies that there is a periodic minimizing  sequence for $\min \W_\eta$ over $\mathcal A_1$. Taking a diagonal sequence   $\eta \to 0$, $R\to \infty$, we may also  conclude that there exists a periodic minimizing sequence for $\W$, and this finishes the proof of  Proposition \ref{prow}.
\\

We are now going to prove the matching upper bound corresponding to the lower bound given in Proposition \ref{proergodic}:
\begin{pro}\label{upperbound}Assume \eqref{assv1}--\eqref{assv3}  and \eqref{assumpsigma}--\eqref{assmu2}.
  For any $\varepsilon >0$,  there exists $r>0$ such that for $n\in \mathbb N^*$ there exists $A_n\subset (\mathbb R^d)^n$ with $|A_n| \ge n!\left( \frac{\pi r^d}{n}\right)^n$ 
   such that 
  for all $(x_1,\ldots,x_n)\in A_n$ there holds
  \begin{equation}\label{bsup}
\limsup_{n\to \infty}   n^{-1-\frac{s}{d}} \left(  H_n(x_1, \dots, x_n) - n^2 \I(\muv) \right) \le \min \widetilde{\W} +\varepsilon\, , 
\end{equation} in the case \eqref{kernel}, respectively
  \begin{equation}\label{bsup2}
\limsup_{n\to \infty}   n^{-1} \left(  H_n(x_1, \dots, x_n) - n^2 \I(\muv)+ \frac{n}{d} \log n  \right) \le \min \widetilde{\W} +\varepsilon\, , 
\end{equation}in the cases \eqref{wlog}--\eqref{wlog2d}. 
\end{pro}
The proof follows precisely the strategy of \cite[Section 6]{rs},  with some care to be taken about the extra dimension introduced in our case. The main difficulty comes from the possible degeneracy of $\muv$ near $ \p \Sigma$ that we allow, as in \cite{ss1d}.  Indeed, we need to partition $\Sigma'= n^{1/d} \Sigma$ into nondegenerate regions in which $\int \muv' \in \mathbb N$, in which we can  paste a screened minimizer of $\W$ provided by Proposition \ref{screening}. When $\muv'$ becomes very small these regions may have to become very large, or worse, very elongated in some direction, and this would prevent the application of Proposition \ref{screening}.  
As in \cite{ss1d} this is overcome by  allowing an exceptional narrow boundary layer in which the construction is less optimal but induces only negligible errors.

\begin{proof}
\noindent
{\bf Step 1.} \emph{Subdividing the domain.}\\
For $t>0$ we define the tubular neighborhood of $\p \Sigma'$ and its boundary  to be
$$\Sigma_t'=\{ x\in \Sigma', \dist(x, \p \Sigma')<t\}\qquad \Gamma_t= \{ x\in \Sigma', \dist (x, \p \Sigma')=t\}.$$
Since \eqref{assumpsigma} holds, $\Gamma_t$ is $C^1$ for $t<t_c$ small enough.

Pick $1>\bm>0$  a small number. By assumption \eqref{assmu2}, if $\alpha>0$ in that assumption,  rescaling by $n^{1/d}$, if  $\dist (x, \p \Sigma') \ge\frac{n^{1/d}}{\cc^{1/\alpha} }\bm^{1/\alpha} $ where $\cc$ is the constant in \eqref{assmu2},  then 
$ \muv'(x)\ge  \bm$ while if $\dist (x, \p \Sigma') \le\frac{n^{1/d}}{\cc}\bm^{1/\alpha} $  we have $\muv'(x) \le \frac{\cC}{\cc} \bm$.
 We may even find $T \in [\frac{n^{1/d}}{\cc^{1/\alpha}}\bm^{1/\alpha} , \frac{2 n^{1/d} }{\cc^{1/\alpha} } \bm^{1/\alpha}]$ such that 
 $\muv'(\Sigma_T') \in  \mathbb N$, and  $\muv' \ge \bm $ in $\Sigma_T' $.  We note that we may have taken $\bm $ small enough so that $T<t_c$ and 
 \begin{equation}
 \label{longueur}
\hal \mathcal{H}^{d-1} (\p \Sigma') \le  \mathcal{H}^{d-1} (\p \Sigma_t') \le 2  \mathcal{H}^{d-1} (\p \Sigma') \quad \text{for all } t\le T.\end{equation}
If $\alpha=0$ in assumption \eqref{assmu2} then $\muv'$ is bounded below by a positive constant  on its support and we simply take $T=0$. 
For shortness, in what follows we denote $\mathcal{H}^{d-1}(\p \Sigma') $ by $|\p \Sigma'|$.
\smallskip 

\noindent
{\bf Step 2.} \emph{Defining a vector field  in $\Sigma_T'$.}\\
In the region $\Sigma_T'$, we have the lower bound $\muv \ge \bm$ and there is no degeneracy. We may then proceed as in \cite{ss2d,rs}.
 We start by subdividing $\Sigma_T'$  into rectangles of size comparable to $R$ for $R$ large enough, then producing a rescaled version of the construction of Proposition~\ref{screening} on each  such rectangle. We make sure the points all remain well-separated so as  to ensure that we stay in the equality case in Lemma \ref{prodecr} and in formula \eqref{binspl}.

The next lemma is a straightforward modification  of \cite[Lemma 6.5]{ss2d}.
\begin{lem}[Tiling the interior of the domain]\label{tiling}
 There exists a constant $C_0>0$ depending on $\underline{m}$ and $\overline{m}$  such that, given any $R>1$, there exists for any $n\in\mathbb N^*$ a collection $\mathcal K_n$ of closed hyperrectangles in $\Sigma_T'$ with disjoint interiors,  whose sidelengths are between  $ 2R$ and $2(R+C_0/R)$,  and which are such that
\begin{equation}\label{tile1}
\left\{x\in \Sigma_T': d(x,\p \Sigma_T')\ge C_0 R \right\} \subset  \bigcup_{K\in\mathcal K_n} K:= \Sigma'_{\rm{int}},
\end{equation}
\begin{equation}\label{tile2}
\bigcup_{K\in\mathcal K_n} K \subset \left\{ x\in \Sigma': d(x,\p \Sigma')\ge\hal  C_0 R \right\},
\end{equation}and
\begin{equation}\label{inttile}
\forall K\in\mathcal K_n, \quad \int_K \mu_V' \in \mathbb N.
\end{equation}
\end{lem}
For each $K\in \mathcal K_n$ as above we denote
\begin{equation}\label{mk}
m_K:= \dashint_K \mu_V'\, , 
\end{equation} and note that $m_K|K|\in \mathbb N$.  By Proposition \ref{prow}, given $\ep>0$ as in the statement of Proposition \ref{upperbound}, we may find a periodic $E$ in $\mathcal{A}_1$ such that $\W(E) \le \min_{  \mathcal A_1}\W+ \frac{\ep}{4\overline{m}^{s/d}}$. 
If $k=1$, we note that by periodicity of $E$, we must have 
\begin{equation}\label{decrper}
\lim_{z\to \infty}\lim_{R\to \infty} \frac{1}{R^d } \int_{K_R\times (\R\backslash (-z,z)) } \yg |E|^2 =0,\end{equation} hence the assumptions of Proposition \ref{screening} are satisfied for $E$ if $R$ is large enough. Also, by Proposition \ref{pro53}, the points in $E$ are well-separated, and the screening construction preserves that property, and makes the points well separated from the boundary.  By Proposition \ref{Wbb} we have
\begin{equation}\label{stst}
\W_\eta(E) \le \W(E) + C \eta^{\frac{d-s}{2}}.\end{equation}

For each $K\in \mathcal K_n$ we thus  apply Proposition \ref{screening} to $E$
with $R m_K^{1/d}$ and  over the hyperrectangle $m_K^{1/d} K$. We may do this  since  
$m_K |K|= \int_K \muv'$ is an integer by \eqref{inttile}, and its sidelengths are indeed in $[  2R m_K^{1/d}, 3R m_K^{1/d}]$ if $R$ is chosen large enough depending on $C_0$.
 We call the result of this screening $\hat E$ and we define 
$E_K $ to be 
\begin{equation}\label{premdef}
E_K = \nab h_K +\sigma_{m_K} \hat E\quad \text{ on } K\times\mathbb R^k\ ,
\end{equation}
 where  $\sigma_m$ denotes the rescaling of a vector field to scale $m$ as follows:
\[
\sigma_m E= m^{\frac{s+1}{d}} E\left(m^{1/d}\ \cdot\ \right)
\]
and   $h_K$ is introduced
  to correct for the difference between $m_K$ and $\mu_V'$ as the solution to 
\[
\left\{\begin{array}{ll}
-\div(\yg\nab h_K)= c_{d,s}  \left(m_K - \mu_V'\drd \right) & \text{ in }  K\times [-1,1]^k\ ,\\[3mm]
\p_\nu h_K =0 & \text{ on } \p ( K\times [-1,1]^k)  .\end{array}\right.
\]
By Lemma \ref{lem57}, we have
\begin{equation}\label{esthk}
\left\|\nab h_K\right\|_{L^2(K\times\mathbb R^k, \yg)} \le C_R \left\|m_K- \mu_V'\right\|_{L^\infty(K\times\mathbb R^k)}\le C_R n^{-\beta/d}\ ,
\end{equation}
using assumption \eqref{assmu1}.

 The fact that $ E$ has  been chosen to be an almost  minimizer of $\W$ over $\mathcal A_1$, \eqref{decrper} and the conclusions of Proposition \ref{screening} and \eqref{stst} yield
\begin{multline}\label{eqhe}
\int_{m_K^{1/d} K\times\mathbb R^k} \yg\left| \hat E_{m_K^{1/d} \eta}\right|^2-m_K \left| K\right|
\c \g(m_K^{1/d} \eta) \\  \le m_K\left|K\right|\left( \min_{\mathcal A_1 } \W + C (m_K^{1/d} \eta)^{\frac{d-s}{2}}+\frac{\ep}{4\overline{m}^{1+s/d}}+C_\eta o_{R}(1)\right),
\end{multline}where $o_R(1)$ tends to $0$ as $R \to \infty$. 
By change of scales, we have 
\begin{equation}\label{chsc1}
 \int_{K\times\mathbb R^k}\yg\left|\left(\sigma_{m_K} \hat E\right)_\eta\right|^2 = m_K^{\frac{s}{d}}\int_{m_K^{1/d}K\times\mathbb R^k}\yg\left|\hat E_{m_K^{1/d}\eta}\right|^2\ ,
\end{equation}
with the convention $s=0$  in the cases \eqref{wlog}--\eqref{wlog2d}.

The scaling of $\W$ \eqref{scalingW}--\eqref{scalinglog},  together with \eqref{esthk}, \eqref{eqhe} thus yield
\begin{multline}\label{estek}
 \int_{K\times\mathbb R}\yg|(E_K)_\eta|^2\le m_K|K|\c\g(\eta) \\+ |K|\left(\min_{{\mathcal A}_{m_K}}\W  +m_K^{1+s/d}\frac{\ep}{4\overline{m}^{s/d}}+ C m_K^{1+\frac{s}{d}+ \hal - \frac{s}{2d}} \eta^{\frac{d-s}{2}} + C_\eta o_{R}(1)\right)\ .
\end{multline}
The interior electric field  is then set to be
$E_{\mathrm{int}}= \sum_{K\in \mathcal K_n}E_K\indic_{K\times \mr}$ and it satisfies
\begin{equation}\label{eqeint}
-\div(\yg E_{\mathrm{int}}) = c_{d,s} \left( \sum_{p\in \Lambda_{\mathrm{int}}} \delta_p - \mu_V'\indic_{\Sigma_{\mathrm{int}}'}\drd\right)  \quad \text{in} \ \mathbb R^{d+k}\ ,
\end{equation}
for some discrete set $\Lambda_{\mathrm{int}}$, made of well-separated points (i.e. points whose distance to each other is bounded below by a constant independent of $R$ and $\eta$).
 Indeed, no divergence is created at the interfaces between the hyperrectangles since the normal components of $\nab h_K$ and $\hat E$ are zero. 
 \smallskip
 
\noindent
{\bf Step 2.} \emph{Placing the points in $\Sigma'  \backslash \Sigma_{\mathrm{int}}'$.}\\
Since $\p \Sigma_T'\in C^1$ as seen in the first step,  the set
\[
 \Sigma'_{\mathrm{bound}}:=\Sigma_T'\backslash \Sigma'_{\mathrm{int}}
\]
is a strip near $\p \Sigma_T'$ of volume $\le C n^{\frac{d-1}{d}}$ and width $\ge \hal C_0 R$ by  Lemma \ref{tiling}. Since $\int_{\Sigma'_T}\mu_V'\in \mathbb N$ and \eqref{inttile} holds,  $\int_{\Sigma'_{\mathrm{bound}}} \mu_V'$ is also an integer. We now need to place $\int_{\Sigma'_{\mathrm{bound}}} \mu_V'$  points in $ \Sigma'_{\mathrm{bound}}$, all separated by distances bounded below by some constant $r_0>0$ independent of $n$, $\eta$, and $R$  (up to changing $r_0$ if necessary).  Proceeding as in \cite[Section 6.3, Step 4]{ss2d}, using the fact that $\p \Sigma_T'$  is $C^1$ and making several layers, we may split $\Sigma'_{\mathrm{bound}}$ into regions $\mathcal C_i$ such that $\int_{\mathcal C_i}\mu_V'=1$ and  $\mathcal C_i$ is a set with piecewise $C^1$ boundary, containing a ball of radius $C_1 $  and contained in a ball $B_i$ of radius $C_2$, where $C_1, C_2>0$ depend only on the dimension.

We continue with examining $\Sigma' \backslash \Sigma_T'$ in the case $T\neq 0$ (i.e. if $\alpha>0$ in \eqref{assmu2}). 
For $\Sigma_t'$ as defined in Step 1, let us  set 
\begin{equation}\label{deff}
f(t)= \int_{\Sigma_t'} \muv'(x)\, dx.\end{equation}
In view of the assumption \eqref{assmu2} we have 
$$f(t)  \ge  \cc \int_0^{t}  {\mathcal H}^{d-1} (\Gamma_s) \left(\frac{s}{n^{1/d}}\right)^{\alpha} \, ds.$$
In view of \eqref{longueur} and the scaling we deduce that  
\begin{equation}\label{minorf}
f(t)  \ge  \frac{\cc|\p \Sigma|}{2(\alpha+1)}( t^{\alpha +1}) n^{\frac{d-1-\alpha}{d}} .\end{equation}
We now define inductively a sequence of $t$'s terminating at $T$: let  $t_0=0$ and  let  $t_1$  be the smallest such that 
$f(t_1)  \in  \mathbb N $ and  $$t_1 \ge  n^{\frac{\alpha}{d(\alpha+d)}}  \qquad f(t_1) \ge  |\p \Sigma'| n^{\frac{ \alpha(1-d)}{d(\alpha+d) }}= |\p \Sigma | n^{\frac{d-1}{d}  ( 1- \frac{\alpha}{\alpha+d}) }.$$
In view of \eqref{minorf} we may satisfy this condition with $t_1\le C n^{\frac{\alpha}{d(\alpha+d)}}$ where $C$ depends only on $\cc$ and $\alpha$.
We  let $N_1=f(t_1)$ and note that from the above discussion we have 
\begin{equation}\label{encadn1} |\p \Sigma' |    n^{{\frac{\alpha( 1-d)}{d(  \alpha+d)} }}\le    N_1 \le C  |\p \Sigma'|n^{ {\frac{\alpha( 1-d)}{d(  \alpha+d)} }} .\end{equation}
We write
\begin{equation}\label{encadt1}
   n^{\frac{\alpha}{d(\alpha+d)}}  \le  t_1-t_0\le  C n^{\frac{\alpha}{d(\alpha+d)}} .\end{equation}
Let $j\ge 2$. Given $t_{j-1}$ we then let $t_j$ be the smallest so that   
$$t_j \ge t_{j-1} +   n^{\frac{\alpha}{d^2}}   t_{j-1}^{-\frac{\alpha}{d}  } \qquad  f(t_j)- f(t_{j-1}) \in \mathbb N $$ and we define $N_j = f(t_j)-f(t_{j-1})$. Since $\muv'$ is bounded below  in $(\Sigma_{t_{j-1}}')^c$ by $ \cc t_{j-1}^\alpha n^{-\frac{\alpha}{d}}$ in view  of \eqref{assmu2},   we deduce that 
\begin{multline}\label{minonj}
N_j \ge \hal \cc |\p \Sigma '|   (t_j-t_{j-1})  (t_{j-1} n^{-1/d})^{\alpha}\ge   \hal \cc   |\p \Sigma|n^{\frac{d-1}{d}} t_{j-1}^\alpha n^{-\frac{\alpha}{d}} n^{\frac{\alpha}{d^2}}   t_{j-1}^{-\frac{\alpha}{d}  }\\ =  \frac{\cc}{2}|\p \Sigma| n^{\frac{d^2-d-\alpha d +\alpha}{d^2}} t_{j-1}^{ \frac{\alpha(d-1)}{d}} \ge
\frac{\cc}{2} |\p \Sigma|  n^{\frac{d-1}{d} (\alpha^2 d-\alpha+d)}
  \end{multline}
  where we have inserted that $t_{j-1}\ge t_1 \ge n^{\frac{\alpha}{d(\alpha+d)}}$.
  The exponent in the right-hand side is positive, and so we conclude that, as soon as $n$ is large enough,  we may make $N_j \ge 1$ while keeping 
$$t_j- t_{j-1}\le 2 n^{\frac{\alpha}{d^2}}   t_{j-1}^{-\frac{\alpha}{d}  }.$$
We thus write, for all $j\ge 2$
\begin{equation}
\label{encadt}
n^{\frac{\alpha}{d^2}}   t_{j-1}^{-\frac{\alpha}{d}  }\le t_j- t_{j-1}\le 2 n^{\frac{\alpha}{d^2}}   t_{j-1}^{-\frac{\alpha}{d}  }
\end{equation}
and we note that since $t_{j-1} \ge t_1 \ge n^{\frac{\alpha}{d(\alpha+d)}}$ we have 
$\frac{t_j-t_{j-1}   }{t_{j-1}} \le 2 n^{\frac{\alpha}{d^2 }} n^{(-\frac{\alpha}{d}-1)\frac{\alpha}{d(\alpha+d)}}= 2$
hence \begin{equation}\label{compara} t_{j-1}\le t_j \le 3 t_{j-1}.\end{equation} 
To bound $t_j-t_{j-1}$ from below, we write $t_j-t_{j-1}\ge n^{\frac{\alpha}{d^2}}  T^{-\alpha/d} = 2^{-\alpha/d} \cc^{1/d} \bm^{-d}$ and we may write
\begin{equation}\label{tjtj}
 2^{-\alpha/d} \cc^{1/d} \le t_j-t_{j-1}\le 2 n^{\frac{\alpha }{d(\alpha+d)}}.\end{equation}
\eqref{compara} implies, in view of \eqref{assmu2}, that 
   \begin{equation}\label{muvmin1}
  \cc  n^{-\alpha/d}t_{j-1}^\alpha\le  \muv'\le 3^\alpha \cC  n^{-\alpha/d}t_{j-1}^\alpha\quad \text{in} \ \Sigma_{t_j}'\backslash \Sigma_{t_{j-1}}'\end{equation}
   or in other words, using \eqref{encadt} that 
    \begin{equation}\label{encadmu}
  C(t_j-t_{j-1})^{-d}\le  \muv'\le    C' (t_j-t_{j-1})^{-d}\quad \text{in} \ \Sigma_{t_j}'\backslash \Sigma_{t_{j-1}}'\end{equation}
for some positive $C, C'$ depending only on the previous constants.  
   
   We may also bound $N_j$ from above, using the upper bound in \eqref{assmu2} and \eqref{longueur} :
   $$N_j= f(t_j)-f(t_{j-1}) \le 2  \cC (t_j-t_{j-1}) |\p \Sigma'|    n^{-\frac{\alpha}{d}} t_j^{\alpha}  $$ and combining this with \eqref{minonj} and \eqref{compara}
   we deduce 
   \begin{equation}\label{encadN}
  \hal \cc |\p \Sigma '|   (t_j-t_{j-1})  t_{j-1}^{\alpha}  n^{-\alpha/d} \le    N_j \le 6  \cC|\p \Sigma'| (t_j-t_{j-1})      t_{j-1}^{\alpha}n^{-{\alpha/d}}.\end{equation}

   This construction terminates at $t_J=T$, and if \eqref{encadt} is not satisfied at $j=J$ we may always merge the last two steps to have it satisfied up to a factor $2$ in the right-hand side.
   \smallskip 
   
   The next step consists, for each $1\le j \le J$, in partitioning $\Sigma_{t_j}'\backslash \Sigma_{t_{j-1}'}$ into $N_j$ regions of sidelength comparable to $t_{j}-t_{j-1}$, in such a way that $\int_{\mathcal  C_i} \muv=1$ in each region. 
   To do that, we partition $\Gamma_{t_{j-1}}$ into $N_j$ regions, and we deduce a partitioning of $\Sigma_{t_j}'\backslash \Sigma_{t_{j-1}'}$ by going along the normals to $\Gamma_{t_{j-1}}$ on the boundaries of the partition of  $\Gamma_{t_{j-1}}$.
   Given a cell $\tilde{\mathcal C}$ on $\Gamma_{t_{j-1}}$ we denote by ${\mathcal C}$ the corresponding cell in  $\Sigma_{t_j}'\backslash \Sigma_{t_{j-1}'}$ obtained this way. We claim that we can construct $N_j$ such cells $\tilde{\mathcal{C}_i}$ such that 
   $\int_{{\mathcal C_i}} \muv =1 $ for each $i$ and the sidelengths of the cells are comparable to $t_j-t_{j-1}$. 
   First, we note that if $j\ge 2$, from \eqref{longueur}, \eqref{encadN} and \eqref{encadt}  we have that $|\Gamma_{t_{j-1}}|/N_j$ is bounded above and below by constants times $(t_j-t_{j-1})^{d-1} $. If $j=1$, the same holds using instead \eqref{encadt1} and \eqref{encadn1}.
 We may thus  subdivide $\Gamma_{t_{j-1}}$ iteratively, the same way we did for Lemma \ref{tiling}, or above as in \cite[Section 7.3, Step 4]{ss2d}, except we are on a curved hypersurface instead of a flat one. In order to ensure that the result can be achieved, we need a lower and upper bound on the density of $\muv'$ by a constant times $(t_j-t_{j-1})^{-d}$, provided precisely by \eqref{encadmu} in the case $j\ge 2$. In the case $j=1$, then we replace the lower bound by a bound on the integrated density of $\muv'$ along the normal to $\p \Gamma_0= \p \Sigma'$: for $x \in \p \Sigma'$ and $\vec{\nu}$ the inner normal at $x$, we have 
  $$\int_{0}^{t_1} \muv'(x+ s\vec{\nu})\, ds \ge \cc \int_0^{n^{\frac{\alpha}{d(\alpha+d)} }} \left(\frac{s}{n^{1/d}} \right)^\alpha\, ds = \frac{\cc}{1+\alpha}  n^{-\alpha/d} n^{\frac{\alpha(\alpha+1) }{d(\alpha+d)} }= n^{\frac{\alpha(1-d)}{d(\alpha+1)}}$$
  which is comparable to $t_1^{d-1}$ as desired.

   Thus we may indeed partition 
   $\Sigma_{t_j}'\backslash \Sigma_{t_{j-1}'}$ into $N_j$ cells ${\mathcal C_i}$, which are contained in a ball $B_i$ of radius $C_2 (t_{j}-t_{j-1}) $ and contain a ball of radius $C_1( t_{j}-t_{j-1})$, for some $C_1, C_2$ depending only on the above constants $\cc, \cC, \alpha,d$, and in which $\int_{{\mathcal C_i}} \muv' =1$. We may add the cells obtained this way for $j$ ranging from $0$ to $J$, to the ones obtained at the beginning of the  step in $\Sigma_{\mathrm{bound}}'$. We call $r_i$ their diameter scale, i.e. $r_i=1$ in the case of the cells obtained in $\Sigma_{\mathrm{bound}}'$, and $r_i= t_{j}-t_{j-1}$ for the cells in $\Sigma_{t_j}'\backslash \Sigma_{t_{j-1}}'$.  From \eqref{tjtj} we have $0<c\le r_i \le 2 n^{\frac{\alpha}{d(\alpha+d)}}$.
   We let $p_i$ be the center of each inner ball included in $\mathcal{C}_i$. This way $p_i$ is at distance $\ge cr_i$ from $\p \mathcal C_i$, and thus all the points are separated by distances independent of $R, \eta$ and $n$.
 \smallskip 
 
 \noindent
 {\bf Step 3.}  \emph{Completing the construction.}\\  
 We then let $v_i$ solve 
\begin{equation}
\left\{\begin{array}{ll}
-\div(\yg\nab v_i)= c_{d,s} \left(  \delta_{p_i} - \muv' \indic_{{\mathcal C_i}}\drd \right)  &  \quad \text{ in } B_i\times[-r_i,r_i]^k\ ,\\[3mm]
\nab v_i \cdot \vec {\nu}= 0 & \text{ on } \p B_i\times[-r_i,r_i]^k,\end{array}\right.\end{equation}  and we  set
\[
E_{\mathrm{bound}}:= \sum_i \indic_{\mathcal C_i} \nab v_i\, .
\]
This way, we globally  have 
\begin{equation}\label{eqebound}-\div (\yg E_{\mathrm{bound} } )=  \c \left( \sum_{i } \delta_{p_i} - \mu_V'\indic_{\mathcal C_i}\right)\quad \text{in } \mr^{d+k}\ .
\end{equation}

We now evaluate the energy generated by these vector fields.  For each cell $i$, we write $v_i =  h_i+u_i$ where 
$$-\div (\yg \nab h_i)= \c \left( \delta_{p_i} -\left(\dashint_{\mathcal C_i} \muv'\drd\right) \indic_{\mathcal C_i} \drd \right) \quad \text{in} \ B_i \times [-r_i,r_i]^k$$
and 
$$-\div (\yg \nab u_i) = \c\left(    \left(\dashint_{\mathcal C_i} \muv'\drd\right)- \muv'\right)  \indic_{\mathcal C_i}  \drd  \quad \text{in} \ B_i \times [-r_i,r_i]^k$$
with the same  zero Neumann  boundary data. The energy of $u_i$ is estimated 
%via Lemma \ref{lem58} 
using assumption \eqref{assmu1} and $r_i \le  2 n^{\frac{\alpha}{d(\alpha+d)}}$:
$$\int_{B_i\times [-r_i,r_i]^k}  \yg |\nab u_i|^2 \le  C r_i^{2d-s} r_i^{2\beta} n^{-2\beta/d}\le n^{-\frac{2\beta}{d}  +\frac{\alpha(2d-s+2\beta)}{d(\alpha +d)}  }= n^{\frac{  -2\beta d+ 2\alpha d-s\alpha   }{d(\alpha +d)}}\le 1 
$$ where for the last inequality, we used \eqref{condab}.
The energy of $h_i$ is estimated by scaling as in \eqref{chsc1} to be 
$$\int_{B_i \times [-r_i,r_i]^k} \yg |(\nab h_i)_\eta |^2\le \c \g(\eta)+  \begin{cases}C m_i^{s/d} & \text{in case} \ \eqref{kernel}\\
- \frac{\c}{d} \log m_i +C & \text{in cases} \ \eqref{wlog}-\eqref{wlog2d}
\end{cases}$$ 
Combining these estimates and using \eqref{scalingW}--\eqref{scalinglog}, we find that in all cases 
$$\int_{\mathcal C_i} \yg |(\nab v_i)_\eta |^2 \le C_\eta+ \frac{\min_{\mathcal{A}_{m_i} }\W}{m_i} = C_\eta + |\mathcal C_i| \min_{\mathcal{A}_{m_i}} ,$$
since $\int_{\mathcal C_i} \muv'=1$ and $m_i = \dashint_{\mathcal C_i} \muv'= \frac{1}{|\mathcal C_i|}.$

The total contribution of  the cells in $\Sigma'_{\mathrm{bound}} $ is bounded by the number of those cells (since there $m_i$ is bounded below), itself proportional to the volume of that strip i.e. $n^{\frac{d-1}{d}}$. 
The number of cells in $\Sigma'\backslash \Sigma_T'$ is on the other hand bounded by 
\begin{multline*}\sum_{j=1}^J N_j \le 
 C|\p \Sigma'| n^{-\alpha/d}  \sum_{j=1}^J (t_j-t_{j-1}) t_j^\alpha \le C |\p \Sigma'| n^{-\alpha/d}  \int_0^T t^\alpha \, dt \\
 = C |\p \Sigma| n^{\frac{d-1-\alpha}{d}} T^{1+\alpha} \le C \bm^{1+1/\alpha}n\end{multline*}
where we have used the properties  \eqref{compara} and inserted the upper bound $T \le C \bm^{1/\alpha} n^{1/d}$. 
As soon as $n$ is large enough, we may thus bound the total contribution of all the cells $\mathcal C_i$ as follows
\begin{equation}\label{energybord}
\int\yg |\Phi_\eta(E_{\mathrm{bound}})  |^2 \le C_\eta n \bm^{1+1/\alpha} + \sum_i  |\mathcal C_i| \min_{\mathcal{A}_{m_i}}\W .
\end{equation}
Finally, setting
\[
E =E_{\mathrm{bound}}+ E_{\mathrm{int}} \quad  \mbox{ in } \R^{d+k}
\]
and  letting
\[
 \Lambda= \Lambda_{\mathrm{int}} \cup \bigcup_i \{p_i\} ,
\]
then $\Lambda$ is made of $\int_{\Sigma'} \muv'=n$ points, which are all well separated, and $E$ satisfies
\begin{equation}\label{eglob}
-\div (\yg E) = c_{d,s} \left( \sum_{p \in \Lambda} \delta_p- \mu_V'\drd\right) \quad \text{in} \ \mr^{d+k}\ . \end{equation}
Combining the estimates \eqref{estek} and \eqref{energybord}, we have 
\begin{multline}\label{energytot}
\int_{\R^{d+k}} \yg |E_\eta|^2
\\  \le \sum_{K\in \mathcal K_n} m_K|K|(\c\g(\eta)+ C_\eta o_{R}(1))+ \sum_{K\in\mathcal K_n}  |K|\left(\min_{\mathcal A_{m_K}} \W+ m_K^{1+s/d} \frac{\ep}{ 4\overline{m}^{s/d}}+ C m_K^{1+ \frac{s}{d}+ \hal - \frac{s}{2d}} \eta^{\frac{d-s}{2}}\right)\\
 +   C_\eta n \bm^{1+1/\alpha} + \sum_i  |\mathcal C_i| \min_{\mathcal{A}_{m_i}}\W +  C_\eta n \bm^{1+1/\alpha}. 
\end{multline}
First we note that 
\begin{equation}\label{boundch}
\sum_{K\in \mathcal K_n} m_K |K| = \int_{\Sigma'_{\mathrm{int}}  }\mu_V' \le n. 
\end{equation}
Second, in view of the continuity  of $\muv$  and that of  $m\mapsto \min_{{\mathcal{A}}_m}  \W$ (see \eqref{scalingW}--\eqref{scalinglog}) 
we may  recognize a Riemann sum and write (using  the fact that $m_K \le \overline{m}$ by \eqref{assmu2}) 
\begin{multline}\label{riemsumappr}
\int_{\R^{d+k}} \yg |E_\eta|^2 \le n\c \g(\eta)+ \int_{\Sigma'} \min_{\mathcal A_{\muv'(x)}}\W \, dx  + \frac{ n \ep}{4}  + C_\eta n \bm^{1+1/\alpha}
+C_\eta o_{R}(n) + C {\overline m}^{\hal + \frac{s}{2d}} \eta^{\frac{d-s}{2}} \\ \le    n \c \g(\eta)+    n \left( \int_{\Sigma }\min_{\mathcal A_{\muv(x)}}\W \, dx  + \frac{\ep}{4}\right) + C_\eta n \bm^{1+1/\alpha}+C_\eta o_{R}(n) + C {\overline m}^{\hal + \frac{s}{2d}} \eta^{\frac{d-s}{2}} . 
\end{multline}

\noindent
{\bf Step 4.} \emph{Conclusion.}\\
We may now  define our test configuration as
\[
\xbf = \{x_i= n^{-1/d} x_i'\}_{i=1}^n \mbox{ where } \Lambda = (x_1', \dots, x_n')\ 
\] and let $h_{n}'$ be the associated potential as in \eqref{rescalh}. 
Note that the configuration depends on $R $ and $\eta$, and that the points $x_i'$ are separated by distances independent of $R, \eta$ and $n$. For $\eta$ small enough, we are thus in the equality case in Lemma \ref{prodecr}, and using 
 the splitting formula \eqref{propsplit}, we deduce that 
\begin{multline}\label{splitrecall}
\limsup_{n\to \infty}n^{-1-s/d} \left(H_n(x_1,\dots,x_n)-n^2 \I(\muv) \right)\\ \le  \limsup_{n\to \infty} \left(\frac{1}{n\c}\int_{\mathbb R^{d+k}}\yg|\nab h_{n,\eta}'|^2 - \g(\eta)\right) + C\eta^{\frac{d-s}{2}}\ ,
\end{multline}in case \eqref{kernel}, respectively
\begin{multline}\label{splitrecall2}
\limsup_{n\to \infty}n^{-1} \left(H_n(x_1,\dots,x_n)-n^2 \I(\muv)+ \frac{n}{d} \log n  \right) \\ \le  \limsup_{n\to \infty} \left(\frac{1}{n\c}\int_{\mathbb R^{d+k}}\yg|\nab h_{n,\eta}'|^2 - \g(\eta)\right) + C\eta^{\frac{d}{2}}\ ,
\end{multline}in the cases \eqref{wlog}--\eqref{wlog2d}, 
where we used that all the points are in $\Sigma$ where the function $\zeta$ vanishes. We now note that projecting $E$ onto gradients decreases the energy: 
\begin{align}\nonumber
\int_{\mr^{d+k}} \yg|E_\eta|^2 &= \int_{\mr^{d+k}} \yg|\nab h_{n,\eta}'|^2 + \int_{\mr^{d+k}}\yg |E_\eta - \nab h_{n,\eta}'|^2 + 2 \int_{\mr^{d+k}}\yg \nab h_{n,\eta}' \cdot \left( E_\eta - \nab h_{n,\eta}' \right)\\[3mm]\label{eetagtrheta}
&\ge  \int_{\mr^{d+k}} \yg|\nab h_{n,\eta}'|^2  - 2 \int_{\mr^{d+k}} h_{n,\eta}' \,\div(\yg( E_\eta - \nab h_{n,\eta}') )= \int_{\mr^{d+k}} \yg|\nab h_{n,\eta}'|^2 \end{align}
where we have used that  $-\div (\yg(E_\eta - \nab h_{n,\eta}'))=0$ by definition of $h_{n, \eta}'$ and \eqref{eglob}. Combining \eqref{splitrecall}, \eqref{splitrecall2}, \eqref{eetagtrheta} and \eqref{riemsumappr}, 
the proof is concluded by  taking successively $n\to \infty$, then $R$ large enough, $\bm$ small enough  and  $\eta$ small enough (and changing the configuration of points accordingly) to conclude that we can get 
\[
\limsup_{n\to \infty} n^{-1-s/d}\left(  H_n(x_1, \dots, x_n) -n^2  \I(\mu_V)\right)   \le \left(\frac{1}{\c} \int_{\Sigma} \min_{{\mathcal A}_{\mu_V(x)} }  \W \, dx+\frac{\ep}{2}\right)\ ,
\] and the analogue in the cases \eqref{wlog}--\eqref{wlog2d}.
Finally, inserting \eqref{scalingW}--\eqref{scalinglog} and the definition of $\widetilde{\W}$ we obtain the upper bound result, together with the fact that indeed $\min \widetilde{\W}=  \frac{1}{\c} \int_{\Sigma} \min_{{\mathcal A}_{\mu_V(x)} }  \W $ (by comparing the upper bound with  the lower bound result Proposition \ref{proergodic}).   Since the points are well-separated, we may also move each of them by a distance $r$ small enough and keep  the same estimate up to an  additional error $\ep/2$, in a neighborhood $A_n$ of the configuration. The details are identical to \cite{rs}.

This concludes the proof of Proposition \ref{upperbound}.
\end{proof}

 Theorem \ref{th2} then follows by comparison with Proposition \ref{proergodic}.
Theorem \ref{stat} follows from all the previous results exactly as in \cite{rs}. The details are left to the reader.

\bigskip

\noindent Mircea Petrache:
\newline {\tt mircea.petrache@upmc.fr} \newline
Sylvia Serfaty:
\newline {\tt serfaty@ann.jussieu.fr}

\end{document}